\pdfoutput=1
\documentclass[11pt, a4paper]{article}
\usepackage{amsmath, amssymb, amsfonts, amsthm}
\usepackage{bm}
\usepackage{graphicx}
\usepackage{algorithm}
\usepackage{algpseudocode}
\usepackage[colorlinks=true, linkcolor=blue, citecolor=blue]{hyperref}
\usepackage{geometry}
\usepackage{tikz}
\usetikzlibrary{shapes.geometric, arrows.meta, positioning, fit, backgrounds, calc, decorations.pathmorphing}
\usepackage{xcolor}
\usepackage{booktabs}
\usepackage{float}

\theoremstyle{plain}
\newtheorem{theorem}{Theorem}
\newtheorem{proposition}{Proposition}
\newtheorem{lemma}{Lemma}

\theoremstyle{definition}

\newtheorem{remark}{Remark}

\newcommand{\tens}[1]{\bm{#1}}

\newcommand{\trace}{\operatorname{tr}}

\newcommand{\sgn}{\operatorname{sgn}}
\newcommand{\softplus}{\operatorname{softplus}}
\newcommand{\subdiff}{\partial}

\newcommand{\macaulay}[1]{\langle #1 \rangle_+}

\title{\textbf{Discovering Thermodynamically Admissible Dissipation Potentials via Grammar-Based Symbolic Regression}}
\author{F. Califano and J. Ciambella}
\date{\today}

\begin{document}
\maketitle

\begin{abstract}
\noindent Constitutive laws for inelastic materials must satisfy strict thermodynamic admissibility requirements, yet current data-driven approaches sacrifice interpretability, even when formal guarantees are provided by physics-encoded architectures. We propose a symbolic regression framework for the data-driven discovery of dissipation potentials governing the evolution of internal variables within the Generalized Standard Materials (GSM) formalism. Starting from the Clausius--Duhem inequality, we enforce the thermodynamic requirements, convexity and non-negativity, that the dual dissipation potential must satisfy to guarantee non-negative mechanical dissipation. These requirements are formulated in the general subdifferential setting, encompassing rate-dependent (viscoelastic) and viscoplastic dissipative mechanisms, including potentials with genuine elastic domains, within a unified framework. Candidate potentials are generated by a composition-extended convexity-preserving grammar that guarantees thermodynamic admissibility \emph{by construction}. The framework is validated on synthetic datasets spanning Newtonian, power-law, and Bingham viscoplastic ground truths under process and measurement noise, and on experimental oscillatory shear measurements of a synthetic elastomer across multiple strain amplitudes and frequencies, where the discovered potentials reproduce the amplitude-dependent softening of the dynamic moduli and outperform a calibrated linear Zener baseline.
\end{abstract}


\section{Introduction}

Constitutive modeling of inelastic materials, encompassing viscoelastic polymers, viscoplastic metals, filled elastomers, and biological tissues, requires capturing complex nonlinear, history-dependent behaviors while rigorously respecting the laws of thermodynamics. Classical rheological assemblages of springs and dashpots~\cite{ciambella2009abaqus, ciambella2021anisotropic} provide physically interpretable descriptions in the linear regime, but their expressivity is insufficient for the nonlinear phenomena routinely observed in experiments: the Payne effect in filled elastomers~\cite{payne1962dynamic,payne1963dynamic}, rate-dependent yielding in metallic alloys~\cite{kabliman2021application, batra1990effect}, stress-softening in soft tissues~\cite{holzapfel2001biomechanics}, and complex viscoplastic flow~\cite{perzyna1966fundamental}. Historically, constitutive model development has relied on expert intuition and \emph{ad hoc} phenomenological assumptions, frequently yielding formulations that lack generalizability beyond the calibration regime.

Data-driven approaches have sought to overcome these limitations, and can be broadly grouped according to how physical principles are incorporated~\cite{fuhg2024review}. 
\emph{Black-box} neural networks approximate arbitrary constitutive maps directly from data, but disregard the underlying thermodynamics and offer no guarantee of admissibility outside the training range. 
\emph{Physics-Informed Neural Networks} (PINNs) enforce balance laws and dissipation inequalities \emph{in a weak sense}, as soft penalty terms added to the training loss; while effective in integrating sparse data with physical knowledge, they require careful tuning of the loss weights and offer no formal guarantee that the learned model satisfies the imposed constraints. 
\emph{Physics-Encoded Neural Networks}, in contrast, embed convexity and thermodynamic admissibility \emph{directly into the network structure}. 
Yet, regardless of whether physical constraints are imposed weakly or by construction, all these neural-network-based approaches share a common limitation: the learned representations are opaque, extrapolation is uncertain, and integration into finite element codes requires specialized subroutines that complicate industrial deployment~\cite{ciambella2009abaqus}.

\emph{Symbolic Regression} (SR) addresses this tension by searching the space of mathematical expressions to discover compact, human-readable formulas~\cite{makke2024review}. The dominant paradigm, originating with genetic programming~\cite{koza1992genetic} and famously demonstrated for the rediscovery of physical laws from experimental data~\cite{schmidt2009distilling}, evolves populations of candidate expressions under accuracy--complexity trade-offs, and is embodied today in high-performance implementations such as PySR~\cite{cranmer2023pysr}. In contrast to black-box methods, SR produces \emph{white-box models}: closed-form expressions amenable to asymptotic analysis and directly deployable in numerical simulation software~\cite{abdusalamov2023automatic, kabliman2021application}.

Within constitutive modeling, SR has been successfully deployed for hyperelastic energy discovery from invariant-based or stretch-based parameterizations~\cite{abdusalamov2023automatic, hou2024automated, kissas2024language}; for implicit yield surface identification in elastoplasticity~\cite{bomarito2021development, garbrecht2023pGPSR}; for empirical flow laws of age-hardenable aluminium alloys and high-chromium martensitic steels~\cite{kabliman2026identification}; and for geomaterial constitutive laws~\cite{zhu2025physics}.
Recent hybrid pipelines first train a neural network model to fit the data, and afterward derive a human-readable formula from that trained model~\cite{bahmani2023elastoplasticity, bahmani2024physics}. However, all these SR efforts have predominantly targeted rate-independent constitutive functions, strain-energy densities and yield functions, that define \emph{equilibrium} or \emph{threshold} states. None addresses the identification of inelastic materials that exhibit history-dependent behavior governed by \emph{dissipation potentials}, which prescribe the evolution of internal variables via a gradient-flow structure within the Generalized Standard Materials (GSM) formalism~\cite{halphen1975materiaux, lemaitre2000mechanics}. The GSM framework encompasses rate-dependent (viscoelastic) and viscoplastic dissipative mechanisms within a unified thermodynamic structure, wherein constitutive behavior derives entirely from two scalar potentials: the Helmholtz free energy $\psi$ and the dissipation potential $\varphi$ (or its Legendre--Fenchel dual $\varphi^*$). Discovering $\varphi^*$ from data presents a fundamentally different challenge compared to energy identification: the regression must identify not merely an energy landscape, but a \emph{dynamical evolution law} coupling thermodynamic driving forces to internal variable rates, subject to the strict requirement that $\varphi^*$ be convex and non-negative to guarantee non-negative dissipation. Recent works have begun to integrate the GSM framework with constitutive neural networks for finite-strain viscoelasticity~\cite{rosenkranz2024viscoelasticity, asad2023mechanics, holthusen2026inelastic, tac2023benchmarking, abdolazizi2024vcanns,califano2026enhancing, kalina2026physics} and plasticity~\cite{flaschel2025convex, boes2026accounting, boes2026plasticity, jadoon2025plasticity}, but these approaches inherit the interpretability limitations discussed above. A first step in this direction has been recently taken by~\cite{ji2026ickan}, who introduced Inelastic Constitutive Kolmogorov--Arnold Networks (iCKANs): leveraging the edge-wise learnable univariate activations of KANs~\cite{liu2024kan} and extending input-convex KAN architectures previously developed for hyperelasticity~\cite{thakolkaran2025kan, califano2026enforcing}, an input-convex KAN is combined with a post-training symbolification step over a restricted library of convex, non-decreasing functions, yielding closed-form elastic and inelastic potentials that are thermodynamically admissible by construction, validated on viscoelastic polymers. To the best of our knowledge, however, no grammar-based symbolic regression framework for dissipation potentials has been proposed in the literature.

The present work extends grammar-based SR to the identification of dissipation potentials for inelastic materials. Assuming the elastic energy density $\psi$ is known \emph{a priori}, from physical considerations or independent identification via equilibrium experiments~\cite{abdusalamov2023automatic}, we seek the dual dissipation potential $\varphi^*$ governing the evolution of internal variables. Thermodynamic constraints are encoded directly into context-free grammar production rules, extending the approach of~\cite{kissas2024language} to dissipative systems, so that convexity and non-negativity are guaranteed \emph{by construction}. The main contributions of this paper are:

\begin{enumerate}
    \item A composition-extended convexity-preserving grammar $\mathcal{G}_{\mathrm{cvx}}^{\mathrm{comp}}$ that operationalizes the chain rule for convex functions, guaranteeing thermodynamic consistency for all generated expressions within a unified framework encompassing viscoelastic and viscoplastic dissipation, including potentials with genuine elastic domains.
    
    \item A targeted identification architecture that decouples dissipative discovery from elastic identification.
    
    \item Convexity-preserving genetic operators (term-level crossover, class-preserving mutation) with formal proofs that grammar membership is an evolutionary invariant.
    
    \item Validation on synthetic inelastic datasets with process and measurement noise, demonstrating exact symbolic recovery rates exceeding 90\% under moderate noise conditions, and on oscillatory shear measurements of a synthetic elastomer performed at the Laboratory of Materials and Structures, Department of Structural and Geotechnical Engineering, Sapienza University of Rome across multiple strain amplitudes and frequencies.
\end{enumerate}

\section{Thermodynamic Framework}\label{sec:thermo}

We review the essential elements of the Generalized Standard Materials (GSM) formalism \cite{halphen1975materiaux, lemaitre2000mechanics}, with particular emphasis on the convexity and non-negativity requirements of the dissipation potential that will later constrain the symbolic search space. The presentation adopts the subdifferential setting from the outset, which naturally accommodates both smooth (viscoelastic) and non-smooth (viscoplastic) potentials within a single framework.

In the following vectors and second-order tensors are denoted by bold italic letters (e.g.\ $\tens{A}$, $\tens{\varepsilon}$, $\tens{\sigma}$). The symbol $:$ denotes the inner product of second-order tensors, $\tens{A}:\tens{B} = A_{ij}B_{ij}$ (Einstein summation implied), and reduces to the standard dot product for vectors. The Euclidean (Frobenius) norm of a tensor $\tens{A}$ is written $|\tens{A}| = \sqrt{\tens{A}:\tens{A}}$. The identity tensor is $\tens{I}$, $\trace(\cdot)$ denotes the trace, and a prime $(\cdot)'$ denotes the deviatoric part of a tensor. The Macaulay bracket $\macaulay{x} = \max(x,0)$ extracts the positive part of a scalar. The symbol $\subdiff$ denotes the convex subdifferential, and $\nabla$ the gradient with respect to the indicated argument.

The thermodynamic state of the material is characterized by the observable strain tensor $\tens{\varepsilon}$ and a set of internal variables $\tens{z} = (z_1, \ldots, z_N)$ that may represent dissipative mechanisms such as inelastic strains and microstructural rearrangements. The Helmholtz free energy density takes the form
\begin{equation}\label{eq:free-energy}
    \psi = \psi(\tens{\varepsilon}, \tens{z}).
\end{equation}

The First and Second Laws of Thermodynamics, under isothermal conditions, imply the Clausius--Duhem inequality:
\begin{equation}\label{eq:clausius-duhem}
    \mathcal{D} = \tens{\sigma} : \dot{\tens{\varepsilon}} - \dot{\psi} \geq 0,
\end{equation}
where $\tens{\sigma}$ is the Cauchy stress and $\mathcal{D}$ denotes the intrinsic dissipation rate per unit volume. Applying the chain rule to the time derivative of the free energy yields
\begin{equation}\label{eq:psi-rate}
    \dot{\psi} = \frac{\partial \psi}{\partial \tens{\varepsilon}} : \dot{\tens{\varepsilon}} + \frac{\partial \psi}{\partial \tens{z}} : \dot{\tens{z}}.
\end{equation}
Substituting \eqref{eq:psi-rate} into \eqref{eq:clausius-duhem}, we obtain
\begin{equation}\label{eq:dissipation-expanded}
    \mathcal{D} = \left( \tens{\sigma} - \frac{\partial \psi}{\partial \tens{\varepsilon}} \right) : \dot{\tens{\varepsilon}} - \frac{\partial \psi}{\partial \tens{z}} : \dot{\tens{z}} \geq 0.
\end{equation}

For the inequality~\eqref{eq:dissipation-expanded} to hold for arbitrary strain rates $\dot{\tens{\varepsilon}}$, including purely reversible elastic processes, the coefficient of $\dot{\tens{\varepsilon}}$ must vanish identically, yielding the \emph{state law}:
\begin{equation}\label{eq:state-law-stress}
    \tens{\sigma} = \frac{\partial \psi}{\partial \tens{\varepsilon}}.
\end{equation}
Defining the \emph{thermodynamic driving force} conjugate to the internal variables as
\begin{equation}\label{eq:driving-force-def}
    \tens{A} := -\frac{\partial \psi}{\partial \tens{z}},
\end{equation}
the Clausius--Duhem inequality reduces to:
\begin{equation}\label{eq:residual-dissipation}
    \mathcal{D} = \tens{A} : \dot{\tens{z}} \geq 0.
\end{equation}

which constrains admissible evolution laws, but does not uniquely determine one: any map $\tens{f}: \mathbb{R}^N \to \mathbb{R}^N$ satisfying $\tens{A} : \tens{f}(\tens{A}) \geq 0$ is thermodynamically admissible. 

To close the constitutive framework, we seek \emph{structural conditions} on the evolution law that guarantee~\eqref{eq:residual-dissipation} for every admissible process. A natural candidate is a gradient-flow structure, in which the evolution derives from a scalar potential \cite{halphen1975materiaux, lemaitre2000mechanics, germain1983continuum}. We therefore postulate the existence of a \emph{dissipation potential} $\varphi: \mathbb{R}^N \to \mathbb{R} \cup \{+\infty\}$, proper, lower semicontinuous, and convex, with $\varphi(\tens{0}) = 0$ and $\varphi \geq 0$, governing the internal-variable evolution through the \emph{normality rule} (or Biot equation):
\begin{equation}\label{eq:normality}
    \tens{A} \in \subdiff_{\dot{\tens{z}}}\,\varphi(\dot{\tens{z}}),
\end{equation}
where $\subdiff_{\dot{\tens{z}}}\,\varphi$ denotes the convex subdifferential of $\varphi$ with respect to $\dot{\tens{z}}$:
\begin{equation}\label{eq:subdiff-def}
    \subdiff_{\dot{\tens{z}}}\,\varphi(\dot{\tens{z}}) := \bigl\{\tens{A} \in \mathbb{R}^N :\; \varphi(\tens{w}) \geq \varphi(\dot{\tens{z}}) + \tens{A} : (\tens{w} - \dot{\tens{z}}) \;\;\forall\,\tens{w} \in \mathbb{R}^N\bigr\}.
\end{equation}
When $\varphi$ is differentiable at $\dot{\tens{z}}$, the subdifferential reduces to the singleton $\subdiff\varphi(\dot{\tens{z}}) = \{\nabla_{\dot{\tens{z}}}\varphi(\dot{\tens{z}})\}$ and \eqref{eq:normality} recovers the classical gradient relation $\tens{A} = \nabla_{\dot{\tens{z}}}\varphi(\dot{\tens{z}})$. On the other hand, the subdifferential formulation can accommodate non-smooth potentials, such as the indicator function of the elastic domain in rate-independent plasticity or the Bingham potential in viscoplasticity, within the same framework.

Geometrically, the normality rule~\eqref{eq:normality} states that $\tens{A}$ belongs to the normal cone of the level sets of $\varphi$ at the current flow rate~$\dot{\tens{z}}$, and implies the Principle of Maximum Dissipation: for a given driving force $\tens{A}$, the actual internal variable rate $\dot{\tens{z}}$ maximizes the dissipated power $\tens{A}:\dot{\tens{z}} - \varphi(\dot{\tens{z}})$ among all admissible rates.

\medskip

It is noted that, the normality rule~\eqref{eq:normality} is expressed in terms of $\dot{\tens{z}}$: given a flow rate, it returns the conjugate force $\tens{A}$. In practice, however, the driving force $\tens{A}$ is the quantity computed from the free energy and the current state, and the task is to determine the resulting flow rate $\dot{\tens{z}}$, that is, to invert the subdifferential relation. Passing to the convex conjugate achieves this inversion directly and provides an explicit evolution law for $\dot{\tens{z}}$ in terms of $\tens{A}$, which is the form required for numerical time integration.

The \emph{dual dissipation potential} $\varphi^*: \mathbb{R}^N \to \mathbb{R} \cup \{+\infty\}$ is defined via the Legendre--Fenchel transform:
\begin{equation}\label{eq:legendre-fenchel}
    \varphi^*(\tens{A}) = \sup_{\dot{\tens{z}} \in \mathbb{R}^N} \left\{ \tens{A} : \dot{\tens{z}} - \varphi(\dot{\tens{z}}) \right\}.
\end{equation}
By standard properties of the convex conjugate, $\varphi^*$ is convex, non-negative, and satisfies $\varphi^*(\tens{0}) = 0$ and by biconjugation theorem, $(\varphi^*)^* = \varphi$ for closed proper convex functions, the normality rule~\eqref{eq:normality} can equivalently be expressed in \emph{dual form}:
\begin{equation}\label{eq:evolution}
    \dot{\tens{z}} \in \subdiff_{\tens{A}}\,\varphi^*(\tens{A}).
\end{equation}
When $\varphi^*$ is differentiable at $\tens{A}$, this reduces to the single-valued evolution law $\dot{\tens{z}} = \nabla_{\tens{A}}\,\varphi^*(\tens{A})$.

\medskip

With the duality framework in hand, the following proposition establishes that two simple, checkable properties of $\varphi^*$ are \emph{sufficient} for the Second Law. Crucially, the result requires \emph{no smoothness assumption} on $\varphi^*$.

\begin{proposition}[Sufficient conditions for non-negative dissipation]\label{prop:dissipation}
Let $\varphi^*: \mathbb{R}^N \to \mathbb{R} \cup \{+\infty\}$ be a proper convex function satisfying:
\begin{enumerate}
    \item[\textnormal{(P1)}] \textbf{Convexity:}\; $\varphi^*$ is convex;
    \item[\textnormal{(P2)}] \textbf{Non-negativity and normalization:}\; $\varphi^*(\tens{A}) \geq 0$ for all $\tens{A}$, with $\varphi^*(\tens{0}) = 0$.
\end{enumerate}
If the evolution law satisfies $\dot{\tens{z}} \in \subdiff_{\tens{A}}\,\varphi^*(\tens{A})$, then the dissipation is non-negative:
\begin{equation}\label{eq:dissipation-bound}
    \mathcal{D} = \tens{A} : \dot{\tens{z}} \geq \varphi^*(\tens{A}) \geq 0.
\end{equation}
\end{proposition}

\noindent\textit{Proof.}\; Applying the subdifferential definition~\eqref{eq:subdiff-def} to the dual potential $\varphi^*$ at the point $\tens{A}$ (with $\dot{\tens{z}}$ playing the role of the subgradient and $\tens{A}$ that of the point, dually to~\eqref{eq:subdiff-def}), the inclusion $\dot{\tens{z}} \in \subdiff_{\tens{A}}\varphi^*(\tens{A})$ means that for every test point $\tens{B} \in \mathbb{R}^N$:
\[
    \varphi^*(\tens{B}) \;\geq\; \varphi^*(\tens{A}) + \dot{\tens{z}} : (\tens{B} - \tens{A}).
\]
Choosing in particular $\tens{B} = \tens{0}$ and using $\varphi^*(\tens{0}) = 0$ from (P2):
\[
    0 \geq \varphi^*(\tens{A}) + \dot{\tens{z}} : (\tens{0} - \tens{A}) = \varphi^*(\tens{A}) - \tens{A} : \dot{\tens{z}},
\]
which rearranges to $\mathcal{D} = \tens{A} : \dot{\tens{z}} \geq \varphi^*(\tens{A})$. The second inequality, $\varphi^*(\tens{A}) \geq 0$, follows directly from~(P2). \qed

\medskip

Proposition~\ref{prop:dissipation} establishes the central design principle of this work: \emph{any convex function $\varphi^*(\tens{A})$ satisfying $\varphi^*(\tens{A}) \geq 0$ and $\varphi^*(\tens{0})=0$ generates, through the evolution law $\dot{\tens{z}} \in \subdiff_{\tens{A}}\,\varphi^*(\tens{A})$, a kinetic relation that is automatically consistent with the Second Law of Thermodynamics.} No differentiability assumption is required: the result applies equally to smooth potentials (yielding classical viscoelasticity) and to non-smooth potentials (accommodating viscoplasticity). Properties~(P1) and~(P2) will therefore serve as the \emph{hard constraints} encoded into the symbolic search space (Section~\ref{sec:convexity}).

\begin{remark}[Classification of dissipative behavior]\label{rem:classification}
The regularity of the dual dissipation potential $\varphi^*$ determines the type of dissipative response:
\begin{enumerate}
    \item[\textnormal{(i)}] \textbf{Viscoelasticity} ($\varphi^*$ smooth and strictly convex). The evolution law $\dot{\tens{z}} = \nabla_{\tens{A}}\varphi^*(\tens{A})$ is a single-valued, smooth map. Example: $\varphi^*(\tens{A}) = \vert\tens{A}\vert^2/(2\eta)$ (Newtonian viscosity). The material exhibits no elastic domain: any nonzero $\tens{A}$ produces immediate flow. 
    \item[\textnormal{(ii)}] \textbf{Viscoplasticity} ($\varphi^*$ finite-valued, convex, but not everywhere differentiable). Non-smooth points typically arise where $|\tens{A}|$ equals a yield threshold~$\sigma_Y$, introducing an elastic domain: for $|\tens{A}| \leq \sigma_Y$, $\dot{\tens{z}} = 0$; above it, $\dot{\tens{z}}$ depends on the rate of loading. Example: $\varphi^*(\tens{A}) = \macaulay{|\tens{A}|-\sigma_Y}^2/(2\eta)$ (Bingham viscoplastic), where $\macaulay{\cdot} = \max(\cdot,0)$ denotes the Macaulay bracket. For $|\tens{A}| > \sigma_Y$, the flow rule reads $\dot{\tens{z}} = (|\tens{A}|-\sigma_Y)\tens{A}/(\eta\,\vert \tens{A}\vert)$; for $|\tens{A}| \leq \sigma_Y$, $\varphi^*(\tens{A}) = 0$ and the subdifferential gives $\dot{\tens{z}} = 0$ (elastic domain). Note that the expression $\eta\,\dot{\tens{z}}^2/2 + \sigma_Y|\dot{\tens{z}}|$ sometimes associated with the Bingham model is in fact the \emph{primal} dissipation potential $\varphi(\dot{\tens{z}})$; its Legendre--Fenchel conjugate yields the dual form given above.
    
    \item[\textnormal{(iii)}] \textbf{Force-magnitude-independent flow} ($\varphi^*$ positively 1-homogeneous in $\tens{A}$). The dual potential $\varphi^*(\tens{A}) = \sigma_Y|\tens{A}|$ yields the flow rule $\dot{\tens{z}} \in \sigma_Y\,\subdiff|\tens{A}|$: $\dot{\tens{z}} = \sigma_Y\,\tens{A}/\vert\tens{A}\vert$ for $\vert\tens{A}\vert \neq 0$ and $\dot{\tens{z}} \in \{\tens{v} \in \mathbb{R}^N : \vert\tens{v}\vert \leq \sigma_Y\}$ for $\vert\tens{A}\vert = 0$. Because $\subdiff\varphi^*$ is 0-homogeneous in $\vert\tens{A}\vert$, the flow rate depends only on the direction of the driving force, not its magnitude.

    \item[\textnormal{(iv)}] \textbf{Rate-independent plasticity} ($\varphi^*$ is the indicator function of a convex set $K$). The dual potential $\varphi^*(\tens{A}) = \mathcal{I}_K(\tens{A})$ yields $\dot{\tens{z}} \in N_K(\tens{A})$ (normal cone), while the conjugate is the support function $\varphi(\dot{\tens{z}}) = \sigma_K(\dot{\tens{z}})$, which is positively 1-homogeneous in~$\dot{\tens{z}}$. For $K = \{\tens{A} \in \mathbb{R}^N : |\tens{A}| \leq \sigma_Y\}$ (von Mises plasticity), $\dot{\tens{z}} = \tens{0}$ when $|\tens{A}| < \sigma_Y$ (elastic domain) and $\dot{\tens{z}}$ is indeterminate at $|\tens{A}| = \sigma_Y$ (determined by consistency). Since indicator functions are not finite-valued algebraic expressions, they lie outside the grammar $\mathcal{G}_{\mathrm{cvx}}^{\mathrm{comp}}$ built in this paper. However, viscoplastic potentials with genuine elastic domains, such as the Bingham potential in~(ii), provide finite-valued approximations that converge to the rate-independent limit as the viscosity $\eta \to 0^+$, following the classical Perzyna regularisation~\cite{perzyna1966fundamental}.
\end{enumerate}
The grammar $\mathcal{G}_{\mathrm{cvx}}^{\mathrm{comp}}$ developed in Section~\ref{sec:convexity} generates potentials spanning categories~(i)--(iii). 
\end{remark}

\section{Methodology}\label{sec:methodology}

The main objective of this section is to develop a symbolic regression methodology for discovering the dual dissipation potential $\varphi^*$. The key challenge is to ensure that the discovered potential satisfies the thermodynamic constraints of convexity and non-negativity, which are sufficient for guaranteeing non-negative mechanical dissipation (Proposition~\ref{prop:dissipation} in previous section).  The key components are detailed in the following subsections.

\subsection{Problem Formulation}

We formulate the methodology within a general three-dimensional continuum framework before reducing to a one-dimensional setting for the computational examples in Sec.~\ref{sec:synthetic} and data fitting in Sec.~\ref{sec:experimental}. Upon assuming that the inelastic dissipation is purely isochoric, a common assumption for many materials, the internal viscous strain tensor $\tens{z}$ is strictly deviatoric ($\trace(\tens{z}) = 0$). We correspondingly decompose the observable strain tensor into its volumetric and deviatoric parts, $\tens{\varepsilon} = \frac{1}{3}(\trace\tens{\varepsilon})\tens{I} + \tens{\varepsilon}'$.

The Helmholtz free energy $\psi$ is assumed as an isotropic quadratic form and additively decomposed into a purely elastic volumetric part and an inelastic deviatoric part:
\begin{equation}\label{eq:3d-psi}
    \psi(\tens{\varepsilon}, \tens{z}) = \frac{1}{2} K_\infty (\trace\tens{\varepsilon})^2 + G_1 \trace((\tens{\varepsilon}' - \tens{z})^2),
\end{equation}
where $K_\infty$ is the bulk modulus and $G_1$ is the deviatoric shear modulus.
The internal variable $\tens{z}$ is an inelastic strain that relaxes towards the observable deviatoric strain $\tens{\varepsilon}'$ under the driving force $\tens{A}$. In the following we assume that the elastic coefficients $G_1$ and $K_\infty$ are known \emph{a priori}. Adopting this two-stage procedure, elastic identification followed by dissipative identification, decouples the discovery problem and represents standard practice in constitutive model calibration.

\begin{table}[t]
\centering
\caption{Classical dissipation potentials and their physical interpretations in the 1D case. Here $\eta > 0$ denotes viscosity, $\sigma_Y > 0$ the yield stress, and $q > 1$ the power-law exponent. This table motivates the primitive function library $\mathcal{F}_{\mathrm{prim}}$ (Tab.~\ref{tab:primitives}).}
\label{tab:potentials}
\begin{tabular}{@{}lrll@{}}
\toprule
Dual potential $\varphi^*(A)$ & Flow rule $\dot{z}$ & &Material behavior \\
\midrule
$\displaystyle\frac{1}{2\eta}A^2$ & $\displaystyle\frac{A}{\eta}$ && Newtonian viscosity \cite{batchelor1967introduction} \\[12pt]
$\displaystyle\frac{1}{q}|A|^q$ & $|A|^{q-1}\sgn(A)$ && Power-law \cite{ostwald1925ueber} \\[12pt]
$\sigma_Y |A|$ & $\sigma_Y \sgn(A)$ && Constant-rate flow (rate limiter)$^\dagger$ \\[12pt]
$\displaystyle\frac{1}{2\eta}\macaulay{|A| - \sigma_Y}^2$ & 
$\displaystyle\frac{\macaulay{|A|-\sigma_Y}}{\eta}\sgn(A)$&$ |A|\geq \sigma_Y $
& Bingham viscoplastic$^\ddagger$ \cite{bingham1922fluidity} \\[7pt]
& $0$&$ |A| < \sigma_Y $&\\[12pt]
$\displaystyle\frac{\eta}{a^2}\left(\cosh(aA) - 1\right)$ & $\displaystyle\frac{\eta}{a}\sinh(aA)$ & &Eyring viscosity \cite{eyring1936viscosity} \\[8pt]
\bottomrule
\end{tabular}

\vspace{4pt}
\begin{minipage}{\linewidth}
\footnotesize
$^\dagger$ The potential $\sigma_Y|A|$ is 1-homogeneous in $A$, yielding a flow rate whose magnitude $\sigma_Y$ is independent of $|A|$. This model constrains the rate (not the force) and is not rate-independent in the classical sense; see Remark~\ref{rem:classification}(iii)--(iv).\\
$^\ddagger$ The expression $\eta\dot{z}^2/2 + \sigma_Y|\dot{z}|$ sometimes listed as the Bingham ``dual'' potential is in fact the primal dissipation potential $\varphi(\dot{z})$; its Legendre--Fenchel conjugate yields the dual $\varphi^*(A) = \macaulay{|A|-\sigma_Y}^2/(2\eta)$ listed here.
\end{minipage}
\end{table}

The Cauchy stress $\tens{\sigma}$ and the deviatoric thermodynamic driving force $\tens{A}$ follow directly from standard thermodynamic arguments:
\begin{align}
    \tens{\sigma} &= \frac{\partial \psi}{\partial \tens{\varepsilon}} = K_\infty (\trace\tens{\varepsilon})\tens{I} + 2G_1 (\tens{\varepsilon}' - \tens{z}), \label{eq:3d-stress} \\
    \tens{A} &= -\frac{\partial \psi}{\partial \tens{z}} = 2G_1 (\tens{\varepsilon}' - \tens{z}). \label{eq:3d-driving-force}
\end{align}

The kinetic evolution is governed by the normality rule $\dot{\tens{z}} \in \subdiff_{\tens{A}}\varphi^*(\tens{A})$ ~\eqref{eq:evolution}.

Material characterization is typically conducted under pure shear condition \cite{Califano2023}. Therefore, the volumetric strain vanishes ($\trace\tens{\varepsilon} = 0$, implying $\tens{\varepsilon}' = \tens{\varepsilon}$), causing the tensorial quantities to decouple. Letting $\gamma$, $\tau$, $z$, and $A$ denote the active scalar shear components, the generalized state equations simplify to:
\begin{equation}\label{eq:1d-shear-state}
\tau = G_1 (\gamma - z) = A, \qquad A = G_1 (\gamma - z).
\end{equation}
When the free energy includes an additional equilibrium shear spring of modulus $G_\infty$ (as in the synthetic benchmarks of Section~\ref{sec:synthetic} and the experimental application of Section~\ref{sec:experimental}), the state law generalizes to $\tau = G_\infty\gamma + G_1(\gamma-z) = G_\infty\gamma + A$, while the thermodynamic driving force $A$ is unchanged.

Given $N_s$ measured strain--stress trajectories $\{(t_j, \gamma_j, \tau_j^{\mathrm{meas}})\}_{j=1}^{N_s}$ with $\gamma_j=\gamma(t_j)$ and $\tau_j^{\mathrm{meas}}=\tau^{\mathrm{meas}}(t_j)$, and the known free energy, the identification problem of the dual dissipation potential $\varphi^*$ is:
\begin{equation}\label{eq:sr-objective}
    \text{find } \varphi^* \in \mathcal{G}_{\mathrm{cvx}}^{\mathrm{comp}} \quad \text{such that} \quad \tau^{\mathrm{pred}}(t_j) \approx \tau_j^{\mathrm{meas}} \quad \forall\, t_j,
\end{equation}
where the predicted stress $\tau_j^{\mathrm{pred}}$ is obtained by solving the initial value problem:
\begin{equation}\label{eq:ivp}
    \dot{z} \in \subdiff_A \varphi^*(A), \quad A = G_1(\gamma_j - z), \quad z(0) = 0, \quad \tau_j^{\mathrm{pred}} = G_1(\gamma_j - z).
\end{equation}
This is a nonlinear inverse problem: the unknown $\varphi^*$ appears inside the differential equation whose solution determines the observable stress. 

\subsection{Convexity-Preserving Grammar}\label{sec:convexity}

Before presenting the formal grammar construction, we outline its conceptual structure by means of an example. Let us consider the following candidate potential:
\begin{equation}\label{eq:grammar-example}
    \varphi^*(A) = c_1\,|A|^p + c_2\,\bigl(\cosh(aA) - 1\bigr)^q,
\end{equation}
that, with $c_1, c_2, a > 0$, $p \geq 1$, $q \geq 1$, is convex, non-negative and symmetric in $A$, obtained as a combination of Power-law and Eyring viscosities (see Tab.~\ref{tab:potentials}). The construction proceeds in three steps:
\begin{enumerate}
    \item[\textup{(i)}] $g_1(A) = |A|^p$ and $g_2(A) = \cosh(aA) - 1$ are called \emph{inner primitives}, i.e., the basic building blocks for constructing more complex potentials.
    \item[\textup{(ii)}] $f(u) = u^q$ ($q \geq 1$) is the \emph{outer function}, i.e., a function that takes the inner primitive $g_2(A)$ as input and produces a new function $h(A) = f(g_2(A))$ by composition. 
    \item[\textup{(iii)}] The final potential is the \emph{positive linear combination} $\varphi^*(A) = c_1\,g_1(A) + c_2\,h(A)$ ($c1> 0, c2> 0$).
\end{enumerate}

The construction procedure relies on the following standard results from convex analysis \cite{boyd2004convex, rockafellar1970convex}.

\begin{proposition}[Composition rule for convexity]\label{prop:composition}
Let $g: \mathbb{R} \to \mathbb{R}$ be convex and let $f: \mathbb{R} \to \mathbb{R}$ be convex and non-decreasing on the range of $g$. Then the composition $h = f \circ g$, defined by $h(A) = f(g(A))$, is convex.\footnote{This is a standard result in convex analysis. The result holds without any differentiability assumption on $f$ or $g$: for $\lambda \in [0,1]$, convexity and monotonicity of $f$ give $f(g(\lambda A_1 + (1-\lambda)A_2)) \leq f(\lambda g(A_1) + (1-\lambda) g(A_2)) \leq \lambda f(g(A_1)) + (1-\lambda) f(g(A_2))$, where the first inequality uses the non-decreasing property of $f$ combined with the convexity of $g$, and the second uses the convexity of $f$. When $f$ and $g$ are both twice differentiable, an equivalent proof follows from the chain rule: $h''(A) = f''(g(A))\,(g'(A))^2 + f'(g(A))\,g''(A) \geq 0$, since $f'' \geq 0$, $f' \geq 0$, and $g'' \geq 0$.}
\end{proposition}

\begin{proposition}[Closure under positive linear combinations]\label{prop:closure}
Let $f_1, f_2: \mathbb{R} \to \mathbb{R}$ be convex and non-negative, and let $c > 0$. Then:
\begin{enumerate}
    \item[\textnormal{(i)}] $f_1 + f_2$ is convex and non-negative;
    \item[\textnormal{(ii)}] $c\,f_1$ is convex and non-negative.
\end{enumerate}
\end{proposition}
\noindent In addition, the normalization $f(\tens{0}) = 0$ required by~Proposition~\ref{prop:dissipation} is trivially preserved: if $f_1(\tens{0}) = f_2(\tens{0}) = 0$, then $(f_1 + f_2)(\tens{0}) = 0$ and $(c\,f_1)(\tens{0}) = 0$.

Propositions~\ref{prop:composition} and~\ref{prop:closure} together provide the theoretical foundation for convexity-by-construction approaches in data-driven constitutive modeling. Input-convex neural architectures~\cite{amos2017input, fuhg2024review} exploit both properties to enforce convexity of learned strain-energy or dissipation potentials through architectures built on positive-weight sums and compositions with convex, non-decreasing activation functions. The grammar $\mathcal{G}_{\mathrm{cvx}}^{\mathrm{comp}}$ introduced below operationalises the same mathematical principles in an explicitly symbolic setting.

\medskip

By restricting outer functions to be convex and non-decreasing and inner functions to be convex and non-negative, every composition generated by the grammar inherits convexity by construction. Besides convexity, compositions must also respect domains. For instance, $f(u)=u^p$ with non-integer $p$ is only defined for $u\ge 0$, so we must ensure that the inner function takes values in $\mathbb{R}^+$. In our setting this is handled by design: the grammar guarantees that all inner subexpressions are non-negative through construction (primitives are non-negative, positive scaling preserves non-negativity, and sums of non-negative terms remain non-negative). Moreover, outer functions are restricted to those defined on $\mathbb{R}^+$ (such as $u^q$ with $q\ge 1$, $\exp(u)-1$, $\cosh(u)-1$, and $\text{softplus}(u)-\log(2)$), and constant tuning maintains non-negativity by enforcing lower bounds of zero for coefficients and one for exponents. Building on these results, we define the following function classes:

\paragraph{Primitive (inner) functions.}
We begin with a library of primitive building blocks $\mathcal{F}_{\mathrm{prim}}$, chosen so that each $g(A)$ is convex and non-negative and satisfies the natural normalization $g(0)=0$. We also enforce symmetry in $A$, which reflects the fact that dissipation depends on the magnitude of the driving force rather than its sign:
\begin{equation}\label{eq:primitives}
\begin{aligned}
    \mathcal{F}_{\mathrm{prim}} = \Big\{ 
    & |A|^p \; (p \geq 1), \; \cosh(aA) - 1, \; \log\cosh(aA), \\
    & \delta\left(\sqrt{1 + A^2/\delta^2} - 1\right), \; \sinh^2(aA), \\
    & \frac{1}{a^2}\left(e^{a|A|} - 1 - a|A|\right), \; \frac{1}{a}\left(e^{aA^2} - 1\right), \\
    & \macaulay{|A| - \sigma_Y}^r \; (r \geq 1,\; \sigma_Y \geq 0) \Big\},
\end{aligned}
\end{equation}
where $a > 0$ and $\delta > 0$ are scale parameters.

Table~\ref{tab:primitives} summarizes the primitive functions, their induced flow rules, and physical interpretations. These primitives are not meant to be exhaustive, but to span a range of commonly used dissipation mechanisms while remaining convex. The majority of these primitives are $C^\infty$; the exception is $|A|^p$ with $1 \leq p < 2$, which is $C^1$ but not $C^2$ at $A = 0$ (see Remark~\ref{rem:comp-regularity}).

\begin{table}[t]
\centering
\caption{Primitive functions in $\mathcal{F}_{\mathrm{prim}}$ with their induced flow rules and physical interpretations. All functions are convex, non-negative, symmetric, and normalized ($g(0) = 0$).}
\label{tab:primitives}
\begin{tabular}{@{}lll@{}}
\toprule
Primitive & Flow rule & Physical interpretation \\
\midrule
$|A|^p$, $p \geq 1$ & $p|A|^{p-1}\sgn(A)$ & Power-law (Newtonian viscosity when $p=2$)\\[4pt]
$\cosh(aA) - 1$ & $a\sinh(aA)$ & Eyring viscosity \\[4pt]
$\log\cosh(aA)$ & $a\tanh(aA)$ & Saturating viscosity \\[4pt]
$\delta\left(\sqrt{1 + A^2/\delta^2} - 1\right)$ & $\displaystyle\frac{A}{\sqrt{A^2 + \delta^2}}$ & Pseudo-Huber \\[8pt]
$\sinh^2(aA)$ & $a\sinh(2aA)$ & Hyperbolic Eyring \\[4pt]
$\displaystyle\frac{1}{a^2}\left(e^{a|A|} - 1 - a|A|\right)$ & $\displaystyle\frac{1}{a}\left(e^{a|A|} - 1\right)\sgn(A)$ & Arrhenius activation \\[8pt]
$\displaystyle\frac{1}{a}\left(e^{aA^2} - 1\right)$ & $2Ae^{aA^2}$ & Exponential stiffening \\[4pt]
$\macaulay{|A| - \sigma_Y}^r$, $r \geq 1$, $\sigma_Y \geq 0$ & $r\macaulay{|A|-\sigma_Y}^{r-1}\sgn(A)$ & Viscoplastic yield ($r{=}2$: Bingham), $|A|\geq\sigma_Y$\\[4pt]
& $0$ & $|A| < \sigma_Y$ \\[4pt]
\bottomrule
\end{tabular}
\\
\begin{flushleft}
\begin{footnotesize}
\end{footnotesize}
\end{flushleft}
\end{table}

\paragraph{Outer functions.}
To enrich the functional class beyond positive sums of primitives, we allow compositions with a small set of outer functions. The class $\mathcal{F}_{\mathrm{out}}$ is selected so that every $f:\mathbb{R}^+\to\mathbb{R}^+$ is convex and non-decreasing and satisfies $f(0)=0$, exactly matching the hypotheses of Proposition~\ref{prop:composition}:
\begin{equation}\label{eq:outer}
    \mathcal{F}_{\mathrm{out}} = \left\{ u^q \; (q \geq 1), \; e^{u} - 1, \; \cosh(u) - 1, \; \softplus(u) - \log 2 \right\},
\end{equation}
where $\softplus(u) = \log(1 + e^u)$.\footnote{The softplus function~\cite{dugas2000softplus} is a smooth, convex approximation of the rectified linear unit (ReLU) $\max(u,0)$, which coincides with the Macaulay bracket $\macaulay{u}$. The shift by $-\log 2$ enforces the normalization $f(0) = 0$ required by property (P2).}
The verification that each primitive and outer function satisfies the required properties is straightforward and follows by direct computation of the second derivative, evaluation at zero, and symmetry arguments.

We now define the convexity-preserving grammar $\mathcal{G}_{\mathrm{cvx}}^{\mathrm{comp}}$ through its production rules:\label{def:grammar}
\begin{align}
    \texttt{ConvexExpr} &\to \texttt{ConvexTerm} \mid \texttt{ConvexExpr} + \texttt{ConvexExpr} \mid c \cdot \texttt{ConvexExpr}, \quad c > 0 \label{eq:grammar-expr} \\[4pt]
    \texttt{ConvexTerm} &\to \texttt{Primitive} \mid \texttt{Composition} \label{eq:grammar-term} \\[4pt]
    \texttt{Primitive} &\to |A|^p \mid (\cosh(aA) - 1) \mid \log\cosh(aA) \nonumber \\
    & \quad \mid \delta(\sqrt{1 + A^2/\delta^2} - 1) \mid \sinh^2(aA) \nonumber \\
    & \quad \mid (e^{a|A|} - 1 - a|A|)/a^2 \mid (e^{aA^2} - 1)/a \nonumber \\
    & \quad \mid \macaulay{|A| - \sigma_Y}^r \label{eq:grammar-prim} \\[4pt]
    \texttt{Composition} &\to \texttt{OuterFunc}(\texttt{InnerFunc}) \label{eq:grammar-comp} \\[4pt]
    \texttt{OuterFunc} &\to (\cdot)^q \mid e^{(\cdot)} - 1 \mid \cosh(\cdot) - 1 \mid \softplus(\cdot) - \log 2 \label{eq:grammar-outer} \\[4pt]
    \texttt{InnerFunc} &\to \texttt{Primitive} \mid \texttt{InnerFunc} + \texttt{InnerFunc} \mid c \cdot \texttt{InnerFunc} \mid \texttt{Composition} \label{eq:grammar-inner}
\end{align}
where $p \geq 1$, $q \geq 1$, $r \geq 1$, $a \neq 0$, $\delta > 0$, and $\sigma_Y \geq 0$.

\noindent In \eqref{eq:grammar-expr}--\eqref{eq:grammar-inner}, the symbols \texttt{ConvexExpr}, \texttt{ConvexTerm}, \texttt{Primitive}, \texttt{Composition}, \texttt{OuterFunc}, and \texttt{InnerFunc} denote \emph{classes of admissible functions}. The arrow ``$\to$'' should be read as ``may be constructed as''. Thus, \eqref{eq:grammar-expr} states that an admissible potential can be a single convex term, a sum of admissible potentials, or a positively scaled admissible potential. A \texttt{ConvexTerm}~\eqref{eq:grammar-term} is either a single \texttt{Primitive} (one elementary convex dissipative mechanism from Tab.~\ref{tab:primitives}) or a \texttt{Composition} of the form $f(g(A))$~\eqref{eq:grammar-comp} with $f \in \mathcal{F}_{\mathrm{out}}$ and $g$ built by \texttt{InnerFunc}~\eqref{eq:grammar-inner}. Since \texttt{InnerFunc} may itself contain a \texttt{Composition}, the grammar permits arbitrary nesting up to a maximum depth\footnote{To maintain tractability, the maximum depth of nested compositions is controlled by the parameter $D_{\max}$, which limits the recursion depth to ensure a finite search space while covering physically relevant potentials.}. The recursive structure of \texttt{InnerFunc} ensures that the inner argument $g(A)$ remains convex and non-negative, since it is built from primitives, their positive combinations, and nested compositions of convex functions.

By structural induction on the grammar, it is possible to prove that above generated grammar satisfies the following (see Appendix~\ref{proof:grammar} for the complete argument). 

\begin{theorem}[Grammar Consistency]\label{thm:grammar}
Every expression $\varphi^*$ generated by $\mathcal{G}_{\mathrm{cvx}}^{\mathrm{comp}}$ satisfies properties \emph{(P1)} and \emph{(P2)} in Proposition~\ref{prop:dissipation}: $\varphi^*$ is convex, non-negative, and normalized ($\varphi^*(0) = 0$).
\end{theorem}

\begin{figure}[t]
    \centering
    \begin{tikzpicture}[
        level 1/.style={sibling distance=5cm, level distance=1.4cm},
        level 2/.style={sibling distance=2.5cm, level distance=1.4cm},
        level 3/.style={sibling distance=1.8cm, level distance=1.4cm},
        level 4/.style={sibling distance=1.5cm, level distance=1.4cm},
        every node/.style={font=\small},
        operator/.style={circle, draw, thick, fill=orange!25, minimum size=9mm},
        terminal/.style={rectangle, draw, thick, fill=blue!20, rounded corners=3pt, 
                         minimum width=12mm, minimum height=7mm},
        constant/.style={rectangle, draw, thick, fill=violet!20, rounded corners=3pt, 
                         minimum width=12mm, minimum height=7mm},
        composition/.style={rectangle, draw=green!60!black, thick, fill=green!25, 
                            rounded corners=3pt, minimum width=14mm, minimum height=8mm},
        edge from parent/.style={draw, thick, -{Stealth[length=2.5mm]}}
    ]
    
    \node[operator] {$+$}
        child {
            node[operator] {$\times$}
            child {node[constant] {$c_1$}}
            child {
                node[composition] {$(\cdot)^q$}
                child {
                    node[operator] {$+$}
                    child {node[terminal] {$|A|^{p_1}$}}
                    child {
                        node[operator] {$\times$}
                        child {node[constant] {$c_2$}}
                        child {node[terminal] {$|A|^{p_2}$}}
                    }
                }
            }
        }
        child {
            node[operator] {$\times$}
            child {node[constant] {$c_3$}}
            child {node[terminal] {$|A|^{p_3}$}}
        };
    
    \matrix[anchor=north west, draw, rounded corners, fill=gray!5, 
            column sep=3mm, row sep=1mm, inner sep=4pt,
            font=\footnotesize] at (4.5, 0) {
        \node[operator, minimum size=5mm] {}; & \node {Binary operator}; \\
        \node[terminal, minimum width=8mm, minimum height=5mm] {}; & \node {Primitive $g \in \mathcal{F}_{\mathrm{prim}}$}; \\
        \node[constant, minimum width=8mm, minimum height=5mm] {}; & \node {Constant $c > 0$}; \\
        \node[composition, minimum width=10mm, minimum height=5mm] {}; & \node {Outer function $f \in \mathcal{F}_{\mathrm{out}}$}; \\
    };
    
    \end{tikzpicture}
    \caption{Expression tree representation of a candidate dissipation potential $\varphi^*(A) = c_1 (|A|^{p_1} + c_2 |A|^{p_2})^q + c_3 |A|^{p_3}$ generated by the convexity-preserving grammar $\mathcal{G}_{\mathrm{cvx}}^{\mathrm{comp}}$. Internal nodes (orange circles) represent binary operators; blue rectangular nodes denote primitive functions from $\mathcal{F}_{\mathrm{prim}}$; violet nodes are positive scaling constants; the green node represents composition with an outer function from $\mathcal{F}_{\mathrm{out}}$, applied to the convex non-negative subtree below it.}
    \label{fig:expression-tree}
    \end{figure}
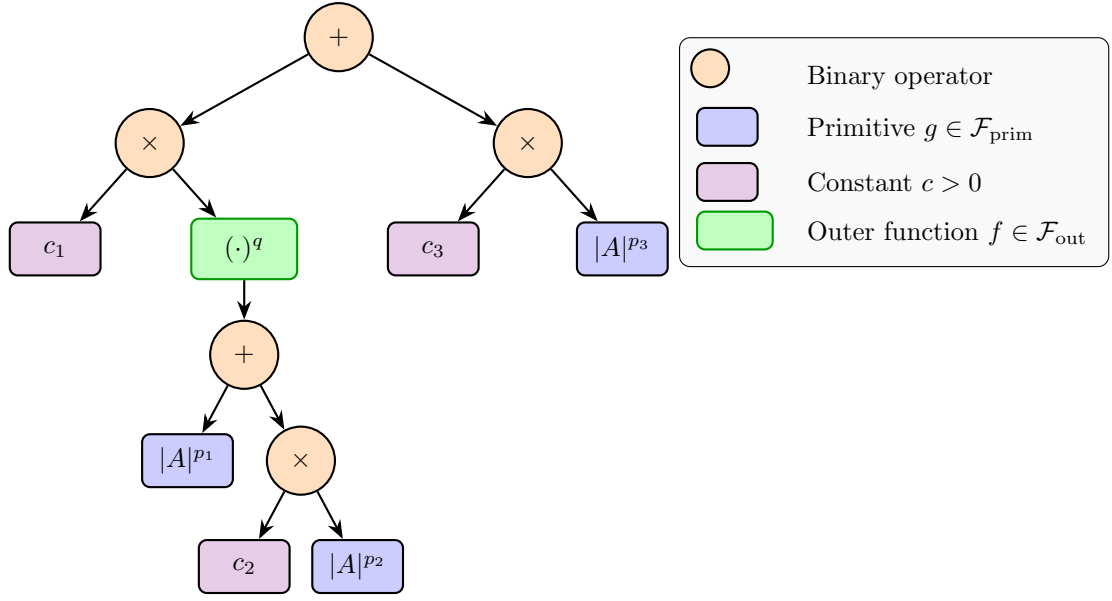

\paragraph{Constitutive equivalence.}\label{par:symbolic-equivalence}
It is worth remarking that two expressions $\varphi^*_1(A)$ and
$\varphi^*_2(A)$ are \emph{constitutively equivalent} if they induce the same
evolution law, i.e., their subdifferentials coincide for all $A$:
$\subdiff_A \varphi^*_1(A) = \subdiff_A \varphi^*_2(A)$ for all $A \in \mathbb{R}$.
For convex functions this occurs if and only if $\varphi^*_1(A) - \varphi^*_2(A) = c$ for some constant $c$
independent of $A$. The normalization $\varphi^*(0) = 0$ (automatically
satisfied by our grammar) resolves this ambiguity: if two normalized convex potentials have identical
subdifferentials, then they are identical, i.e., $\varphi^*_1(0) = \varphi^*_2(0) = 0$
together with $\subdiff_A \varphi^*_1 = \subdiff_A \varphi^*_2$ implies $\varphi^*_1(A) = \varphi^*_2(A)$ for all $A$.

\begin{remark}[Differentiability of compositions]\label{rem:comp-regularity}
While Proposition~\ref{prop:composition} guarantees \emph{convexity} of $f \circ g$ without smoothness requirements, the \emph{differentiability} of the composition depends on the regularity of both $f$ and $g$. In the grammar $\mathcal{G}_{\mathrm{cvx}}^{\mathrm{comp}}$, all outer functions $f \in \mathcal{F}_{\mathrm{out}}$ are $C^\infty$ on $\mathbb{R}^+$, and most inner primitives are $C^\infty$ everywhere except for $|A|^p$ with $1 \leq p < 2$, which is $C^1$ but not $C^2$ at $A = 0$, and the Macaulay bracket primitive $\macaulay{|A|-\sigma_Y}^r$, which is $C^{r-1}$ at $|A| = \sigma_Y$ for integer $r$ and $C^1$ for $r > 1$. For any $p > 1$, the flow rule $\partial\varphi^*/\partial A = cp|A|^{p-1}\sgn(A)$ remains continuous and single-valued, so the ODE formulation $\dot{z} = \partial\varphi^*/\partial A$ is well-posed. Only the limiting case $p = 1$ requires the full subdifferential formulation of Proposition~\ref{prop:dissipation}, as the flow rule $\dot{z} \in c\,\subdiff|A|$ becomes set-valued at $A = 0$. The numerical consequences of this loss of regularity, both for the numerical scheme and for the smoothness of the fitness landscape, are discussed in Section~\ref{sec:optimization}.
\end{remark}

\begin{remark}\label{rem:macaulay}
The shifted Macaulay bracket primitive $\macaulay{|A| - \sigma_Y}^r$ with $r \geq 1$ and $\sigma_Y \geq 0$ extends the grammar to viscoplastic potentials exhibiting a genuine elastic domain. For $|A| \leq \sigma_Y$, the primitive vanishes identically, giving $\dot{z} = 0$ (no dissipation); for $|A| > \sigma_Y$, it produces rate-dependent flow $\dot{z} = r\macaulay{|A|-\sigma_Y}^{r-1}\sgn(A)$ whose intensity grows with the excess force above the yield threshold~$\sigma_Y$. The special case $r = 2$ recovers the classical Bingham viscoplastic dual potential $\macaulay{|A|-\sigma_Y}^2/(2\eta)$ (Tab.~\ref{tab:potentials}). Convexity of $\macaulay{|A|-\sigma_Y}^r$ for $r \geq 1$ follows from two successive applications of the composition rule (Proposition~\ref{prop:composition}): first, $A \mapsto |A| - \sigma_Y$ is convex and $\max(\cdot,0)$ is convex and non-decreasing, so $A \mapsto \macaulay{|A|-\sigma_Y}$ is convex and non-negative; second, $(\cdot)^r$ with $r \geq 1$ is convex and non-decreasing on $\mathbb{R}^+$, so the full map $A \mapsto \macaulay{|A|-\sigma_Y}^r$ is convex. Non-negativity and normalization ($\macaulay{|0| - \sigma_Y}^r = 0$ for $\sigma_Y \geq 0$) are immediate. For $\sigma_Y = 0$ the primitive reduces to $|A|^r$, so the threshold $\sigma_Y \geq 0$ lets the GP smoothly interpolate between yield-free and yield-stress behaviors. The non-differentiability of the bracket at $|A| = \sigma_Y$ has implications for parameter tuning, addressed in Section~\ref{sec:optimization}.
\end{remark}

\subsection{Evolutionary Optimization}\label{sec:optimization}

The evolutionary algorithm sketched in Fig.~\ref{fig:flowchart} relies on the following main components: \emph{population initialization} by random sampling from $\mathcal{G}_{\mathrm{cvx}}^{\mathrm{comp}}$; \emph{fitness evaluation} via forward rollout (forward time integration of the evolution law $\dot{z} = \partial_A \varphi^*$) and comparison with training data; \emph{selection and elitism} through tournament selection with elite preservation; \emph{convexity-preserving genetic operators} (crossover and mutation) that maintain grammar membership; and \emph{parameter optimization} by a constrained two-mode local search (L-BFGS-B for smooth candidates, Nelder--Mead for those containing the non-differentiable Macaulay bracket) at each generation.

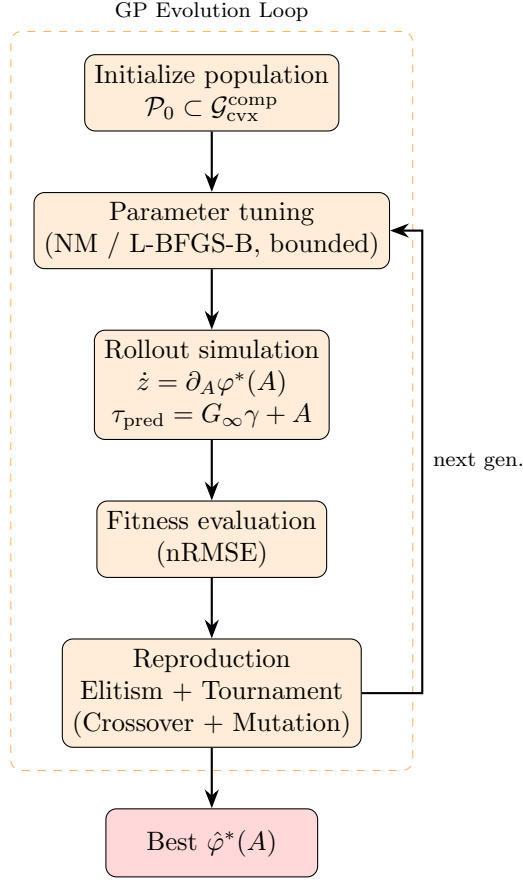
\begin{figure}[t]
\centering
\begin{tikzpicture}[
    node distance=0.8cm and 1.2cm,
    box/.style={rectangle, draw, rounded corners, minimum width=2.8cm, minimum height=0.9cm, align=center, font=\small},
    data/.style={box, fill=blue!15},
    process/.style={box, fill=green!15},
    gp/.style={box, fill=orange!15},
    output/.style={box, fill=red!15},
    arrow/.style={-{Stealth[length=2.5mm]}, thick}
]

\node[gp] (init) {Initialize population\\$\mathcal{P}_0 \subset \mathcal{G}_{\mathrm{cvx}}^{\mathrm{comp}}$};
\node[gp, below=of init] (tune) {Parameter tuning\\(NM / L-BFGS-B, bounded)};
\node[gp, below=of tune] (rollout) {Rollout simulation\\$\dot{z} = \partial_A \varphi^*(A)$\\$\tau_{\text{pred}} = G_\infty\gamma + A$};
\node[gp, below=of rollout] (fitness) {Fitness evaluation\\(nRMSE)};
\node[gp, below=of fitness] (evolve) {Reproduction\\Elitism + Tournament\\(Crossover + Mutation)};

\node[output, below=of evolve] (output) {Best $\hat{\varphi}^*(A)$};

\draw[arrow] (init) -- (tune);
\draw[arrow] (tune) -- (rollout);
\draw[arrow] (rollout) -- (fitness);
\draw[arrow] (fitness) -- (evolve);
\draw[arrow] (evolve) -- (output);

\draw[arrow] (evolve.east) -- ++(0.8,0) |- node[pos=0.25, right, font=\scriptsize] {next gen.} (tune.east);

\begin{scope}[on background layer]
\node[draw=orange!60, dashed, rounded corners, fit=(init)(tune)(rollout)(fitness)(evolve), inner sep=0.3cm, label={[font=\scriptsize]above:GP Evolution Loop}] {};
\end{scope}

\end{tikzpicture}
\caption{Schematic of the GP evolution loop: candidate potentials from the composition-extended convex grammar $\mathcal{G}_{\mathrm{cvx}}^{\mathrm{comp}}$ are evolved via genetic programming with constrained parameter tuning at each generation. For each candidate $\hat{\varphi}^*(A)$, the evolution law $\dot{z} = \partial_A \hat{\varphi}^*(A)$ is integrated forward in time (rollout), and the predicted stresses are compared with the target values. The specific quantities entering the nRMSE differ between settings (stress signals for the synthetic benchmarks and the experimental campaign, storage/loss moduli $G', G''$ for DMA fitting); see the respective sections for details.}
\label{fig:flowchart}
\end{figure}

\paragraph{(I) Population initialization.}\label{rem:initialization}
The initial population $\mathcal{P}_0$ is generated with a tri-modal sampler that combines simple ``base'' individuals with progressively more expressive multi-term ones, ensuring that the genetic operators always start from a diverse pool. Each of the $N_{\mathrm{pop}}$ individuals is drawn at random from one of three categories:
\begin{itemize}
    \item with probability $25\%$, a \emph{pure-base} individual: a single primitive $g\in\mathcal{F}_{\mathrm{prim}}$ without any outer wrapping or scaling, guaranteeing that individual primitives survive the initial population and are not hidden inside compositions that the complexity guard (the node-count cap $N_{\max}$ applied during fitness evaluation, introduced below) would discard;
    \item with probability $25\%$, a \emph{single-block} individual: one term produced by the recursive sampler described below, possibly including an outer composition and a positive scaling;
    \item with probability $50\%$, a \emph{multi-block} individual: an additive combination of $K\in\{1,2,3\}$ independent terms produced by the same recursive sampler.
\end{itemize}
The recursive sampler builds each term by walking the grammar $\mathcal{G}_{\mathrm{cvx}}^{\mathrm{comp}}$: at every internal node, with probability $\pi_{\Sigma}$ it emits an intra-block sum $\mathtt{expr} + \mathtt{expr}$; otherwise it selects uniformly among (i) a primitive base $g\in\mathcal{F}_{\mathrm{prim}}$, (ii) an outer composition $f\circ\mathtt{expr}$ with $f\in\mathcal{F}_{\mathrm{out}}$, or (iii) a positive scaling $c\cdot\mathtt{expr}$ with $c>0$. A branch terminates as soon as a primitive base~(i) is selected; options~(ii) and~(iii) recurse on their inner argument. As a safeguard, recursion is also forced to stop when the depth $D_{\max}$ is reached, at which point a primitive is returned. Throughout this work $\pi_{\Sigma}=0.05$.

Tunable parameters are initialized from uniform distributions:
\begin{itemize}
    \item Primitive coefficients and amplitudes: $c, a \sim \mathcal{U}(0.05, 3.0)$, constrained to $c, a > 0$; structural scaling constants $c \cdot \mathtt{expr}$ use the slightly wider range $\mathcal{U}(0.01, 3.0)$;
    \item Power-law and composition exponents: $p, q \sim \mathcal{U}(1.0, 3.0)$, constrained to $p, q \geq 1$;
    \item Macaulay exponent: $r$ initialized to one of the two physical regimes $\{1, 2\}$ ($r=1$ Coulomb, $r=2$ Bingham), constrained to $r \geq 1$;
    \item Pseudo-Huber transition scale: $\delta \sim \mathcal{U}(0.05, 1.5)$, constrained to $\delta > 0$;
    \item Yield threshold: $\sigma_Y \sim \mathcal{U}(0.01, c_{\max})$, where the upper bound $c_{\max} = (G_\infty + G_1)\max_t|\gamma(t)|$ is set data-drivenly to the scale of the driving force, so that the initial threshold spans the physically relevant range; constrained to $\sigma_Y \geq 0$.
\end{itemize}
These ranges are design choices, selected to span physically plausible dissipation behaviors while maintaining numerical stability; the subsequent parameter optimization refines these initial values based on the training data.

The combination of the tri-modal $25/25/50$ split and the recursive single-term sampler yields a non-trivial joint distribution over expression structures, further shaped by the complexity guard $N_{\max}$ which discards expressions exceeding the maximum node count during fitness evaluation.\footnote{The node count is the total number of nodes in the (SymPy-canonicalized) expression tree, counting every operator and function node together with all leaves, i.e.\ the variable $A$ and every numeric constant (coefficients and exponents included). The same metric defines the complexity guard $N_{\max}$ and the parsimony penalty $\lambda_{\mathrm{cx}}\cdot\mathrm{nodes}$.} Table~\ref{tab:grammar-distribution} reports the resulting structural composition of $\mathcal{P}_0$ in three configurations: (i)~the raw sampler at $D_{\max}=3$ without node filtering; (ii)~the synthetic benchmark setting (Section~\ref{sec:setup}, $D_{\max}=3$, $N_{\max}=12$); and (iii)~the experimental setting (Section~\ref{sec:experimental}, $D_{\max}=6$, $N_{\max}=20$). Statistics are estimated from $20{,}000$ samples drawn with the same tri-modal sampler and $\pi_{\Sigma}=0.05$ used by the GP runs.

Without the node filter, compositions account for nearly half the population and expressions average $19.1$ nodes with up to $8$ additive terms (after SymPy simplification). When the complexity guard is active, the rejection rate is substantial ($61.6\%$ for $N_{\max}=12$, $39.7\%$ for $N_{\max}=20$), and the surviving population shifts toward structurally simpler expressions: in the synthetic setting, $99\%$ of initial individuals are single-term (one primitive or one composition), with a mean of $7.7$ nodes; in the experimental setting, $85\%$ are still single-term despite the larger budget, with a mean of $10.5$ nodes. This bias toward parsimony in the initial population does not prevent the evolutionary search from discovering complex structures when the data requires them, as demonstrated by the experimental results in Section~\ref{sec:experimental}.

\begin{table}[t]
\centering
\caption{Structural distribution of initial-population expressions generated by $\mathcal{G}_{\mathrm{cvx}}^{\mathrm{comp}}$, estimated from $20{,}000$ Monte Carlo samples using the grammar implementation. Three configurations are shown: the raw grammar without node filtering, the synthetic benchmark setting (Section~\ref{sec:setup}), and the experimental setting (Section~\ref{sec:experimental}).}
\label{tab:grammar-distribution}
\begin{tabular}{@{}lccc@{}}
\toprule
 & Raw grammar & Synthetic & Experimental \\
 & ($D_{\max}{=}3$) & ($D_{\max}{=}3$, $N_{\max}{=}12$) & ($D_{\max}{=}6$, $N_{\max}{=}20$) \\
\midrule
\emph{Rejection rate}             &  n/a    & 61.6\% & 39.7\% \\
\addlinespace
Single primitive                  & 15.8\% & 38.2\% & 24.4\% \\
Single composition                & 43.7\% & 60.8\% & 60.4\% \\
Sum of primitives                 &  6.0\% &  0.3\% &  9.1\% \\
Mixed (sums + compositions)       & 34.4\% &  0.7\% &  6.2\% \\
\addlinespace
Mean additive terms               &  1.66  &  1.01  &  1.16  \\
Mean tree nodes                   & 19.1   &  7.7   & 10.5   \\
\bottomrule
\end{tabular}
\end{table}

\paragraph{(II) Evolution.}
The population evolves for $\Gamma_{\max}$ generations. For each candidate $\hat{\varphi}^*(A)$, the evolution law $\dot{z} = \partial_A \hat{\varphi}^*(A)$ is integrated from rest ($z_0 = 0$) under the prescribed loading, and the predicted shear stress $\tau_{\mathrm{pred}}(t)$ is computed. Fitness is then evaluated by comparing the model response with the data; the specific metric depends on the application and is detailed in Sections~\ref{sec:synthetic} and~\ref{sec:experimental}.

\paragraph{(III) Selection and Elitism.}
Parent selection is performed via \emph{tournament selection}: to produce each new individual, a random subset of $k = \lfloor f_t \cdot N_{\mathrm{pop}} \rfloor$ candidates is drawn from the current population, where $f_t \in (0,1)$ is the tournament fraction controlling selection pressure, and the candidate with the best (lowest) fitness wins the tournament and becomes a parent. Additionally, the top $\lfloor f_e \cdot N_{\mathrm{pop}} \rfloor$ individuals, where $f_e \in (0,1)$ is the elite fraction, are copied unchanged into the next generation (\emph{elitism}), ensuring that the best solutions discovered so far are never lost.

\paragraph{(IV) Convexity-Preserving Genetic Operators.}
The population evolves through crossover and mutation, designed to maintain membership in $\mathcal{G}_{\mathrm{cvx}}^{\mathrm{comp}}$. Each individual is represented as an ordered list of \emph{terms} $\{B_k\}_{k=1}^K$ that recombine additively into the candidate potential $\hat{\varphi}^*(A) = \sum_{k=1}^K B_k(A)$, with each $B_k$ a \texttt{ConvexExpr} produced by the recursive sampler~\eqref{eq:grammar-expr} (a primitive, a composition, a positive scaling, or a sum thereof). Genetic operators act at the level of these terms, so that the normalization $B_k(0) = 0$ established at generation time is preserved through every recombination. Each offspring undergoes crossover with probability $p_{\mathrm{cx}}$ and mutation with probability $p_{\mathrm{mut}}$; these probabilities are applied independently, so an individual may undergo both operations.

\begin{itemize}
    \item Term-Level crossover: A random term is selected from one parent and either inserted into or used to replace a random term in the other parent. If the resulting expression exceeds the maximum node budget $N_{\max}$, a random subset of terms is retained.
    \item Term-Level mutation: Operations include replacing a term with a newly generated convex term, adding new terms, or removing terms. Removal is only allowed when at least two terms are present, ensuring that the expression remains non-trivial.
\end{itemize}

\begin{lemma}[Crossover preserves grammar membership]\label{lem:crossover}
If $\varphi^*_1, \varphi^*_2 \in \mathcal{G}_{\mathrm{cvx}}^{\mathrm{comp}}$, then their crossover offspring $\varphi^*_{\mathrm{child}} \in \mathcal{G}_{\mathrm{cvx}}^{\mathrm{comp}}$.
\end{lemma}
\begin{proof}
See Appendix~\ref{proof:operators}.
\end{proof}

\begin{lemma}[Mutation preserves grammar membership]\label{lem:mutation}
If $\varphi^* \in \mathcal{G}_{\mathrm{cvx}}^{\mathrm{comp}}$, then the mutated expression $\varphi^*_{\mathrm{mut}} \in \mathcal{G}_{\mathrm{cvx}}^{\mathrm{comp}}$.
\end{lemma}
\begin{proof}
See Appendix~\ref{proof:operators}.
\end{proof}

\paragraph{(V) Full Parameter Optimization.}
At each generation, we perform the optimization of the tunable parameters embedded in the expression tree. This approach, inspired by hybrid symbolic-numeric methods in contemporary symbolic regression frameworks \cite{lacava2021contemporary, kommenda2020parameter}, enables the genetic algorithm to focus on structural discovery while a local optimizer fine-tunes the parametric representation within each structural class.

Consider a candidate potential $\hat{\varphi}^*(A; \boldsymbol{\theta})$ where $\boldsymbol{\theta}$ collects all tunable parameters in the expression tree: amplitude parameters $a$ (non-negative scalings appearing in additive terms and inside compositional primitives), exponents $p$ in power-law terms $|A|^p$, and exponents $q$ in composition outer functions $(\cdot)^q$.

The evolution law induced by $\hat{\varphi}^*$ is
\begin{equation}\label{eq:tuning-evolution}
    \dot{z} = \frac{\partial \hat{\varphi}^*}{\partial A}(A; \boldsymbol{\theta}) =: f(A; \boldsymbol{\theta}).
\end{equation}
The optimal parameters are obtained by minimizing a rollout loss $\mathcal{L}_{\mathrm{rollout}}$ over the training data:
\begin{equation}\label{eq:nlls}
    \hat{\boldsymbol{\theta}} = \arg\min_{\boldsymbol{\theta} \in \Theta} \left\{ \mathcal{L}_{\mathrm{rollout}}(\boldsymbol{\theta}) + \lambda_{\mathrm{ridge}} \| \boldsymbol{\theta} \|_2^2 \right\},
\end{equation}
where the feasible set $\Theta$ encodes the convexity-preserving constraints:
\begin{equation}\label{eq:convexity-constraints}
    \Theta = \left\{ \boldsymbol{\theta} : a \geq 0, \quad p \geq 1, \quad q \geq 1, \quad r \geq 1, \quad \sigma_Y \geq 0 \right\}.
\end{equation}
The constraint $a \geq 0$ ensures that positive combinations of convex terms remain convex, while $p \geq 1$ and $q \geq 1$ preserve convexity of power-law and composition terms. The ridge term $\lambda_{\mathrm{ridge}} \|\boldsymbol{\theta}\|_2^2$~\cite{engl1996regularization}, where $\lambda_{\mathrm{ridge}}$ is a small positive scalar, is applied to all symbolic constants (amplitude parameters and exponents) and improves numerical conditioning; the shear moduli, when jointly optimized, are excluded from the penalty to avoid biasing them toward zero. The specific form of $\mathcal{L}_{\mathrm{rollout}}$ depends on the loading protocol and is detailed in Sections~\ref{sec:synthetic} and~\ref{sec:experimental}.

Problem~\eqref{eq:nlls} is solved by a \emph{selective} two-mode local search adapted to the regularity of the candidate. When the rollout loss is smooth (all primitives except the Macaulay bracket), we use L-BFGS-B~\cite{byrd1995limited,zhu1997algorithm} with finite-difference gradients. When the expression contains the Macaulay bracket $\macaulay{|A|-\sigma_Y}^r$, whose loss is non-differentiable at $|A|=\sigma_Y$ and flat when no training sample activates the bracket, finite-difference gradients vanish and the quasi-Newton iteration stalls; we then switch to a gradient-free Nelder--Mead simplex~\cite{nelder1965simplex,gao2012implementing}, which relies only on function-value comparisons. To reduce the risk of converging to a poor local minimum of the non-convex $(\hat{k},\hat{\sigma}_Y)$ landscape, the search is wrapped in a multi-restart driver and the best solution retained. Box constraints enforce the feasible set $\Theta$, with tighter data-informed bounds on the Macaulay coefficient and threshold; iteration budgets and restart counts are reported in the hyperparameter tables of Sections~\ref{sec:synthetic} and~\ref{sec:experimental}.

\paragraph{(VI) Fitness Evaluation via Trajectory Rollout.}
For each candidate dual potential $\hat{\varphi}^*(A)$ and each trajectory $(t_j, \gamma_j, \tau_j^{\mathrm{meas}})$, the fitness is evaluated by discretising the initial value problem~\eqref{eq:ivp} via the explicit Euler scheme:
\begin{equation}\label{eq:forward-euler-main}
    z_{n+1} = z_n + \Delta t \, \frac{\partial \hat{\varphi}^*}{\partial A}\bigg|_{A = A_n},
\end{equation}
where $A_n = G_1(\gamma_n - z_n)$ and $z(0) = 0$. The predicted stress at each step follows from the state law~\eqref{eq:1d-shear-state}: $\tau^{\mathrm{pred}}(t_n) = A_n$. A scalar error $\mathrm{err}_j$ is then computed by comparing $\tau_j^{\mathrm{pred}}$ with $\tau_j^{\mathrm{meas}}$; its specific form depends on the loading protocol (Sections~\ref{sec:synthetic} and~\ref{sec:experimental}).

Certain candidates can render this explicit rollout numerically stiff: power-law primitives $|A|^p$ with $1 < p < 2$ have a Jacobian $\partial\dot{z}/\partial A = cp(p-1)|A|^{p-2}$ that diverges as $A \to 0$, and exponentially growing primitives produce fast relaxation at large $|A|$. As a lightweight numerical safeguard, the integrator optionally subdivides each macro-step $[t_n, t_{n+1}]$ into $n_{\mathrm{sub}}$ inner Euler steps, with $n_{\mathrm{sub}}$ chosen per candidate from an estimate of the time constant $\tau_{\mathrm{eff}} = |A_{\max}|/\bigl(G_1\,|\dot{z}(A_{\max})|\bigr)$ at the maximal driving force $A_{\max} = G_1\max_t|\gamma(t)|$, so that fast-relaxing candidates receive proportionally finer resolution. This safeguard is enabled for the experimental rollouts of Section~\ref{sec:experimental}, whose long multi-cycle trajectories and rapidly growing discovered potentials benefit from the extra resolution; for the synthetic benchmarks of Section~\ref{sec:synthetic} the base explicit-Euler step was found sufficient, with predictions verified against a higher-order Runge--Kutta integrator.

The aggregate fitness includes a complexity penalty:
\begin{equation}\label{eq:fitness}
    J(\hat{\varphi}^*) = \frac{1}{N_s} \sum_{j=1}^{N_s} \mathrm{err}_j + \lambda_{\mathrm{cx}} \, \mathrm{nodes}(\hat{\varphi}^*),
\end{equation}
where $\lambda_{\mathrm{cx}}$ balances parsimony against accuracy; specific values are given in Sections~\ref{sec:synthetic} and~\ref{sec:experimental}. To allow unrestricted structural exploration in early generations, a linear warmup schedule is applied: during the first half of the evolution ($g \leq G/2$), the parsimony coefficient is set to zero; it then increases linearly to $\lambda_{\mathrm{cx}}$ over the remaining generations to avoid solutions with excessive complexity in the final population.

\begin{algorithm}[t]
\caption{Symbolic Regression for Dissipation Potential Discovery}
\label{alg:symbolic-regression}
\begin{algorithmic}[1]
\Require Measured trajectories $\{(t_j, \gamma_j, \tau_j^{\mathrm{meas}})\}_{j=1}^{N_s}$; known free energy $\psi$; parameters $\Gamma_{\max}, \lambda_{\mathrm{ridge}}, \lambda_{\mathrm{cx}}, f_t, f_e$
\Ensure Discovered dissipation potential $\hat{\varphi}^*(A)$
\State Initialize population $\mathcal{P}_0$ with random expressions from $\mathcal{G}_{\mathrm{cvx}}^{\mathrm{comp}}$
\State $\hat{\varphi}^*_{\mathrm{best}} \gets \mathrm{None}$
\For{generation $\gamma = 1, \ldots, \Gamma_{\max}$}
    \For{each individual $\hat{\varphi}^* \in \mathcal{P}_{\gamma-1}$}
        \State Extract tunable parameters $\boldsymbol{\theta}$ from expression tree
        \State Optimize $\boldsymbol{\theta}$ via L-BFGS-B / Nelder--Mead search minimizing $\mathcal{L}_{\mathrm{rollout}}$ \eqref{eq:nlls} \Comment{Bounds preserve convexity}
        \State Evaluate fitness $J(\hat{\varphi}^*)$~\eqref{eq:fitness} via forward Euler rollout~\eqref{eq:forward-euler-main}
    \EndFor
    \State Sort population by fitness; update $\hat{\varphi}^*_{\mathrm{best}}$ if improved
    \State Copy elite individuals (top $f_e$ fraction) to new population
    \For{remaining slots in new population}
        \State Select parent via tournament selection (size $f_t \cdot N_{\mathrm{pop}}$)
        \State Apply term-level crossover with probability $p_{\mathrm{cx}}$
        \State Apply term-level mutation with probability $p_{\mathrm{mut}}$
    \EndFor
    \State Form new population $\mathcal{P}_\gamma$
\EndFor
\State \Return Best overall individual $\hat{\varphi}^*_{\mathrm{best}}$
\end{algorithmic}
\end{algorithm}

\begin{theorem}[Invariant: Population Thermodynamic Consistency]\label{thm:invariant}
Let $\mathcal{P}_\gamma$ denote the population at generation $\gamma$. If $\mathcal{P}_0 \subset \mathcal{G}_{\mathrm{cvx}}^{\mathrm{comp}}$, then $\mathcal{P}_\gamma \subset \mathcal{G}_{\mathrm{cvx}}^{\mathrm{comp}}$ for all $\gamma \geq 0$. Consequently, every candidate in every generation satisfies properties (P1) and (P2).
\end{theorem}

\begin{proof}
See Appendix~\ref{proof:invariant} for the complete argument.
\end{proof}

\section{Synthetic Benchmarks}\label{sec:synthetic}

\subsection{Benchmark models}\label{sec:benchmarks}

We consider three ground-truth benchmark families, summarised in Tab.~\ref{tab:benchmarks}:
\begin{itemize}
\item[E1)] serves as a baseline linear viscous reference; its quadratic structure is the simplest admissible potential in the grammar:
\begin{equation}\label{eq:gt-quadratic}
    \varphi_{\text{E1}}^*(A) = \frac{1}{2\eta}A^2, \qquad \eta = 1.0.
\end{equation}
\item[E2)] probes the power-law family: $q = 2$ is an internal consistency check against E1; $q = 1.5$ is a sub-quadratic power law representative of intermediate non-Newtonian regimes; $q = 3$ tests strongly nonlinear dissipation, representative of the Payne-effect regime~\cite{payne1962dynamic,payne1963dynamic}:
\begin{equation}\label{eq:gt-abspower}
    \varphi_{\text{E2}}^*(A) = \frac{c}{q}|A|^q, \qquad c = 2.0,\; q \in \{1.5, 2.0, 3.0\}.
\end{equation}
\item[E3)] tests the extended grammar's ability to discover potentials with a genuine elastic domain:
\begin{equation}
    \varphi_{\text{E3}}^*(A) = \frac{1}{2\eta}\macaulay{|A| - \sigma_Y}^2, \qquad \eta = 1.0,\; \sigma_Y \in \{1.0, 2.5, 5.0\},
\end{equation}
that is a Bingham viscoplastic model~\cite{bingham1922fluidity}, representing a Perzyna regularization~\cite{perzyna1966fundamental} with a linear overstress function. The $\sigma_Y$ sweep probes a representative range of the elastic--plastic transition in which the driving force routinely crosses the yield threshold during loading.
\end{itemize}

Potentials E1 and E2 are usually employed to model linear and nonlinear rate-dependent responses of polymers and soft materials in conjunction with an equilibrium spring $G_\infty$: E1 yields the classical \emph{Zener (standard linear solid)} model, while E2 yields a nonlinear generalization thereof with an amplitude-dependent dashpot $\eta(A)$. Potential E3 is a classical model for viscoplasticity, where the Macaulay brackets capture the yield threshold.

The one-dimensional free energy for benchmarks E1 and E2 is chosen as
\[
\psi_{\text{E1, E2}}(\gamma, z) = \tfrac{1}{2}G_\infty \gamma^2 + \tfrac{1}{2}G_1(\gamma - z)^2,
\]
and 
\[
\psi_{\text{E3}}(\gamma, z) = \tfrac{1}{2}G_1(\gamma - z)^2.
\]
with $G_1 = 30.0$ kept fixed in all benchmarks and $G_\infty = 1.0$ for E1 and E2 ($G_\infty = 0$ for E3). The presence of an equilibrium spring in E1, E2 allows us to test the framework's ability to disentangle conservative and dissipative contributions to the response, which is a critical aspect of model discovery in complex materials. 

For the synthetic benchmarks of this section, the elastic parameters $G_\infty$ and $G_1$ are held fixed at the ground-truth values used to generate the data, so that the symbolic regression operates exclusively on the dissipation potential. For the experimental validation in Section~\ref{sec:experimental}, where the elastic moduli are not known \emph{a priori}, $G_\infty$ and $G_1$ are initialized at order-of-magnitude estimates from a preliminary inspection of the material's response and then jointly optimized with the dissipation potential within the GP loop, with bounded multiplicative factors around the initial values (see Section~\ref{sec:experimental} for details).

\begin{table}[t]
\centering
\caption{Benchmark families. All share $G_1 = 30.0$ and $\eta = c/2 = 1.0$; $G_\infty = 1.0$ for E1 and E2, $G_\infty = 0$ for E3.}
\label{tab:benchmarks}
\resizebox{\textwidth}{!}{
\begin{tabular}{@{}llcccl@{}}
\toprule
Tag & & Free Energy &Dissipation Potential & Varied Parameters & Physical target \\
&  &$\psi$ &$\varphi^\ast$ & &\\
\midrule
E1 & \resizebox{0.15\textwidth}{!}{
\begin{tikzpicture}[scale=0.9, every node/.style={transform shape}]
  \def\xJ{0.3}
  \def\xR{2.9}
  \def\xL{0.6}
  \def\xEnd{2.6}
  \def\xMid{1.6}
  \def\yM{1}
  \def\yT{1.7}
  \def\yB{0.3}
  \def\labGap{0.22}

  \fill (0,\yM) circle (2pt);
  \draw (0,\yM) -- (\xJ,\yM);
  \fill (\xJ,\yM) circle (2pt);

  \draw (\xJ,\yM) -- (\xJ,\yT) -- (\xL,\yT);
  \draw (\xJ,\yM) -- (\xJ,\yB) -- (\xL,\yB);

  \coordinate (sInfL) at (1.15,\yT);
  \coordinate (sInfR) at (2.05,\yT);
  \draw (\xL,\yT) -- (sInfL);
  \draw (sInfL) -- ($(sInfL)+(0.08,0)$) -- ($(sInfL)+(0.16,0.2)$)
        -- ($(sInfL)+(0.24,-0.2)$) -- ($(sInfL)+(0.32,0.2)$)
        -- ($(sInfL)+(0.40,-0.2)$) -- ($(sInfL)+(0.48,0.2)$)
        -- ($(sInfL)+(0.56,-0.2)$) -- ($(sInfL)+(0.64,0.2)$)
        -- ($(sInfL)+(0.72,-0.2)$) -- ($(sInfL)+(0.80,0.2)$)
        -- ($(sInfL)+(0.88,-0.2)$) -- ($(sInfL)+(0.90,0)$) -- (sInfR);
  \draw (sInfR) -- (\xEnd,\yT);
  \node at ($(sInfL)!0.5!(sInfR)$) [above=\labGap] {$G_{\infty}$};

  \coordinate (s1L) at (0.95,\yB);
  \coordinate (s1R) at (1.55,\yB);
  \draw (\xL,\yB) -- (s1L);
  \draw (s1L) -- ($(s1L)+(0.08,0)$) -- ($(s1L)+(0.16,0.2)$)
        -- ($(s1L)+(0.24,-0.2)$) -- ($(s1L)+(0.32,0.2)$)
        -- ($(s1L)+(0.40,-0.2)$) -- ($(s1L)+(0.48,0.2)$)
        -- ($(s1L)+(0.56,-0.2)$) -- (s1R);
  \draw (s1R) -- (\xMid,\yB);
  \node at ($(s1L)!0.5!(s1R)$) [above=\labGap] {$G_1$};

  \coordinate (etaL) at (1.72,\yB);
  \coordinate (etaCupL) at (1.78,0.0);
  \coordinate (etaCupR) at (2.12,0.6);
  \coordinate (etaRod) at (2.28,\yB);
  \draw (\xMid,\yB) -- (etaL);
  \fill[gray!40] (etaCupL) -- ($(etaCupL)+(0.34,0)$) -- ($(etaCupL)+(0.34,0.6)$)
        -- ($(etaCupL)+(0,0.6)$) -- cycle;
  \draw (etaCupL) -- ($(etaCupL)+(0,0.6)$);
  \draw (etaCupL) -- ($(etaCupL)+(0.34,0)$);
  \draw ($(etaCupL)+(0,0.6)$) -- ($(etaCupL)+(0.34,0.6)$);
  \draw ($(etaCupL)+(0.27,0.1)$) -- ($(etaCupL)+(0.27,0.5)$);
  \draw ($(etaCupL)+(0.27,0.3)$) -- (etaRod);
  \draw (etaRod) -- (\xEnd,\yB);
  \node at ($(etaCupL)!0.5!(etaCupR)$) [above=\labGap] {$\eta$};

  \draw (\xEnd,\yT) -- (\xR,\yT) -- (\xR,\yM);
  \draw (\xEnd,\yB) -- (\xR,\yB) -- (\xR,\yM);
  \fill (\xR,\yM) circle (2pt);
  \draw (\xR,\yM) -- (3.4,\yM);
  \fill (3.4,\yM) circle (2pt);
\end{tikzpicture}}& $\frac{1}{2}G_\infty \gamma^2 + \frac{1}{2}G_1(\gamma - z)^2$ &$\dfrac{1}{2\eta}A^2$ & -- & Linear viscoelasticity \\[8pt]
E2 & \resizebox{0.15\textwidth}{!}{
\begin{tikzpicture}[scale=0.9, every node/.style={transform shape}]
  \def\xJ{0.3}
  \def\xR{2.9}
  \def\xL{0.6}
  \def\xEnd{2.6}
  \def\xMid{1.6}
  \def\yM{1}
  \def\yT{1.7}
  \def\yB{0.3}
  \def\labGap{0.22}

  \fill (0,\yM) circle (2pt);
  \draw (0,\yM) -- (\xJ,\yM);
  \fill (\xJ,\yM) circle (2pt);

  \draw (\xJ,\yM) -- (\xJ,\yT) -- (\xL,\yT);
  \draw (\xJ,\yM) -- (\xJ,\yB) -- (\xL,\yB);

  \coordinate (sInfL) at (1.15,\yT);
  \coordinate (sInfR) at (2.05,\yT);
  \draw (\xL,\yT) -- (sInfL);
  \draw (sInfL) -- ($(sInfL)+(0.08,0)$) -- ($(sInfL)+(0.16,0.2)$)
        -- ($(sInfL)+(0.24,-0.2)$) -- ($(sInfL)+(0.32,0.2)$)
        -- ($(sInfL)+(0.40,-0.2)$) -- ($(sInfL)+(0.48,0.2)$)
        -- ($(sInfL)+(0.56,-0.2)$) -- ($(sInfL)+(0.64,0.2)$)
        -- ($(sInfL)+(0.72,-0.2)$) -- ($(sInfL)+(0.80,0.2)$)
        -- ($(sInfL)+(0.88,-0.2)$) -- ($(sInfL)+(0.90,0)$) -- (sInfR);
  \draw (sInfR) -- (\xEnd,\yT);
  \node at ($(sInfL)!0.5!(sInfR)$) [above=\labGap] {$G_{\infty}$};

  \coordinate (s1L) at (0.95,\yB);
  \coordinate (s1R) at (1.55,\yB);
  \draw (\xL,\yB) -- (s1L);
  \draw (s1L) -- ($(s1L)+(0.08,0)$) -- ($(s1L)+(0.16,0.2)$)
        -- ($(s1L)+(0.24,-0.2)$) -- ($(s1L)+(0.32,0.2)$)
        -- ($(s1L)+(0.40,-0.2)$) -- ($(s1L)+(0.48,0.2)$)
        -- ($(s1L)+(0.56,-0.2)$) -- (s1R);
  \draw (s1R) -- (\xMid,\yB);
  \node at ($(s1L)!0.5!(s1R)$) [above=\labGap] {$G_1$};

  \coordinate (etaL) at (1.72,\yB);
  \coordinate (etaCupL) at (1.78,0.0);
  \coordinate (etaCupR) at (2.12,0.6);
  \coordinate (etaRod) at (2.28,\yB);
  \draw (\xMid,\yB) -- (etaL);
  \fill[gray!40] (etaCupL) -- ($(etaCupL)+(0.34,0)$) -- ($(etaCupL)+(0.34,0.6)$)
        -- ($(etaCupL)+(0,0.6)$) -- cycle;
  \draw (etaCupL) -- ($(etaCupL)+(0,0.6)$);
  \draw (etaCupL) -- ($(etaCupL)+(0.34,0)$);
  \draw ($(etaCupL)+(0,0.6)$) -- ($(etaCupL)+(0.34,0.6)$);
  \draw ($(etaCupL)+(0.27,0.1)$) -- ($(etaCupL)+(0.27,0.5)$);
  \draw ($(etaCupL)+(0.27,0.3)$) -- (etaRod);
  \draw (etaRod) -- (\xEnd,\yB);
  \draw[-{Stealth[length=2mm]}, thick]
        ($(etaCupL)+(-0.10,-0.10)$) -- ($(etaCupL)+(0.44,0.70)$);
  \node at ($(etaCupL)!0.5!(etaCupR)$) [above=\labGap] {$\eta(A)$};

  \draw (\xEnd,\yT) -- (\xR,\yT) -- (\xR,\yM);
  \draw (\xEnd,\yB) -- (\xR,\yB) -- (\xR,\yM);
  \fill (\xR,\yM) circle (2pt);
  \draw (\xR,\yM) -- (3.4,\yM);
  \fill (3.4,\yM) circle (2pt);
\end{tikzpicture}}& $\frac{1}{2}G_\infty \gamma^2 + \frac{1}{2}G_1(\gamma - z)^2$ &$\dfrac{c}{q}|A|^q$ & $q \in \{1.5, 2.0, 3.0\}$ & Power-law viscoelasticity \\[8pt]
E3 & \resizebox{0.15\textwidth}{!}{
\begin{tikzpicture}[scale=0.9, every node/.style={transform shape}]
  \fill (0,1) circle (2pt);
  \draw (0,1) -- (0.5,1);
  \draw (0.5,1) -- (0.6,1) -- (0.7,1.2) -- (0.8,0.8) -- (0.9,1.2)
        -- (1.0,0.8) -- (1.1,1.2) -- (1.2,0.8) -- (1.3,1) -- (1.6,1);
  \node at (0.9,1.5) {$G_1$};

  \draw (1.6,1) -- (1.8,1);
  \draw (1.8,0.8) -- (2.2,0.8);
  \draw (1.8,1.2) -- (2.2,1.2);
  \draw (1.9,0.85) rectangle (2.1,1.15);
  \draw (2.2,1) -- (2.6,1);
  \node at (2.0,1.5) {$\sigma_Y$};

  \draw (2.6,1) -- (3.0,1);
  \fill (3.0,1) circle (2pt);
\end{tikzpicture}}& $\frac{1}{2}G_1(\gamma - z)^2$ &$\dfrac{1}{2\eta}\macaulay{|A| - \sigma_Y}^2$ & $\sigma_Y \in \{1.0, 2.5, 5.0\}$ & Bingham viscoplasticity \\
\bottomrule
\end{tabular}
}
\end{table}

\subsection{Numerical setup}\label{sec:setup}

Two strain-controlled loading protocols are used depending on the benchmark family. For E1 and E2 we adopt a sinusoidal DMA protocol (LP1) with each dataset consists of $S=15$ strain histories
\begin{equation}\label{eq:dma-strain}
    \varepsilon(t)=\mathcal{A}\sin(2\pi f t), \qquad \mathcal{A}\in\{0.01,0.03,0.05\}, \quad f\in\{0.1,1,10,20,50\}\ \mathrm{Hz},
\end{equation}
covering the full Cartesian product of amplitudes and frequencies. For E3, sinusoidal loading is inadequate because the time-varying strain rate $\dot\gamma = 2\pi f\mathcal{A}\cos(2\pi ft)$ couples to the Bingham flow rule and smears the yield transition into a rounded, rate-modulated loop, masking the kink that identifies $\sigma_Y$. A constant strain-rate protocol drives $|A|$ to grow linearly in time and produces a sharp, identifiable kink at $|A|=\sigma_Y$. We therefore use a triangular wave protocol (LP2)
\begin{equation}\label{eq:ramp-strain}
    \varepsilon(t) = \frac{2\mathcal{A}}{\pi} \arcsin\left( \sin\left( \frac{\pi \dot\gamma_0}{2\mathcal{A}} t \right) \right), \qquad \dot\gamma_0 \in \{0.1, 1, 10, 50\}\ \mathrm{s}^{-1}, \quad \mathcal{A} \in \left\{\frac{\sigma_Y}{2G_1},\; \frac{2\sigma_Y}{G_1},\; \frac{10\sigma_Y}{G_1}\right\},
\end{equation}
The three amplitudes are scaled so that the peak driving force $A_{\max} = G_1\mathcal{A}$ takes the values $\{\sigma_Y/2,\, 2\sigma_Y,\, 10\sigma_Y\}$, probing in order the \emph{sub-yield} regime (no flow), the \emph{near-yield} transition (kink active on a finite portion of the ramp), and the \emph{deep supra-yield} regime (plastic flow dominates). The four strain-rate levels $\dot\gamma_0$ further probe the rate dependence within each amplitude regime. 
For all benchmarks stress trajectories are obtained by integrating the IVP \eqref{eq:tuning-evolution} with initial condition $A(0)=0$. Synthetic data is generated via RK4 with $N_{\mathrm{ppc}}^{\mathrm{RK4}}=64$ points per cycle; fitness evaluation on the other hand uses explicit Euler with $N_{\mathrm{ppc}}^{\mathrm{Euler}}=128$ points per cycle, deliberately mismatching the data-generation integrator to avoid the inverse crime~\cite{colton2013inverse,wirgin2004inverse}. 
A stability cap $\Delta t \leq \tau/d$ (where $\tau=\eta/G_1$ is the relaxation time, $d=2$ for RK4 and $d=4$ for Euler) is enforced at low frequencies.

The Monte Carlo~\cite{metropolis1949monte} noise design is applied to two benchmarks only: E1 (linear viscoelasticity) and the E2 baseline at $q^\star = 3$. The dynamics are contaminated by two independent noise sources: an Ornstein--Uhlenbeck process noise~\cite{uhlenbeck1930theory} on the internal-variable evolution,
\begin{equation}
    \dot{z} = \frac{\partial\varphi^*}{\partial A}(A) + \xi, \qquad \mathrm{d}\xi = -\lambda \xi\,\mathrm{d}t + \sigma_{\mathrm{OU}}\sqrt{2\lambda}\,\mathrm{d}W,
    \label{eq:ou-process}
\end{equation}
with $\lambda = 1\,\mathrm{s}^{-1}$ and $\sigma_{\mathrm{OU}}\in\{0,0.002,0.005\}$, and additive Gaussian measurement noise,
\begin{equation}
    \tau^{\mathrm{meas}}(t) = \tau^{\mathrm{true}}(t) + \epsilon_m(t), \qquad \epsilon_m \sim \mathcal{N}(0,\sigma_m^2),
    \label{eq:measurement-noise}
\end{equation}
with $\sigma_m\in\{0,0.005,0.010\}$. The full $3\times 3=9$ noise grid is run with $N_{\mathrm{trials}}=5$ independent seeds, giving 45 runs per ground-truth potential. The E2 exponent sweep ($q^\star\in\{1.5,2.0,3.0\}$, Section~\ref{sec:results}) and benchmark E3 use noise-free data only, the former to isolate the structural exponent recovery and the latter the yield-identification challenge.

The fitness function for E1 and E2 is the mean relative error over all $S$ loading conditions, based on the predicted and target storage and loss moduli $G'(\omega)$ and $G''(\omega)$\footnote{Algebraically, for a steady-state harmonic strain $\varepsilon(t) = \gamma_0 \sin(\omega t)$, these moduli are defined via the first Fourier harmonic of the stress response: $\sigma(t) \approx \gamma_0 [G'(\omega) \sin(\omega t) + G''(\omega) \cos(\omega t)]$.}:
\begin{equation}\label{eq:gp-rel-err}
    e_{G'} = \frac{1}{S}\sum_{j=1}^{S} \frac{|{G'}_j^{\mathrm{target}} - {G'}_j^{\mathrm{pred}}|}{|{G'}_j^{\mathrm{target}}|},\qquad
    e_{G''} = \frac{1}{S}\sum_{j=1}^{S} \frac{|{G''}_j^{\mathrm{target}} - {G''}_j^{\mathrm{pred}}|}{|{G''}_j^{\mathrm{target}}|},
\end{equation}
combined as
\begin{equation}\label{eq:fitness-dma}
    J(\hat{\varphi}^*) = \frac{e_{G'} + e_{G''}}{2} + \lambda_{\mathrm{cx}}\mathrm{nodes}(\hat{\varphi}^*).
\end{equation}

For E3, the first Fourier harmonic discards the yield-transition signature, so the fitness is instead based on the time-domain stress normalized root-mean-square error (nRMSE) computed over the steady-state ramp cycle:
\begin{equation}\label{eq:fitness-ramp}
    e_{\tau,j} = \sqrt{ \frac{\sum_{i} \big(\tau_j^{\mathrm{meas}}(t_i) - \tau_j^{\mathrm{pred}}(t_i)\big)^2}{\sum_{i} \big(\tau_j^{\mathrm{meas}}(t_i)\big)^2} },
\end{equation}
with the aggregate fitness given by
\begin{equation}\label{eq:fitness-nrmse}
    J(\hat{\varphi}^*) = \frac{1}{S}\sum_{j=1}^{S} e_{\tau,j} + \lambda_{\mathrm{cx}}\mathrm{nodes}\,(\hat{\varphi}^*).
\end{equation}
This is a time-domain stress nRMSE closely related to the metric used in Section~\ref{sec:experimental} for the experimental data, the differences being the integration window (steady-state ramp cycle here, full raw history with initial transient there). The symbolic regression settings are collected in Tab.~\ref{tab:hyperparams} for both the synthetic benchmarks and the experimental validation.

\begin{table}[t]
\centering
\caption{GP hyperparameters for synthetic benchmarks and experimental validation. Both columns adopt the \emph{selective Nelder--Mead} tuner with $N_{\mathrm{restart}}=3$ for expressions that contain the Macaulay bracket primitive (Section~\ref{sec:optimization}); for all other primitives the local search reduces to single-shot L-BFGS-B with the iteration budget reported below.}
\label{tab:hyperparams}
\begin{tabular}{@{}lcc@{}}
\toprule
Parameter & Synthetic & Experimental \\
\midrule
Population size $N_{\mathrm{pop}}$ & 50 & 60 \\
Generations $G$ & 5 & 5 \\
Max depth $D_{\max}$ & 3 & 6 \\
Max nodes $N_{\max}$ & 12 & 20 \\
Tournament fraction $f_t$ & 0.05 & 0.05 \\
Elite fraction $f_e$ & 0.10 & 0.10 \\
Crossover probability $p_{\mathrm{cx}}$ & 0.8 & 0.8 \\
Mutation probability $p_{\mathrm{mut}}$ & 0.5 & 0.5 \\
Local search iterations $N_{\mathrm{iter}}^{\max}$ & 120 & 200 \\
Multi-restart count $N_{\mathrm{restart}}$ (Macaulay only) & 3 & 3 \\
Intra-block sum probability $p_{\Sigma}$ & 0.05 & 0.05 \\
Ridge regularization $\lambda_{\mathrm{ridge}}$ & $10^{-6}$ & $10^{-6}$ \\
Parsimony coefficient $\lambda_{\mathrm{cx}}$ & $10^{-2}$ & 0 \\
Per-evaluation timeout $T_{\mathrm{eval}}$ & $300\,\mathrm{s}$ & $540\,\mathrm{s}$ \\
\bottomrule
\end{tabular}
\end{table}

\subsection{Results}\label{sec:results}

\paragraph{E1: Newtonian viscosity.} 
Across the 45 runs of the quadratic ground truth \eqref{eq:gt-quadratic}, the framework recovers the exact quadratic structure
\[
    \hat{\varphi}^*(A)=\hat{k}A^2
\]
in 30/45 trials (66.7\%). A further 14/45 trials converge to $\hat{\varphi}^*(A)=\hat{k}|A|^{\hat{p}}$ with $\hat{p}\approx 2.00$, a near-equivalent form in the sense of Section~\ref{par:symbolic-equivalence}; counting these as successful recoveries yields 44/45 (97.8\%) across all noise configurations. The single remaining trial, in the highest-noise cell, converges to a hyperbolic candidate of the form $\hat k\bigl(\cosh(\alpha A)-1\bigr)$ with $\hat k=5.0$, $\alpha=0.446$, admissible by construction.

For the exact quadratic recoveries, the fitted prefactor is
\[
    \hat{k} = 0.4958 \pm 0.0021 \quad (\min 0.4923,\ \max 0.4993),
\]
close to the ground truth $k^\star=1/(2\eta)=0.5$, corresponding to a relative error of $\approx 1\%$. For the 14 near-quadratic recoveries, the fitted parameters are $\hat{k}=0.4956\pm 0.0013$ and $\hat{p}=1.997\pm 0.011$, confirming functional equivalence. Over all 45 runs, the DMA moduli relative errors are $e_{G'}=0.0451 \pm 0.0189$ and $e_{G''}=0.0375 \pm 0.0193$.

\begin{table}[H]
\centering
\caption{Newtonian viscosity benchmark (E1), $G_\infty=1.0$, $G_1=30.0$ (45 runs). ``Exact recovery'' denotes trials whose discovered expression simplifies to $\hat{k}A^2$ (integer exponent locked by the grammar/simplifier, $\hat{p}\equiv 2$); ``Near-exact recovery'' additionally includes trials of the form $\hat{k}|A|^{\hat{p}}$ with $|\hat{p}-2|<0.05$. Entries report mean $\pm$ std over $N_{\mathrm{trials}}=5$ runs per noise configuration. $\sigma_m$ is the measurement noise amplitude, $\sigma_{\mathrm{OU}}$ is the process noise amplitude in Eqs.~\eqref{eq:ou-process}-\eqref{eq:measurement-noise}. $e_{G'}$ and $e_{G''}$ are the mean relative errors on storage and loss moduli defined in \eqref{eq:gp-rel-err}.}
\label{tab:mc-quadratic}
\begin{tabular}{@{}cccccc@{}}
\toprule
$\sigma_m$ & $\sigma_{\mathrm{OU}}$ & Exact (\%) & Near-exact (\%) & $e_{G'}$ & $e_{G''}$ \\
\midrule
0.000 & 0.000 &  60.0 & 100.0 & 0.0325 $\pm$ 0.0002 & 0.0170 $\pm$ 0.0009 \\
0.000 & 0.002 &  80.0 & 100.0 & 0.0407 $\pm$ 0.0103 & 0.0334 $\pm$ 0.0047 \\
0.000 & 0.005 &  40.0 & 100.0 & 0.0615 $\pm$ 0.0229 & 0.0603 $\pm$ 0.0200 \\
0.005 & 0.000 &  80.0 & 100.0 & 0.0328 $\pm$ 0.0026 & 0.0217 $\pm$ 0.0035 \\
0.005 & 0.002 &  40.0 & 100.0 & 0.0364 $\pm$ 0.0032 & 0.0348 $\pm$ 0.0091 \\
0.005 & 0.005 & 100.0 & 100.0 & 0.0591 $\pm$ 0.0079 & 0.0391 $\pm$ 0.0055 \\
0.010 & 0.000 &  40.0 & 100.0 & 0.0352 $\pm$ 0.0033 & 0.0262 $\pm$ 0.0049 \\
0.010 & 0.002 & 100.0 & 100.0 & 0.0366 $\pm$ 0.0054 & 0.0378 $\pm$ 0.0067 \\
0.010 & 0.005 &  60.0 &  80.0 & 0.0708 $\pm$ 0.0335 & 0.0674 $\pm$ 0.0281 \\
\bottomrule
\end{tabular}
\end{table}

\begin{figure}[ht]
\centering
\includegraphics[width=\textwidth]{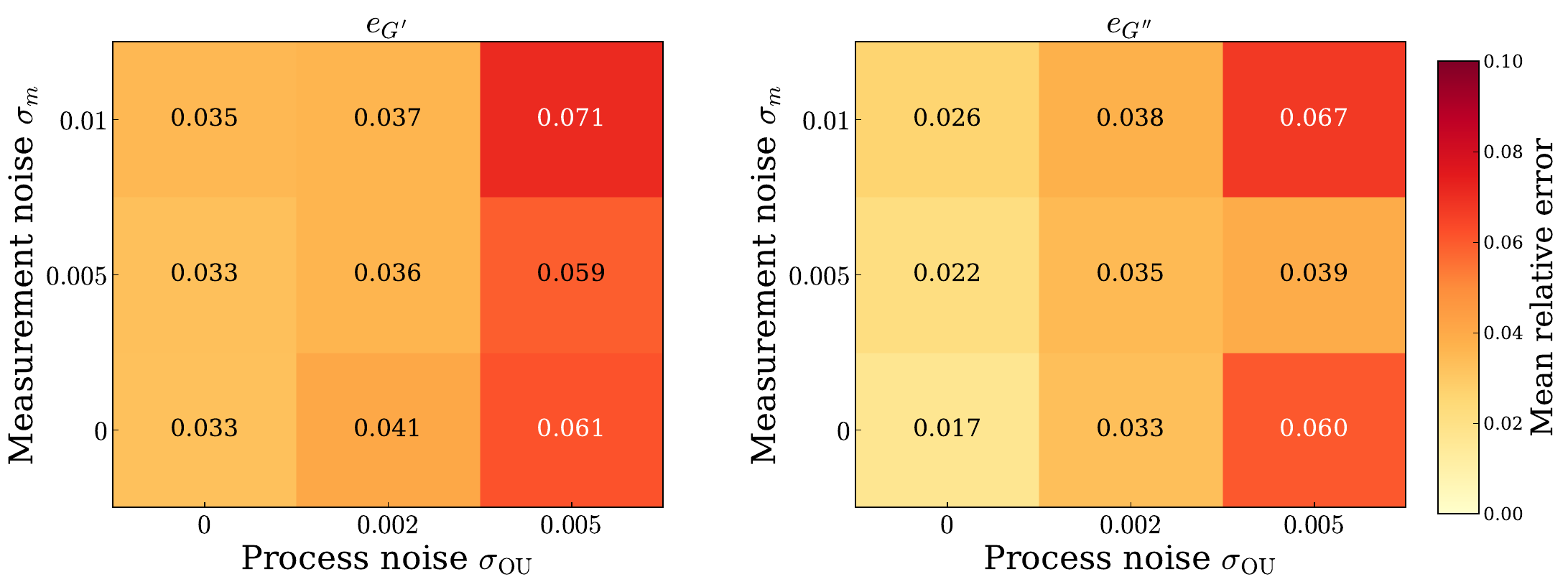}
\caption{Newtonian viscosity benchmark (E1): heatmaps of mean DMA relative error across the noise grid, averaged over $N_{\mathrm{trials}}=5$ independent trials per configuration. \emph{Left}: storage modulus $e_{G'}$ defined in~\eqref{eq:gp-rel-err}. \emph{Right}: loss modulus $e_{G''}$. Both errors increase with process noise $\sigma_{\mathrm{OU}}$ (horizontal axis), while measurement noise $\sigma_m$ (vertical axis) has a milder effect.}
\label{fig:heatmap-quadratic}
\end{figure}

To assess the time-domain predictive accuracy of the discovered potentials beyond scalar error metrics, Fig.~\ref{fig:hysteresis-quadratic} displays representative strain--stress hysteresis loops for the Newtonian viscosity benchmark. Each figure shows a $2\times 3$ grid spanning two noise levels (rows: noise-free and high noise) and three amplitude--frequency combinations (columns: $\mathcal{A}=0.01$, $f=1\,$Hz; $\mathcal{A}=0.03$, $f=10\,$Hz; $\mathcal{A}=0.05$, $f=50\,$Hz). As predicted by the theoretical framework, quadratic viscosity produces a linear response in the driving force, resulting in elliptical hysteresis loops in the strain-stress plane. Moreover, the storage and loss moduli remain independent of the loading amplitude, as shown in Fig.~\ref{fig:dma-moduli-quadratic}.

\begin{figure}[H]
\centering
\includegraphics[width=\textwidth]{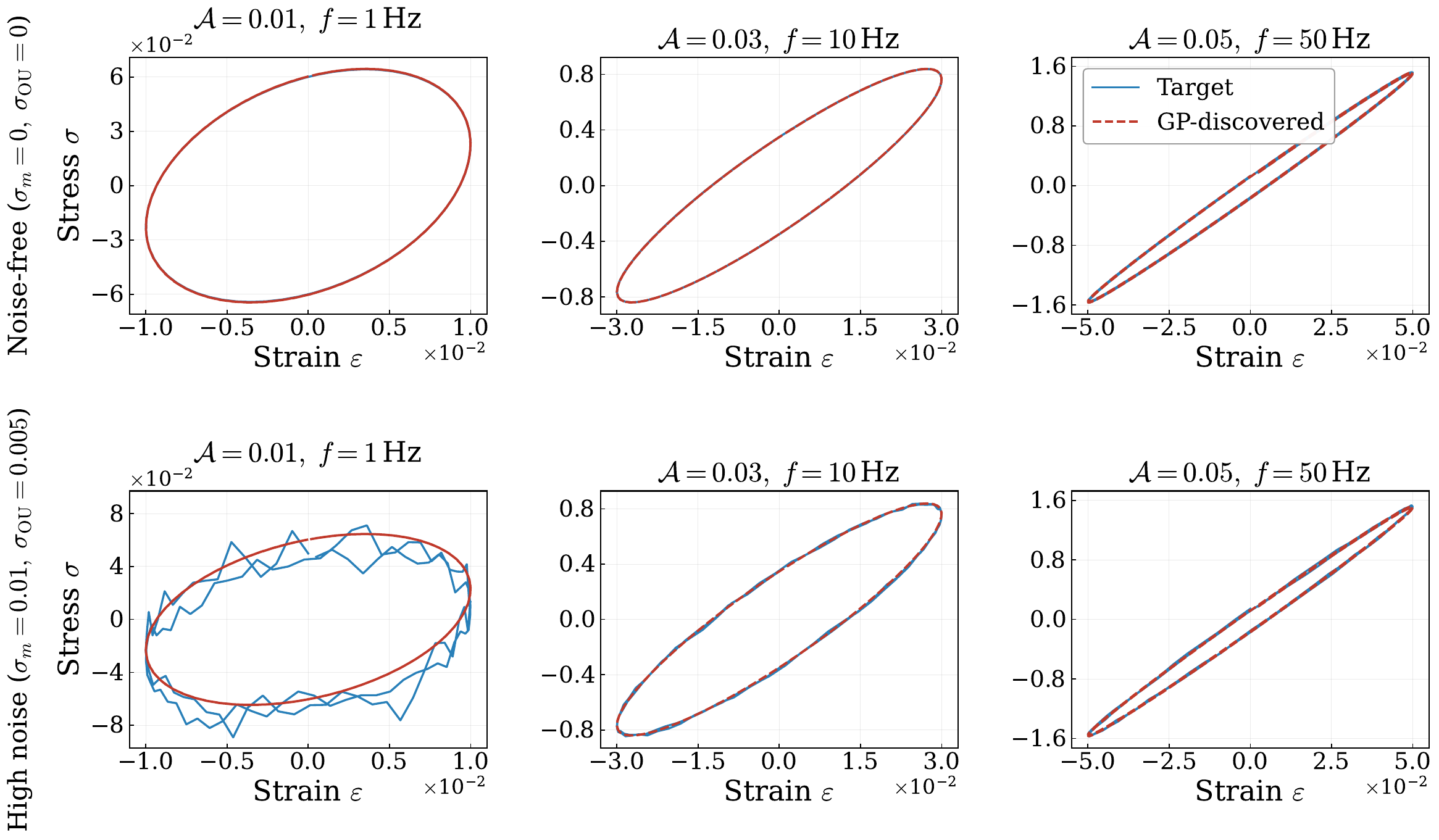}
\caption{Newtonian viscosity benchmark (E1): strain--stress hysteresis loops for representative amplitude--frequency combinations, shown for trial~1. Top row: noise-free, discovered $\hat{\varphi}^*(A)=0.4980\,A^2$. Bottom row: high noise ($\sigma_m=0.010$, $\sigma_{\mathrm{OU}}=0.005$), discovered $\hat{\varphi}^*(A)=0.4905\,A^2$. Ground truth: $0.5\,A^2$.}
\label{fig:hysteresis-quadratic}
\end{figure}

For the quadratic noise-free case, the predicted moduli overlap with the ground truth everywhere, consistent with the low relative errors reported in Tab.~\ref{tab:mc-quadratic}. 

\begin{figure}[H]
\centering
\includegraphics[width=0.9\textwidth]{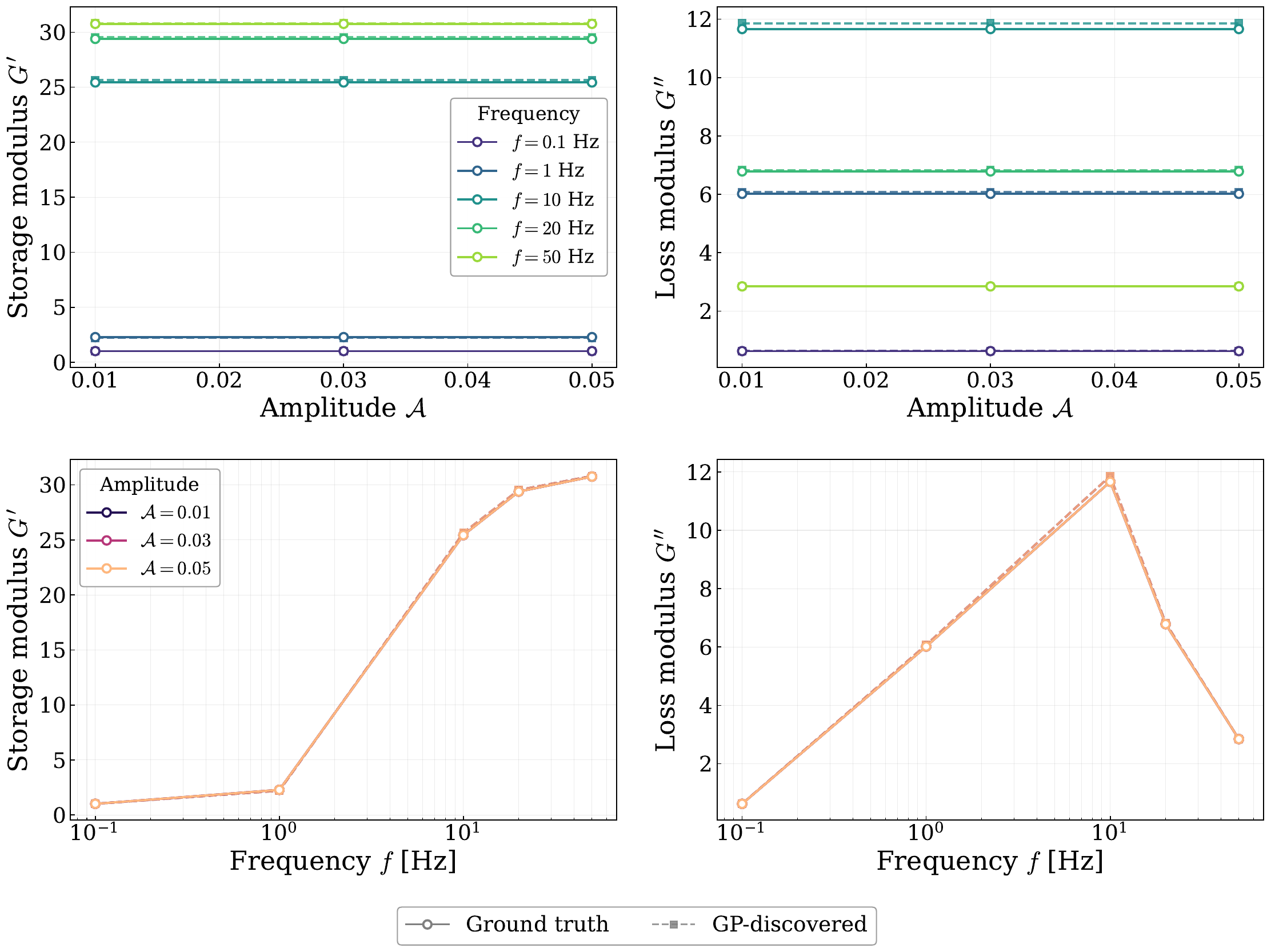}
\caption{DMA moduli comparison for the Newtonian viscosity benchmark (E1), noise-free case, trial~1. $G'$ and $G''$ versus amplitude (top row) and versus frequency on logarithmic scale (bottom row). Both moduli are amplitude-independent, as expected for Newtonian viscosity, and the predicted curves overlap the ground truth everywhere. Solid circles: ground truth; dashed squares: predicted.}
\label{fig:dma-moduli-quadratic}
\end{figure}

\paragraph{E2: Power-law family.}
We begin testing the E2 benchmark, by investigating the recovery performance of the framework in the case of $q=3$. The dominant recovered form is the exact single power-law primitive
\[
    \hat{\varphi}^*(A) = \hat{k}|A|^{\hat{q}},
\]
obtained in 42/45 trials (93.3\%). Across these successful runs we obtain
\[
    \hat{k} = 0.6564 \pm 0.0038, \qquad \hat{q} = 2.986 \pm 0.027,
\]
to be compared with the ground truth $k^\star=c/q=2/3\approx 0.6667$ and $q^\star=3$. The 3 remaining trials, all occurring in cells with non-zero process noise, converge to alternative structures that belong to the admissible grammar and are functionally near-equivalent to the ground truth in the probed regime. These include:
\begin{itemize}
    \item \emph{Mixed power-law + hyperbolic}: e.g., $\log\cosh(0.499\,A) + 0.514\,|A|^{3.44}$, or $\log\cosh(0.441\,A) + 0.549\,|A|^{3.31}$;
    \item \emph{Exponential--quadratic compositions}: e.g., $0.100\,[\exp(2.043\,A^2) - 1]$.
\end{itemize}
These alternative forms remain thermodynamically admissible by construction and achieve comparable DMA errors, illustrating that the grammar may discover different but physically equivalent representations of the same dissipative behavior, particularly under higher noise.

Over all 45 runs, the DMA moduli errors are $e_{G'}=0.0786 \pm 0.1024$ and $e_{G''}=0.0427 \pm 0.0319$. The higher variance in $e_{G'}$ is driven by the non-exact recoveries under the largest noise combinations (Tab.~\ref{tab:mc-abspower}).

\begin{table}[t]
\centering
\caption{Abs-power benchmark (E2), $G_\infty=1.0$, $G_1=30.0$ (45 runs). ``Exact recovery'' denotes recovery of $\hat{\varphi}^*(A)=\hat{k}|A|^{\hat{q}}$ with $|\hat{q}-3|<0.15$. The remaining trials converge to alternative admissible structures (mixed terms, hyperbolic compositions). Entries report mean $\pm$ std over $N_{\mathrm{trials}}=5$.}
\label{tab:mc-abspower}
\begin{tabular}{@{}ccccc@{}}
\toprule
$\sigma_m$ & $\sigma_{\mathrm{OU}}$ & Exact (\%) & $e_{G'}$ & $e_{G''}$ \\
\midrule
0.000 & 0.000 & 100.0 & 0.0180 $\pm$ 0.0003 & 0.0237 $\pm$ 0.0002 \\
0.000 & 0.002 & 100.0 & 0.0620 $\pm$ 0.0655 & 0.0309 $\pm$ 0.0053 \\
0.000 & 0.005 & 100.0 & 0.1689 $\pm$ 0.0917 & 0.0427 $\pm$ 0.0147 \\
0.005 & 0.000 & 100.0 & 0.0176 $\pm$ 0.0008 & 0.0253 $\pm$ 0.0009 \\
0.005 & 0.002 & 100.0 & 0.0417 $\pm$ 0.0169 & 0.0336 $\pm$ 0.0051 \\
0.005 & 0.005 &  80.0 & 0.1700 $\pm$ 0.2171 & 0.0592 $\pm$ 0.0325 \\
0.010 & 0.000 & 100.0 & 0.0192 $\pm$ 0.0020 & 0.0319 $\pm$ 0.0070 \\
0.010 & 0.002 &  80.0 & 0.0637 $\pm$ 0.0578 & 0.0695 $\pm$ 0.0778 \\
0.010 & 0.005 &  80.0 & 0.1461 $\pm$ 0.1099 & 0.0672 $\pm$ 0.0283 \\
\bottomrule
\end{tabular}
\end{table}

\begin{figure}[H]
\centering
\includegraphics[width=\textwidth]{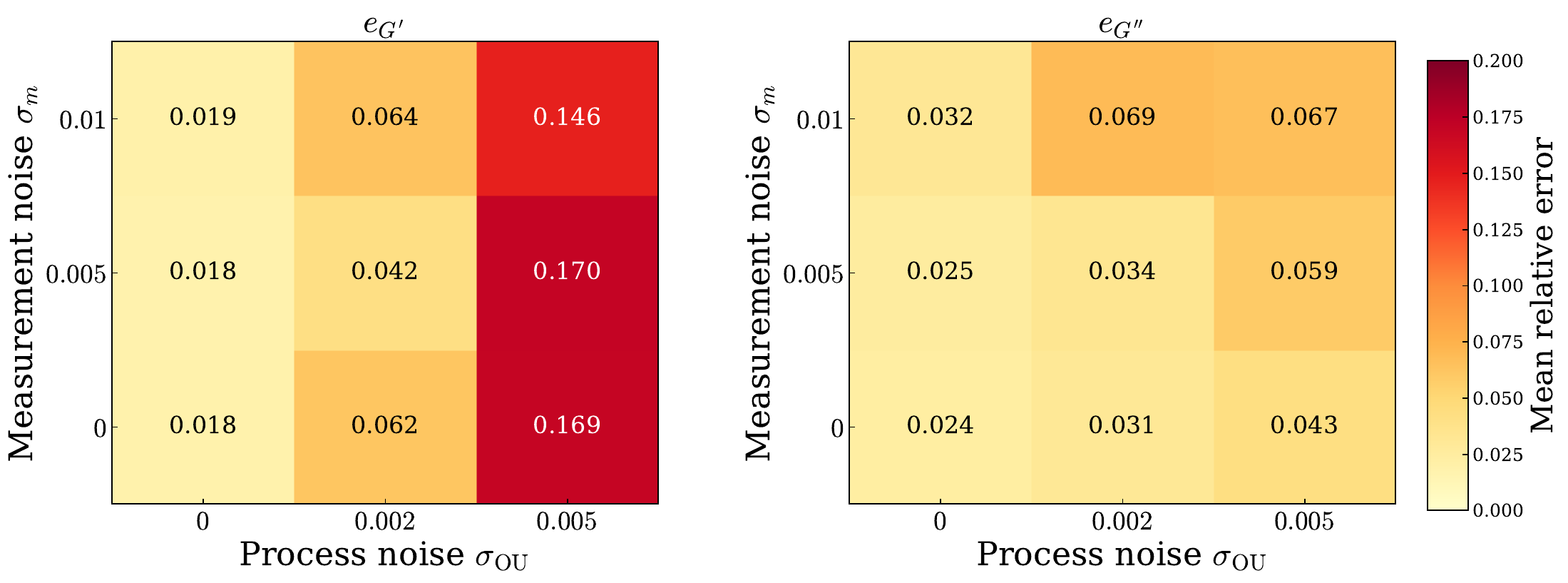}
\caption{Abs-power benchmark (E2): heatmaps of mean DMA relative error across the noise grid, averaged over $N_{\mathrm{trials}}=5$ independent trials per configuration. \emph{Left}: storage modulus $e_{G'}$. \emph{Right}: loss modulus $e_{G''}$. The $G'$ error is more sensitive to process noise than $G''$, reflecting the stronger impact of OU fluctuations on the elastic response reconstruction.}
\label{fig:heatmap-abspower}
\end{figure}

The corresponding stress-strain diagrams for different noise levels and loading conditions are shown in Fig.~\ref{fig:hysteresis-abspower}. The characteristic nonlinear hysteresis loops, whose shape deviates progressively from the elliptical Newtonian response at larger amplitudes, are indeed evident from the figure.

\begin{figure}[H]
\centering
\includegraphics[width=\textwidth]{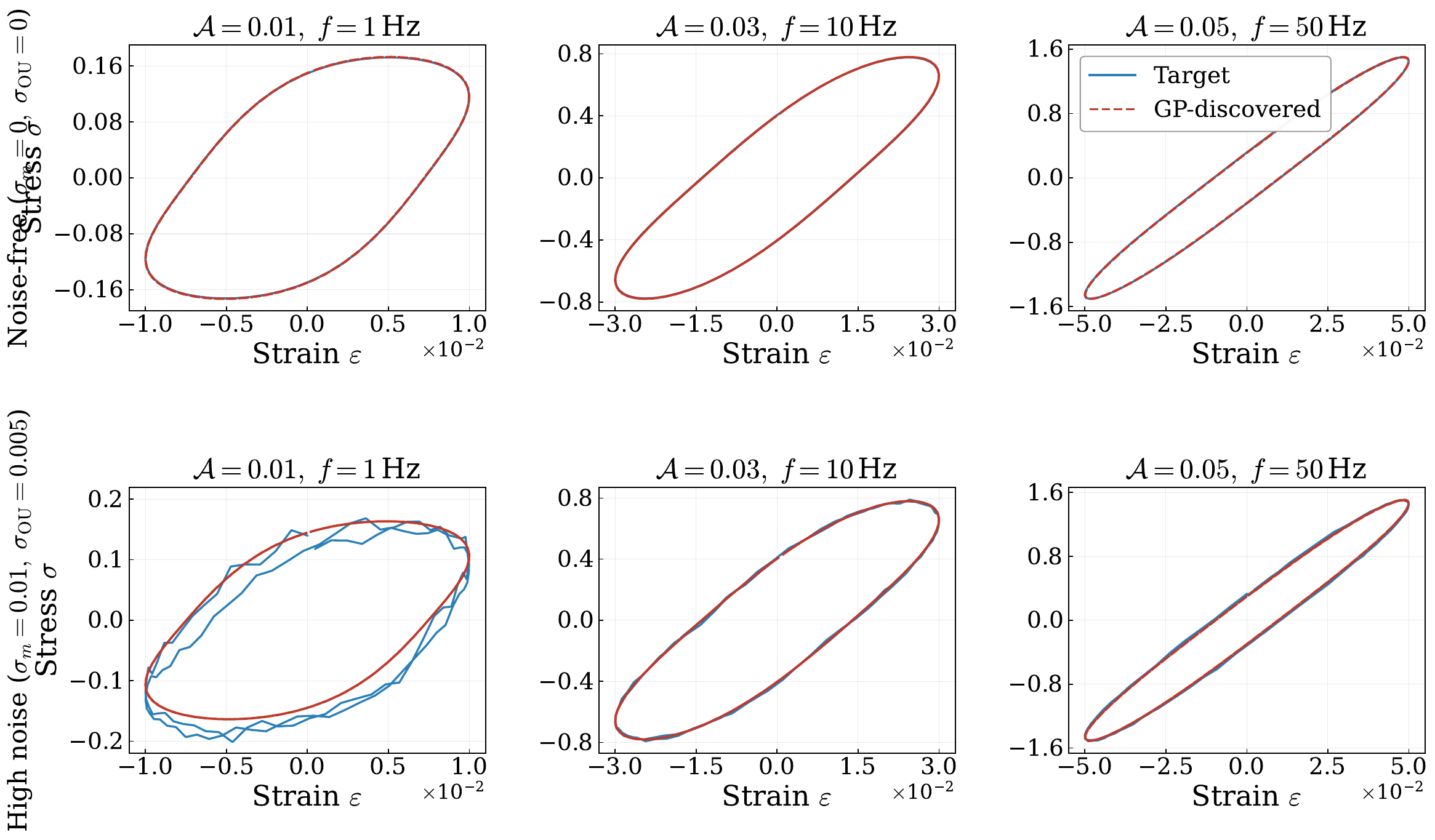}
\caption{Abs-power benchmark (E2): strain--stress hysteresis loops under noise-free (top row) and high-noise (bottom row) conditions, shown for trial~1. Top row: noise-free, discovered $\hat{\varphi}^*(A)=0.6594\,|A|^{2.998}$. Bottom row: high noise ($\sigma_m=0.010$, $\sigma_{\mathrm{OU}}=0.005$), discovered $\hat{\varphi}^*(A)=0.6607\,|A|^{2.985}$. Ground truth: $\tfrac{2}{3}|A|^3$. The nonlinear loops deviate from the elliptical Newtonian shape, with the distortion becoming more pronounced at larger amplitudes.}
\label{fig:hysteresis-abspower}
\end{figure}

Finally, Figure~\ref{fig:dma-moduli-abspower} summarizes the predicted storage and loss moduli, which differently from the E1 case, exhibit pronounced amplitude dependence characteristic of nonlinear behavior. Notably, the frequency dependence of both moduli is also well reproduced, confirming that the framework correctly identifies the dissipation potential functional form even when it deviates substantially from the quadratic regime.

\begin{figure}[H]
\centering
\includegraphics[width=0.9\textwidth]{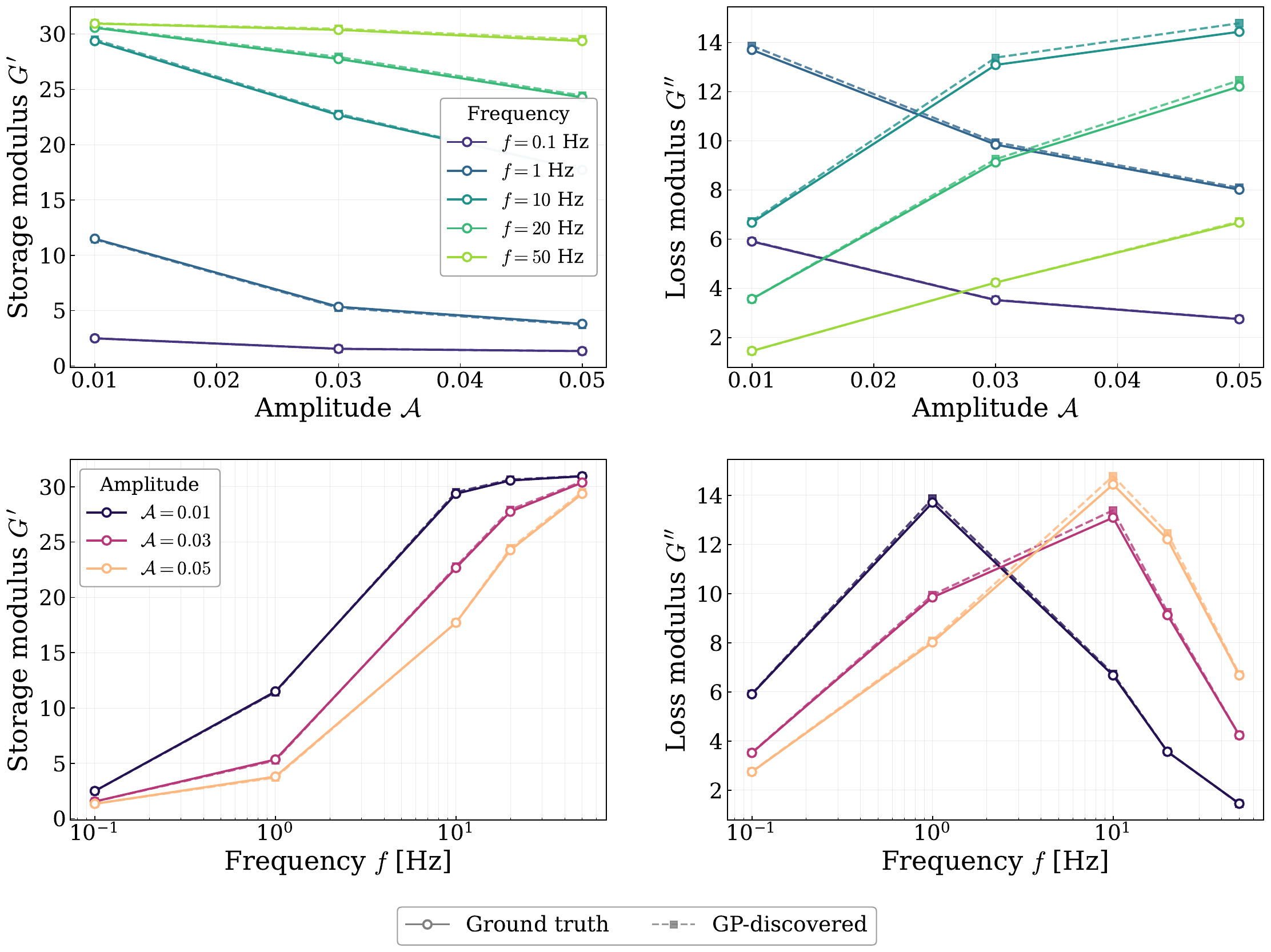}
\caption{DMA moduli comparison for the abs-power benchmark (E2), noise-free case, trial~1. $G'$ and $G''$ versus amplitude (top row) and versus frequency on logarithmic scale (bottom row). $G'$ decreases with amplitude and curves split by amplitude, confirming the nonlinear amplitude-dependent dissipation of the power-law potential. The discovered potential reproduces both the amplitude and frequency dependence closely. Solid circles: ground truth; dashed squares: predicted.}
\label{fig:dma-moduli-abspower}
\end{figure}

We extend E2 to a sweep over $q \in \{1.5,\, 2.0,\, 3.0\}$ using noise-free data ($\sigma_m = \sigma_{\mathrm{OU}} = 0$, $N_{\mathrm{trials}} = 5$). The cases $q = 2.0$ and $q = 3.0$ reproduce E1 and the abs-power case above, providing internal consistency checks. The case $q = 1.5$ probes the \emph{sub-quadratic regime} $q < 2$, in which the flow rule $\dot{z} = c\,|A|^{q-1}\sgn(A)$ has Jacobian $\partial \dot{z}/\partial A \propto |A|^{q-2}$ diverging as $A \to 0$: the effective relaxation time $\tau_{\mathrm{eff}}\propto |A|^{2-q}$ collapses at reversals and the ODE becomes stiff, approaching the rate-independent perfectly-plastic limit $\dot{z} = c\,\sgn(A)$ as $q \to 1^+$ (a Bingham flow with vanishing yield stress). Table~\ref{tab:rate-independent} reports the recovery rates and parameter estimates.

\begin{table}[H]
\centering
\caption{E2 power-law family: exponent sweep (noise-free, $N_{\mathrm{trials}}=5$ per $q^\star$, target $\hat{k}^\star=c/q^\star$). ``Exact recovery'' denotes recovery of the form $\hat{k}|A|^{\hat{q}}$ with $|\hat{q} - q^\star| < 0.15$. ``Alt.\ form'' indicates trials converging to pseudo-Huber, Macaulay bracket, or other admissible structures.}
\label{tab:rate-independent}
\begin{tabular}{@{}ccccccc@{}}
\toprule
$q^\star$ & Exact (\%) & Alt.\ form (\%) & $\hat{q}$ (mean $\pm$ std) & $\hat{k}$ (mean $\pm$ std) & $e_{G'}$ & $e_{G''}$ \\
\midrule
1.5 & 100 & 0 & $1.488 \pm 0.066$ & $1.396 \pm 0.072$ & $0.0171 \pm 0.0041$ & $0.0397 \pm 0.0135$ \\
2.0 & 100 & 0 & $2.000 \pm 0.000$ & $1.000 \pm 0.000$ & $0.0389 \pm 0.0000$ & $0.0298 \pm 0.0000$ \\
3.0 & 100 & 0 & $3.007 \pm 0.000$ & $0.658 \pm 0.000$ & $0.0186 \pm 0.0000$ & $0.0286 \pm 0.0000$ \\
\bottomrule
\end{tabular}
\end{table}

\paragraph{E3: Bingham viscoplastic.}
The Bingham dual dissipation potential
\begin{equation}\label{eq:gt-bingham}
    \varphi^*(A) = \frac{1}{2\eta}\macaulay{|A| - \sigma_Y}^2,
\end{equation}
where $\macaulay{\cdot} = \max(\cdot, 0)$ is the Macaulay bracket, $\eta = 1.0$, and $\sigma_Y > 0$ is the yield stress. The flow rule is
\begin{equation}\label{eq:bingham-flow}
    \dot{z} = \begin{cases}
        0 & \text{if } |A| \leq \sigma_Y, \\[4pt]
        \displaystyle\frac{|A| - \sigma_Y}{\eta}\sgn(A) & \text{if } |A| > \sigma_Y,
    \end{cases}
\end{equation}
exhibiting a genuine elastic domain $|A| \leq \sigma_Y$. The yield stress is swept across
\begin{equation}\label{eq:bingham-sweep}
    \sigma_Y \in \{1.0,\, 2.5,\, 5.0\}
\end{equation}
All runs use noise-free data, the LP2 triangular ramp protocol, and nRMSE fitness.

\begin{table}[H]
\centering
\caption{E3 Bingham viscoplastic benchmark (noise-free, LP2 triangular ramp, $N_{\mathrm{trials}}=5$ per $\sigma_Y$, target $\hat{k}^\star=1/(2\eta)=0.5$). ``Exact recovery'' denotes recovery of the Macaulay bracket structure $\hat{k}\macaulay{|A| - \hat{\sigma}_Y}^{\hat{r}}$ with $|\hat{r} - 2| < 0.15$. ``Alt.\ form'' indicates alternative admissible structures.}
\label{tab:bingham}
\begin{tabular}{@{}ccccccc@{}}
\toprule
$\sigma_Y$ & Exact (\%) & $\hat{\sigma}_Y$ (mean $\pm$ std) & $\hat{r}$ (mean $\pm$ std) & $\hat{k}$ (mean $\pm$ std) & nRMSE & Alt.\ form (\%) \\
\midrule
1.0 & 100 & $1.001 \pm 0.000$ & $2.00 \pm 0.00$ & $0.498 \pm 0.000$ & 0.0103 & 0 \\
2.5 & 100 & $2.503 \pm 0.000$ & $2.00 \pm 0.00$ & $0.497 \pm 0.000$ & 0.0104 & 0 \\
5.0 & 100 & $5.001 \pm 0.000$ & $2.00 \pm 0.00$ & $0.498 \pm 0.000$ & 0.0113 & 0 \\
\bottomrule
\end{tabular}
\end{table}

Table~\ref{tab:bingham} reports a 100\% exact recovery rate across all tested cases, while Fig.~\ref{fig:sigmay-timeseries} displays representative strain--stress hysteresis loops under the LP2 triangular ramp, illustrating the elastic--plastic transition accurately captured by the discovered potential.

\begin{figure}[H]
\centering
\includegraphics[width=\textwidth]{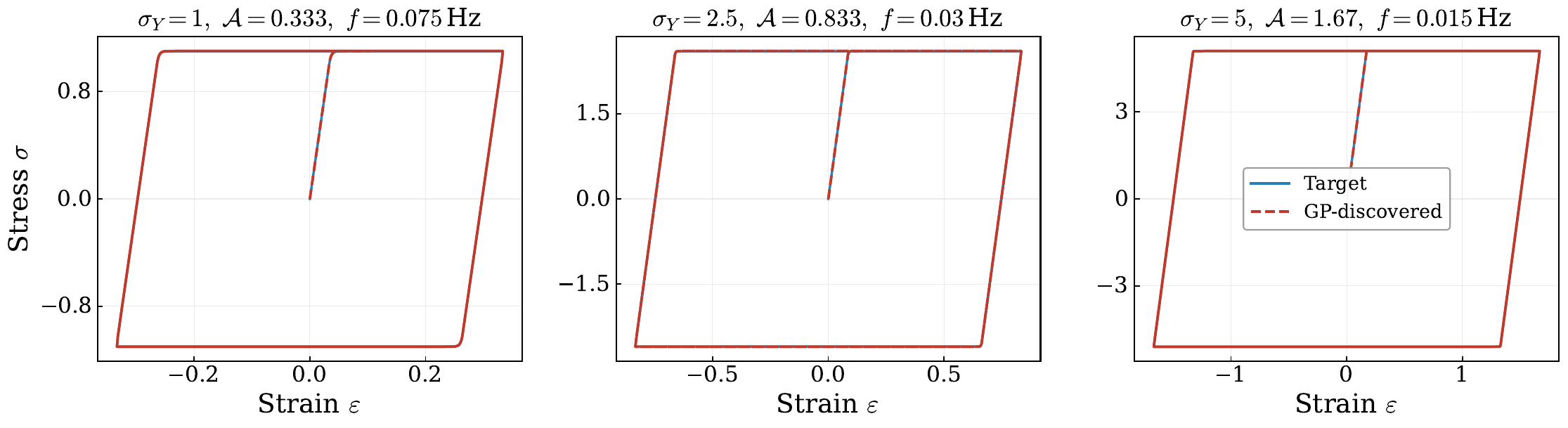}
\caption{E3 Bingham viscoplastic: strain--stress hysteresis loops under the LP2 triangular ramp protocol, for $\sigma_Y^\star \in \{1.0, 2.5, 5.0\}$ at trial~1. The discovered potential reproduces the elastic--plastic transition and the stress plateau characteristic of viscoplastic flow.}
\label{fig:sigmay-timeseries}
\end{figure}

\section{Experimental Validation}\label{sec:experimental}

\subsection{Problem Setup}\label{sec:experimental-setup}

To assess the framework on real-world data, we apply it to oscillatory shear measurements of a synthetic filled elastomer, a polydimethylsiloxane (PDMS, Sylgard~184) matrix reinforced with carbonyl iron powder (CIP) at 30\,vol\%, performed on a universal testing machine. The material exhibits a measurable rate dependence and a clearly hysteretic response under cyclic shear, motivating a non-quadratic dissipation potential.

Specimens are tested in a \emph{double-shear} configuration: two elastomeric pads ($h = 4\,\mathrm{mm}$, bonded area $A = 20\times20\,\mathrm{mm}^{2}$ each) are adhesively bonded between a central metal insert and two outer plates. Shear is applied by displacing the central insert relative to the clamped outer plates; the effective cross-sectional area entering the stress calculation is $A_{c} = 2A = 800\,\mathrm{mm}^{2}$. The engineering shear strain and shear stress are recovered from the rig signals as
\begin{equation}\label{eq:exp-strain-stress}
    \gamma(t) = \frac{u(t)}{L_{0}}, \qquad \tau(t) = \frac{F(t)}{A_{c}},
\end{equation}
where $u(t)$ is the measured cross-head displacement, $F(t)$ the load-cell force, and $L_{0} = 4\,\mathrm{mm}$ the effective gauge length. The acquisition rate is $100\,\mathrm{Hz}$ ($\Delta t = 10\,\mathrm{ms}$). Each test consists of ten sinusoidal cycles to reach a stationary response; the last cycle is retained for analysis after linear resampling onto a uniform phase grid that closes by construction (Section~\ref{sec:setup}). A representative strain history and the corresponding hysteresis loop are shown in Figure~\ref{fig:exp-strain-history}.

\begin{figure}[t]
\centering
\includegraphics[width=0.95\textwidth]{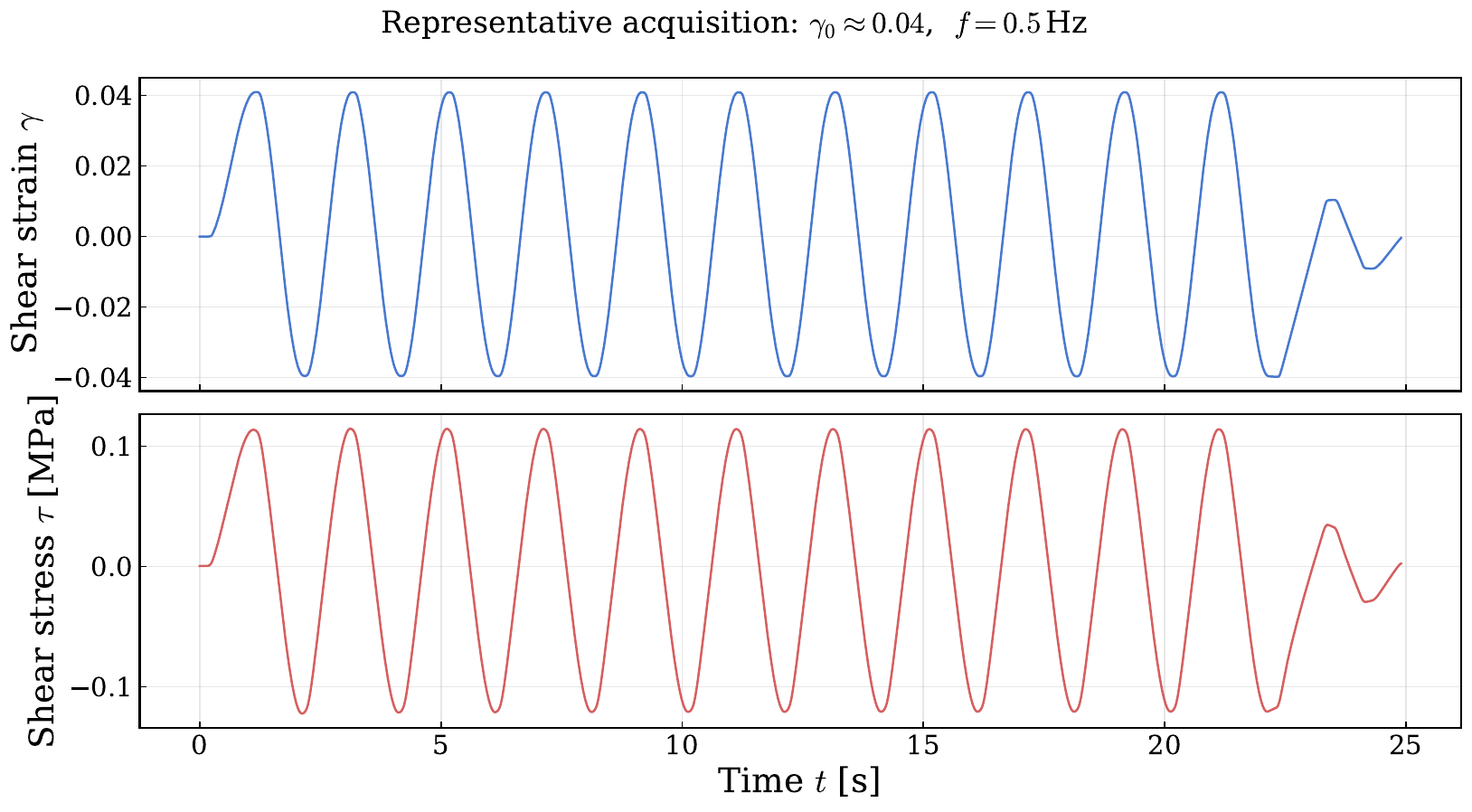}
\caption{Representative experimental acquisition (nominal strain amplitude $\gamma_0 \approx 0.04$, drive frequency $f = 0.50\,$Hz), full ten-cycle history including the initial ramp-up to stationarity. \emph{Top}: prescribed sinusoidal shear strain $\gamma(t)$. \emph{Bottom}: measured shear stress $\tau(t)$. The strain reaches a stationary sinusoidal regime after the first few cycles; the measured stress closes onto a stable periodic response with the same period but a phase lag, signaling rate-dependent dissipation.}
\label{fig:exp-strain-history}
\end{figure}

The campaign spans four nominal strain amplitudes $\gamma_{0} \in \{0.02,\,0.04,\,0.06,\,0.08\}$ and five drive frequencies $f \in \{0.10,\,0.25,\,0.50,\,0.75,\,1.00\}\,\mathrm{Hz}$, for a total of $N_s = 4 \times 5 = 20$ raw acquisitions used for identification, each comprising the full ten-cycle prescribed history sampled at $100\,\mathrm{Hz}$. Although both the strain amplitudes and the loading rates are small, the measured response is already markedly non-linear, motivating the search for a non-quadratic dissipation potential.

\paragraph{Discovery setting.}
Compared with the synthetic benchmarks of Section~\ref{sec:synthetic}, this dataset combines several sources of difficulty: the signals are noisy, the elastic moduli $(G_\infty, G_1)$ are not known a priori and must be jointly identified with the dissipation potential, and the empirical storage/loss moduli display an amplitude-and-frequency dependence.

The experimental identification follows the two-stage strategy outlined in Sec.~\ref{sec:benchmarks}: the elastic moduli $G_\infty$ and $G_1$ are first estimated from the data, and the dissipation potential $\varphi^*(A)$ is then discovered via symbolic regression. Unlike the synthetic benchmarks where $G_\infty$ and $G_1$ are known exactly, here both moduli are treated as \emph{tunable parameters} jointly optimized with the expression constants during the selective tuning step (Nelder--Mead for Macaulay-bearing candidates, L-BFGS-B otherwise) at each generation.

Initial estimates of the elastic moduli are taken from the storage modulus $G'$ extracted from the measured data, which spans $G' \in [2.5,\, 3.1]\,\mathrm{MPa}$ across the test matrix. We set $G_\infty^{(0)} = 2.7\,\mathrm{MPa}$, in line with the mean measured $G'$ at the lowest available frequency, and $G_1^{(0)} = 1.0\,\mathrm{MPa}$ as a viscoelastic increment compatible with the observed amplitude and frequency variation of $G'$. The initial values play a minor role: during evolution both moduli are re-tuned jointly with the dissipation-potential parameters $\{c_i\}$ within wide bounds, $G_\infty \in [0.1\times,\, 4\times]\,G_\infty^{(0)}$ and $G_1 \in [0.01\times,\, 4\times]\,G_1^{(0)}$, by the selective Nelder--Mead / L-BFGS-B routine of Section~\ref{sec:optimization}. Figure~\ref{fig:flowchart-exp} summarises the identification procedure.

\begin{figure}[H]
\centering
\begin{tikzpicture}[
    node distance=0.8cm and 1.2cm,
    box/.style={rectangle, draw, rounded corners, minimum width=2.8cm, minimum height=0.9cm, align=center, font=\small},
    data/.style={box, fill=blue!15},
    process/.style={box, fill=green!15},
    gp/.style={box, fill=orange!15},
    output/.style={box, fill=red!15},
    arrow/.style={-{Stealth[length=2.5mm]}, thick}
]

\node[data] (strain) {Strain data\\$\gamma(t)$};
\node[data, right=of strain] (stress) {Stress data\\$\tau(t)$};

\node[process, below=of $(strain)!0.5!(stress)$] (init-est) {Initial estimates\\$\hat{G}_\infty,\;\hat{G}_1$};

\node[gp, right=of init-est, xshift=1.5cm] (init) {Initialize population\\$\mathcal{P}_0 \subset \mathcal{G}_{\mathrm{cvx}}^{\mathrm{comp}}$};
\node[gp, below=of init] (tune) {Parameter tuning\\(NM / L-BFGS-B, bounded)\\$\{c_i,\, G_\infty,\, G_1\}$};
\node[gp, below=of tune] (rollout) {Rollout simulation\\$\dot{z} = \partial_A \varphi^*(A)$\\$\tau_{\text{pred}} = G_\infty \gamma + A$};
\node[gp, below=of rollout] (fitness) {Fitness evaluation\\nRMSE$(\tau_{\text{pred}},\,\tau_{\text{meas}})$};
\node[gp, below=of fitness] (evolve) {Reproduction\\Elitism + Tournament\\(Crossover + Mutation)};

\node[output, below=of evolve] (output) {Best $\hat{\varphi}^*(A),\;\hat{G}_\infty,\;\hat{G}_1$};

\draw[arrow] (strain) -- (init-est);
\draw[arrow] (stress) -- (init-est);
\draw[arrow] (init-est) -- (init);

\draw[arrow] (stress.east) -- ++(0,-0.4) -| ([xshift=0.8cm]fitness.east) -- (fitness.east);

\draw[arrow] (init) -- (tune);
\draw[arrow] (tune) -- (rollout);
\draw[arrow] (rollout) -- (fitness);
\draw[arrow] (fitness) -- (evolve);
\draw[arrow] (evolve) -- (output);

\draw[arrow] (evolve.west) -- ++(-0.8,0) |- node[pos=0.25, left, font=\scriptsize] {next gen.} (tune.west);

\begin{scope}[on background layer]
\node[draw=orange!60, dashed, rounded corners, fit=(init)(tune)(rollout)(fitness)(evolve), inner sep=0.3cm, label={[font=\scriptsize]above:GP Evolution Loop}] {};
\end{scope}

\end{tikzpicture}
\caption{Schematic of the identification strategy used for experimental data. Unlike the synthetic pipeline in Fig.~\ref{fig:flowchart}, the elastic moduli $G_\infty$ and $G_1$ are \emph{not} known \emph{a priori}: they are initialized from the data and jointly optimized with the expression constants $\{c_i\}$ during the selective Nelder--Mead / L-BFGS-B tuning step at each generation. Fitness is evaluated by comparing the predicted stress $\tau_{\text{pred}}(t)$ directly against the measured stress $\tau_{\text{meas}}(t)$ via nRMSE on the full raw acquisition.}
\label{fig:flowchart-exp}
\end{figure}
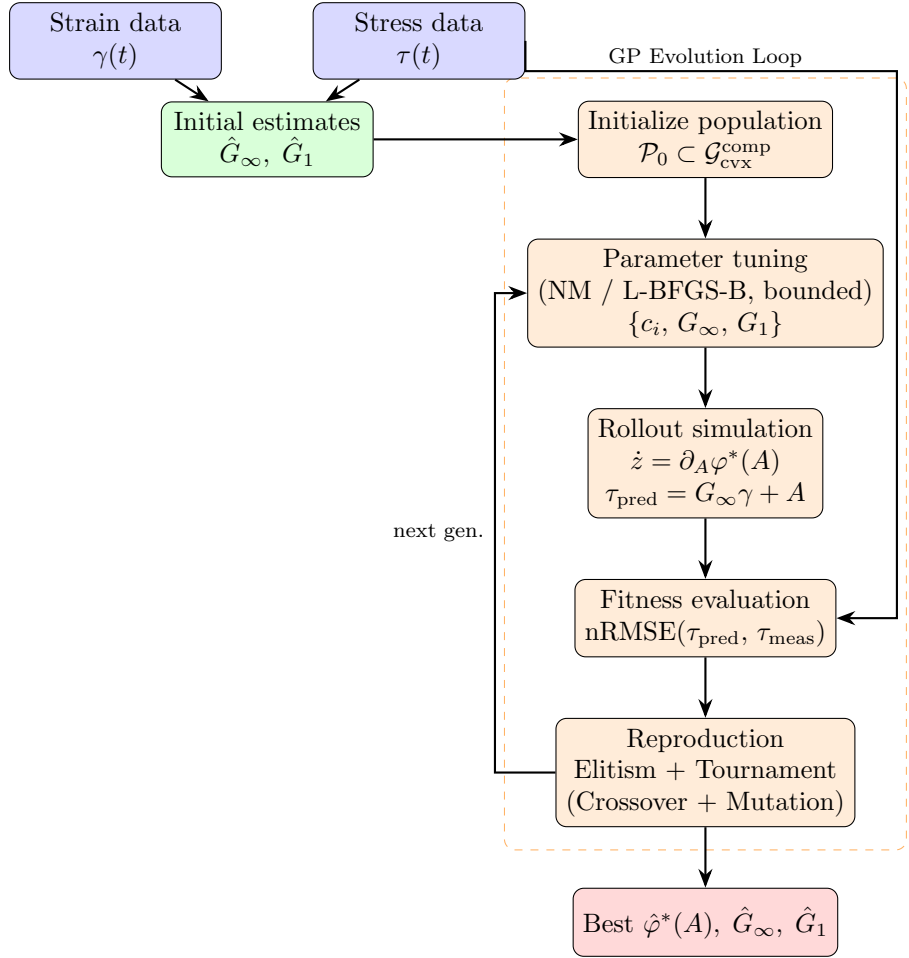

The performance of the GP-discovered potential is benchmarked against a calibrated linear Zener (SLS) model with quadratic dissipation~\eqref{eq:gt-quadratic} with $\hat{k}$ fitted on the full multi-frequency dataset by the same symbolic-regression machinery, restricted to the quadratic primitive and with $G_\infty$ and $G_1$ jointly optimized as in the GP runs. Across the $N_s=20$ retained acquisitions the calibrated Zener model attains $\hat{k} = 3.42$, $\hat{G}_\infty = 2.50\,\mathrm{MPa}$, $\hat{G}_1 = 0.92\,\mathrm{MPa}$ (corresponding to $\hat{\eta} = 1/(2\hat{k}) = 0.146\,\mathrm{MPa\cdot s}$ and a relaxation time $\hat{\tau} = \hat{\eta}/\hat{G}_1 = 0.16\,\mathrm{s}$), with a stress nRMSE of $0.110$. This baseline is physically meaningful but structurally rigid: it forces a single Maxwell relaxation time and amplitude-independent moduli, while the empirical $(G',G'')$ vary by $\sim 30\%$ across the test matrix. It serves as the reference against which the discovered potentials are compared in Section~\ref{sec:results-exp}.

\subsection{Results}\label{sec:results-exp}
GP hyperparameters are listed in the Experimental column of Tab.~\ref{tab:hyperparams}: a population of $60$ individuals evolves for $\Gamma_{\max}=5$ generations within a grammar of depth $D_{\max}=6$ and node budget $N_{\max}=20$, with the inner constant tuner using the \emph{selective Nelder--Mead} scheme of Section~\ref{sec:optimization} (gradient-free simplex for expressions containing the Macaulay bracket primitive, single-shot L-BFGS-B otherwise), allotted $N_{\mathrm{iter}}^{\max}=200$ iterations and $N_{\mathrm{restart}}=3$ multi-restarts when the simplex branch is active. A per-evaluation timeout of $540\,\mathrm{s}$ aborts stiff candidates whose rollout fails to terminate. Unlike the E1 and E2 synthetic benchmarks, which compare first-harmonic moduli $(G',G'')$, the experimental fitness is the time-domain stress nRMSE computed on the \emph{full raw time series} of each acquisition (all prescribed cycles, including the initial transient), with the rollout integrated from rest $z(0)=0$ and averaged across the $N_s=20$ samples. For each identification, $N_{\mathrm{trials}}=10$ independent trials are run with different random seeds.

\begin{figure}[H]
\centering
\includegraphics[width=\textwidth]{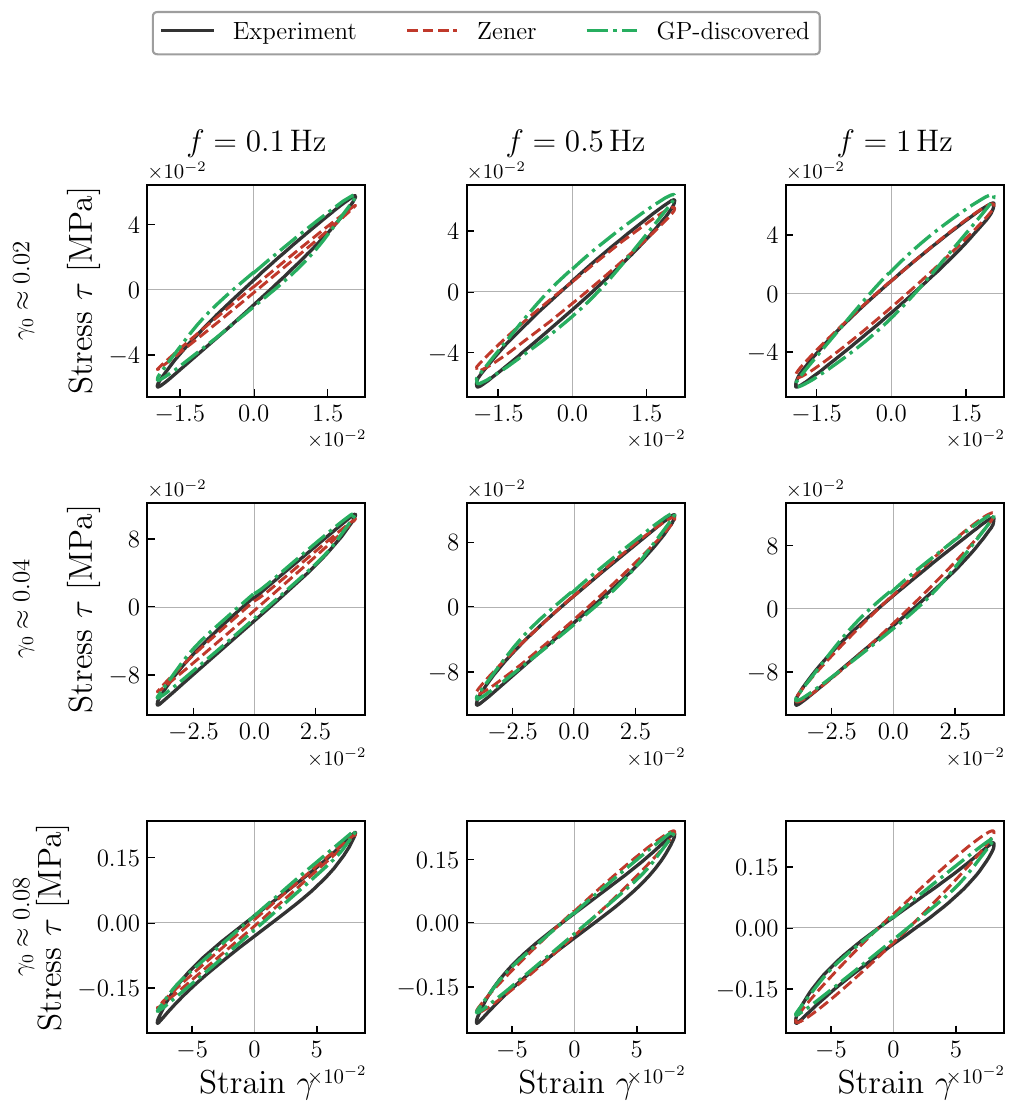}
\caption{Shear hysteresis loops on the synthetic elastomer specimen. Rows: nominal strain amplitudes $\gamma_0 \in \{0.02,\,0.04,\,0.08\}$; columns: drive frequencies $f \in \{0.10,\,0.50,\,1.00\}\,$Hz. Solid (dark): last-cycle experimental loop; dashed (red): calibrated linear Zener (SLS) baseline with $\hat{k}=3.42$, $\hat{G}_\infty=2.50\,$MPa, $\hat{G}_1=0.92\,$MPa, $\hat\eta=0.146\,$MPa$\cdot$s (multi-history nRMSE $= 0.110$); dash-dotted (green): GP-discovered potential of trial~6, $\hat{\varphi}^*(A)=\cosh(3.27\,\sinh^2(4.65A))-1$, $\hat{G}_\infty=2.50\,$MPa, $\hat{G}_1=1.55\,$MPa (nRMSE $= 0.0822$). All model curves are taken from the last cycle of a 20-cycle warm-up rollout to ensure steady-state comparison. Aggregated over the full test matrix of $4 \times 5 = 20$ amplitude--frequency combinations, the per-combination last-cycle nRMSE averages $0.110$ for the Zener baseline and $0.076$ for the GP-discovered potential (ratio $\sim 1.44$; the GP outperforms the Zener in $18$ of the $20$ combinations).}
\label{fig:sin-hysteresis}
\end{figure}

\begin{figure}[H]
\centering
\includegraphics[width=\textwidth]{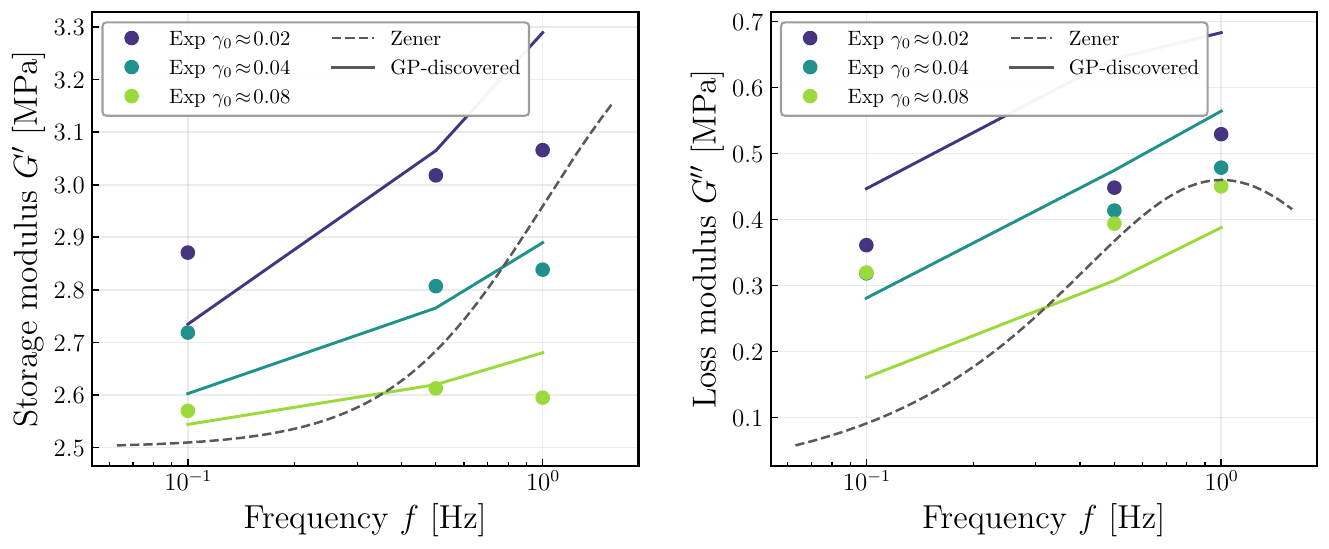}
\caption{Dynamic moduli of the synthetic elastomer specimen extracted from the first harmonic of the experimental loops and of the model rollouts. \emph{Left}: storage modulus $G'(f)$; \emph{right}: loss modulus $G''(f)$. Dots: experimental, one colour per nominal strain amplitude $\gamma_0 \in \{0.02,\,0.04,\,0.08\}$; solid line (same colour): GP-discovered potential of trial~6, $\hat{\varphi}^*(A)=\cosh\!\bigl(3.27\,\sinh^2(4.65\,A)\bigr)-1$ with $\hat{G}_\infty=2.50\,$MPa and $\hat{G}_1=1.55\,$MPa, rolled out on $\gamma_0\sin(2\pi f t)$ at the same amplitude; dashed grey: calibrated linear Zener baseline (amplitude-independent). The experimental softening of $G'$ with amplitude is reproduced by the GP model and lies outside the Zener prediction by construction.}
\label{fig:sin-dma}
\end{figure}

We ran $N_{\mathrm{trials}}=10$ independent GP searches with different random seeds, each evolving a population of $60$ individuals over $5$ generations on the full set of $N_s=20$ raw time histories (no cycle extraction, rollout integrated from rest $z(0)=0$, see Section~\ref{sec:experimental}). Table~\ref{tab:sin-trials} reports, for each trial, the multi-history stress nRMSE and the discovered potential together with the jointly optimized elastic moduli. All ten trials converge to admissible non-quadratic potentials with multi-history nRMSE in the range $0.082$--$0.104$, comfortably below the calibrated linear-Zener baseline (nRMSE $=0.110$).

\begin{table}[H]
\centering
\caption{Experimental validation, synthetic elastomer specimen, multi-history fit over $N_s=20$ raw acquisitions ($\gamma_0 \in \{0.02,0.04,0.06,0.08\}$, $f \in \{0.10,0.25,0.50,0.75,1.00\}\,$Hz, full ten-cycle history with rollout from $z(0)=0$). nRMSE is the time-domain stress error averaged across the $N_s$ histories; $\hat{G}_\infty$ and $\hat{G}_1$ are jointly tuned with the dissipation potential. Trials are sorted by nRMSE.}
\label{tab:sin-trials}
\begin{tabular}{@{}clccc@{}}
\toprule
Trial & Discovered $\hat{\varphi}^*(A)$ & nRMSE & $\hat{G}_\infty$ [MPa] & $\hat{G}_1$ [MPa] \\
\midrule
6  & $\sinh^2(2.3\cdot\!10^{-3}A) + \cosh\!\bigl(3.27\,\sinh^2(4.65A)\bigr) - 1$        & $0.0822$ & $2.50$ & $1.55$ \\
10 & $\cosh\!\bigl(4.04\,\cosh(5.72A) - 4.04\bigr) - 1$                              & $0.0840$ & $2.48$ & $1.48$ \\
2  & $3.87\,\bigl[\cosh(\cosh(7.78A){-}1) - 1\bigr]$                                  & $0.0849$ & $2.44$ & $1.27$ \\
4  & $2.41\,\bigl[\cosh(4.26A) - 1\bigr]^{1.42}$                                     & $0.0859$ & $2.55$ & $1.57$ \\
5  & $[\log\cosh(5.08A)]^{1.32}$                                                     & $0.0859$ & $2.54$ & $1.17$ \\
1  & $[\log\cosh(5.16A)]^{1.34}$                                                     & $0.0885$ & $2.55$ & $1.49$ \\
3  & $[\log\cosh(5.16A)]^{1.33}$                                                     & $0.0889$ & $2.55$ & $1.54$ \\
8  & $\cosh(4.59\,|A|^{1.16}) - 1$                                                   & $0.0974$ & $2.57$ & $1.52$ \\
9  & $3.11\,\exp(3.17\,|A|^{2.29}) + \cosh(0.059A) - 4.11$                            & $0.0978$ & $2.57$ & $1.47$ \\
7  & $\bigl(2.02+2.02+2.25\bigr)\,|A|^{2.15}$                                        & $0.1036$ & $2.54$ & $1.14$ \\
\midrule
\multicolumn{2}{l}{Linear Zener baseline ($\varphi^* = \hat{k}\,A^2$, $\hat{k}=3.42$)} & $0.1103$ & $2.50$ & $0.92$ \\
\bottomrule
\end{tabular}
\end{table}

The best fit is achieved by trial~6 with
\begin{equation}\label{eq:gp-sin-best}
    \hat{\varphi}^*(A) = \sinh^2\!\bigl(2.33\!\cdot\!10^{-3}\,A\bigr) + \cosh\!\bigl(3.268\,\sinh^2(4.646\,A)\bigr) - 1,
\end{equation}
jointly with $\hat{G}_\infty = 2.500\,\mathrm{MPa}$ and $\hat{G}_1 = 1.547\,\mathrm{MPa}$, reaching a multi-history nRMSE of $0.0822$ on the $N_s=20$ retained acquisitions. The leading $\sinh^2(2.33\!\cdot\!10^{-3}\,A)$ term is in fact a degenerate branch of the discovered tree: with the experimental parsimony coefficient set to zero (Tab.~\ref{tab:hyperparams}), the inner tuner drove its multiplicative constant to $2.33\!\cdot\!10^{-3}$, so that over the rollout range $|A|\lesssim 0.15\,\mathrm{MPa}$ both the potential value and its derivative contributed by this term remain $\sim 10^{-8}$ of those produced by the dominant $\cosh\!\bigl(3.27\,\sinh^2(4.65A)\bigr)$ composition. We retain the term in~\eqref{eq:gp-sin-best} for faithfulness to the raw output, but dropping it leaves the fit unchanged. Dropping the degenerate branch, the discovered dissipation potential is effectively the single hyperbolic composition $\cosh\!\bigl(3.27\,\sinh^2(4.65\,A)\bigr) - 1$: the inner $\sinh^2$ rises quadratically near the origin and exponentially at large $|A|$, and the outer $\cosh(\cdot){-}1$ then drives a strongly progressive stiffening of the dissipation rate, accounting for the observed amplitude- and rate-dependent loss not captured by the quadratic Zener. Relative to the calibrated linear baseline, the GP-discovered potential reduces the multi-history stress error by $25.4\%$ at no additional model-class assumption beyond the convexity-preserving grammar of Section~\ref{sec:methodology}.

The remaining nine successful trials cluster around three families, all admissible by the grammar and all clearly non-quadratic: (i) $\log\cosh(\cdot)^q$ with $q \approx 1.33$ (trials~1, 3, 5; $\hat{q}$ stable to the second decimal), (ii) $(\cosh(\cdot)-1)^q$ or nested-$\cosh$ compositions (trials~2, 4, 10), and (iii) more elaborate combinations involving $\cosh(|A|^p)$ or $\exp(|A|^p)$ terms (trials~8, 9). All produce nRMSE below $0.10$. The power-law-only candidate (trial~7) is the worst at $0.1036$, marginally improving over the Zener baseline; this is consistent with the synthetic finding that pure power laws struggle to capture the simultaneous amplitude--frequency dependence of the dynamic moduli.

\section{Conclusion}\label{sec:conclusion}

We have proposed a symbolic regression framework for data-driven discovery of dissipation potentials in inelastic materials within the Generalized Standard Materials formalism. The central contribution is a convexity-preserving grammar $\mathcal{G}_{\mathrm{cvx}}^{\mathrm{comp}}$ that guarantees thermodynamic consistency, both convexity and non-negativity, by construction.

The grammar design exploits the fundamental convexity preservation result under composition (Proposition~\ref{prop:composition}): convexity is preserved when a convex, non-decreasing outer function is composed with a convex inner function. By restricting outer functions to $\mathcal{F}_{\mathrm{out}}$ and inner functions to positive combinations of $\mathcal{F}_{\mathrm{prim}}$, all generated expressions satisfy the thermodynamic requirements (P1) and (P2). Grammar-compliant genetic operators (term-level crossover and mutation), combined with constrained nonlinear parameter optimization, then provide end-to-end guarantees that every candidate potential in every generation is thermodynamically admissible.

Validation on synthetic data demonstrates strong symbolic recovery across three benchmark families (Tab.~\ref{tab:benchmarks}): Newtonian viscosity (E1), recovered through the continuous power-law primitive $|A|^p$ with near-exact recovery rate of 100\% across all noise levels; the power-law family (E2), where the exponent is identified with $93.3\%$ exact recovery and the remaining trials converge to alternative admissible structures (mixed power-law + hyperbolic, exponential--quadratic), including the stiff sub-quadratic regime $q<2$; and Bingham viscoplastic yield via the Macaulay bracket primitive (E3), with full recovery of the yield threshold $\sigma_Y$. Across the Monte Carlo noise grid, the storage-modulus error $e_{G'}$ is more sensitive to process noise $\sigma_{\mathrm{OU}}$ than to measurement noise $\sigma_m$, while the loss-modulus error $e_{G''}$ remains below 5\% in most configurations. The framework eliminates thermodynamic violations entirely across all tested noise levels.

Experimental validation on oscillatory shear data for an elastomeric compound confirms practical applicability. By jointly optimizing elastic moduli alongside expression constants within each GP generation, the method discovers non-quadratic dissipation potentials that capture the nonlinear frequency dependence of the storage and loss moduli more accurately than a linear Zener baseline. More broadly, the same grammar serves as a thermodynamically admissible hypothesis class for post-hoc symbolic extraction from trained neural surrogates.

This approach bridges physics-informed machine learning and classical phenomenological modeling, offering constitutive model discovery with interpretability and formal thermodynamic guarantees. The discovered potentials inherit the variational structure of the GSM framework, enabling direct deployment in finite element codes with guaranteed numerical stability, and the explicit expressions satisfy the interpretability requirements for regulatory applications in biomechanics and aerospace~\cite{FDA2021, ASMEVV40}. By adopting the subdifferential setting from the outset, the framework spans the full viscoelastic--viscoplastic spectrum within a unified identification procedure.

Several directions remain open. The current formulation is restricted to a single internal variable ($N=1$); extension to multiple dissipative branches requires solving a simultaneous state estimation problem, since internal variables are no longer uniquely identifiable from stress--strain data alone. Tensor generalization to three-dimensional stress states requires recasting the grammar in terms of invariants of the driving force tensor $\tens{A}$, following approaches developed for hyperelastic energy discovery~\cite{abdusalamov2023automatic}. More exotic potentials, such as those arising in granular media or shape-memory alloys, may exceed current grammar expressivity, though the modular design permits targeted extensions. Finally, the assumption of a known free energy $\psi$ can be relaxed by simultaneously discovering both potentials through joint or alternating symbolic search.

\appendix

\section{Variational Time Integration}\label{app:time-integration}

While the fitness evaluation in Section~\ref{sec:optimization} uses explicit Euler for computational efficiency, the discovered potentials can be deployed in production finite element codes (ABAQUS\textregistered, LS-DYNA\textregistered) that employ implicit time integration schemes. A distinct advantage of the GSM framework is the existence of a variational principle for time discretization. Consider a time step $\Delta t = t_{n+1} - t_n$. The \emph{incremental energy minimization} problem reads:
\begin{equation}\label{eq:incremental-min}
    \tens{z}_{n+1} = \operatorname*{argmin}_{\tens{z}} \left\{ \psi(\tens{\varepsilon}_{n+1}, \tens{z}) + \Delta t \, \varphi\left( \frac{\tens{z} - \tens{z}_n}{\Delta t} \right) \right\}.
\end{equation}
The first-order optimality condition recovers the backward Euler discretization of \eqref{eq:evolution}:
\begin{equation}\label{eq:backward-euler}
    \frac{\tens{z}_{n+1} - \tens{z}_n}{\Delta t} = \nabla_{\tens{A}} \varphi^*(\tens{A}_{n+1}), \qquad \tens{A}_{n+1} = -\frac{\partial \psi}{\partial \tens{z}}(\tens{\varepsilon}_{n+1}, \tens{z}_{n+1}).
\end{equation}
Since $\psi$ is convex in $\tens{z}$ and $\varphi$ is convex by construction (being the Legendre--Fenchel conjugate of the convex $\varphi^*$), the objective in~\eqref{eq:incremental-min} is strictly convex in~$\tens{z}$, guaranteeing uniqueness of the minimiser and unconditional stability of the implicit scheme. This variational structure also preserves energy dissipation at the discrete level and ensures guaranteed convergence for convex potentials.

\section{Mathematical Proofs}\label{app:proofs}

\subsection{Grammar Consistency (Theorem~\ref{thm:grammar})}\label{proof:grammar}

\noindent\textbf{Theorem~\ref{thm:grammar} (Grammar Consistency).} \emph{Every expression $\varphi^*$ generated by $\mathcal{G}_{\mathrm{cvx}}^{\mathrm{comp}}$ satisfies properties \textup{(P1)} and \textup{(P2)} in Proposition~\ref{prop:dissipation}: $\varphi^*$ is convex, non-negative, and normalized ($\varphi^*(0) = 0$).}

\paragraph{Proof of Theorem~\ref{thm:grammar} (Grammar Consistency).}
\begin{proof}
We proceed by structural induction on the grammar production rules~\eqref{eq:grammar-expr}--\eqref{eq:grammar-inner}.

\textbf{Base case (Primitives):} By construction of the primitive library $\mathcal{F}_{\mathrm{prim}}$ (Tab.~\ref{tab:primitives}, verified by direct computation in Section~\ref{sec:convexity}), every \texttt{Primitive} is convex, non-negative, symmetric, and satisfies $\varphi^*(0) = 0$.

\textbf{Inductive case (Sums and positive scaling):} If $f_1, f_2$ are convex, non-negative, and normalized and $c > 0$, then by Proposition~\ref{prop:closure} both $f_1 + f_2$ and $c\,f_1$ are convex and non-negative; normalization $(f_1+f_2)(0) = (c\,f_1)(0) = 0$ is immediate.

\textbf{Inductive case (Composition):} Let $\varphi^* = f(g(A))$ where $f \in \mathcal{F}_{\mathrm{out}}$ and $g$ is generated by \texttt{InnerFunc}~\eqref{eq:grammar-inner}. Since \texttt{InnerFunc} is built from \texttt{Primitive}s, positive combinations thereof, and nested \texttt{Composition}s, the inductive hypothesis guarantees that $g$ is convex, non-negative, and normalized. By construction of the outer-function class $\mathcal{F}_{\mathrm{out}}$~\eqref{eq:outer}, $f$ is convex and non-decreasing on $\mathbb{R}^+$ with $f(0)=0$. Since $g \geq 0$, its range lies in $\mathbb{R}^+$, where $f$ is non-decreasing, so the hypotheses of Proposition~\ref{prop:composition} are met; hence $\varphi^* = f \circ g$ is convex. Since $g(A) \geq 0$ and $f$ is non-negative on $\mathbb{R}^+$, we have $\varphi^*(A) = f(g(A)) \geq 0$. Finally, $\varphi^*(0) = f(g(0)) = f(0) = 0$.
\end{proof}

\subsection{Invariance of Genetic Operators}\label{proof:operators}
\label{proof:invariant}

\noindent\textbf{Theorem~\ref{thm:invariant} (Invariant: Population Thermodynamic Consistency).} \emph{Let $\mathcal{P}_\gamma$ denote the population at generation $\gamma$. If $\mathcal{P}_0 \subset \mathcal{G}_{\mathrm{cvx}}^{\mathrm{comp}}$, then $\mathcal{P}_\gamma \subset \mathcal{G}_{\mathrm{cvx}}^{\mathrm{comp}}$ for all $\gamma \geq 0$. Consequently, every candidate in every generation satisfies properties (P1) and (P2).}

\paragraph{Proof of Invariance (Theorem~\ref{thm:invariant}).}
\begin{proof}
By induction on $\gamma$. The base case $\gamma = 0$ holds by initialization. For the inductive step, assume $\mathcal{P}_{\gamma-1} \subset \mathcal{G}_{\mathrm{cvx}}^{\mathrm{comp}}$. Population $\mathcal{P}_\gamma$ is formed from:
\begin{itemize}
    \item Elite individuals from $\mathcal{P}_{\gamma-1}$: remain in $\mathcal{G}_{\mathrm{cvx}}^{\mathrm{comp}}$;
    \item Crossover offspring: in $\mathcal{G}_{\mathrm{cvx}}^{\mathrm{comp}}$ by the Lemma below;
    \item Mutated individuals: in $\mathcal{G}_{\mathrm{cvx}}^{\mathrm{comp}}$ by the Lemma below.
\end{itemize}
Coefficient tuning via constrained optimization (the two-mode L-BFGS-B / Nelder--Mead search of Section~\ref{sec:optimization}) maintains $a \geq 0$, preserving the positive linear combination structure. By Theorem~\ref{thm:grammar}, all individuals satisfy (P1) and (P2).
\end{proof}

\paragraph{Proof: Crossover preserves grammar membership.}
Each individual is an additive list of blocks $\varphi^* = \sum_k B_k$, with every block $B_k$ a valid \texttt{ConvexExpr}~\eqref{eq:grammar-expr}. Given two parents $\varphi^*_1, \varphi^*_2 \in \mathcal{G}_{\mathrm{cvx}}^{\mathrm{comp}}$, term-level crossover recombines their blocks (insertion or replacement), so the child block list is a subset of $\{B_k^{(1)}\} \cup \{B_k^{(2)}\}$ and each child block is a valid \texttt{ConvexExpr}. Since the sum production rule $\texttt{ConvexExpr} \to \texttt{ConvexExpr} + \texttt{ConvexExpr}$~\eqref{eq:grammar-expr} accepts \texttt{ConvexExpr} operands of any internal form, the combination $\varphi^*_{\mathrm{child}} = \sum_k B_k^{\mathrm{child}}$ is derivable from the grammar; hence $\varphi^*_{\mathrm{child}} \in \mathcal{G}_{\mathrm{cvx}}^{\mathrm{comp}}$.

\paragraph{Proof: Mutation preserves grammar membership.}
Term-level mutation acts on the block list of $\varphi^* = \sum_k B_k$ and either (i) replaces a block with a freshly sampled \texttt{ConvexExpr} block, (ii) appends one, or (iii) removes a block (when at least two are present). Each freshly sampled block is a valid \texttt{ConvexExpr} by construction, so the result is again an additive combination of valid \texttt{ConvexExpr} blocks, hence in $\mathcal{G}_{\mathrm{cvx}}^{\mathrm{comp}}$ by the sum production rule~\eqref{eq:grammar-expr}.

\section*{CRediT Authorship Contribution Statement}
\textbf{F. Califano:} Conceptualization, Methodology, Software, Validation, Writing -- Original Draft. \textbf{J. Ciambella:} Conceptualization, Supervision, Validation, Writing -- Review \& Editing, Funding Acquisition.

\section*{Declaration of Competing Interest}
The authors declare that they have no known competing financial interests or personal relationships that could have appeared to influence the work reported in this paper.

\section*{Acknowledgments}
The authors gratefully acknowledge Giuseppe Tomassetti for valuable discussions and insightful suggestions that shaped the development of this work. The authors also wish to acknowledge the support of the Italian National Group of Mathematical Physics (GNFM--INdAM).

\section*{Data Availability}
The synthetic datasets and Python code used in this study are available from the corresponding author upon reasonable request.

\bibliographystyle{unsrt}
\bibliography{references}

@article{abdusalamov2023automatic,
  author  = {Abdusalamov, R. and Hillg{\"a}rtner, M. and Itskov, M.},
  title   = {Automatic generation of interpretable hyperelastic material models by symbolic regression},
  journal = {International Journal for Numerical Methods in Engineering},
  volume  = {124},
  number  = {9},
  pages   = {2093--2104},
  year    = {2023}
}

@inproceedings{amos2017input,
  author    = {Amos, B. and Xu, L. and Kolter, J. Z.},
  title     = {Input convex neural networks},
  booktitle = {Proceedings of the 34th International Conference on Machine Learning},
  pages     = {146--155},
  year      = {2017}
}

@book{ASMEVV40,
  author    = {{ASME}},
  title     = {{V\&V} 40--2018: Assessing Credibility of Computational Modeling through Verification and Validation: Application to Medical Devices},
  publisher = {American Society of Mechanical Engineers},
  year      = {2018}
}

@article{bahmani2024physics,
  author  = {Bahmani, B. and Sun, W.},
  title   = {Physics-constrained symbolic model discovery for polyconvex incompressible hyperelastic materials},
  journal = {International Journal for Numerical Methods in Engineering},
  volume  = {125},
  pages   = {e7473},
  year    = {2024}
}

@article{bahmani2023elastoplasticity,
  author  = {Bahmani, B. and Suh, H. S. and Sun, W.},
  title   = {Discovering interpretable elastoplasticity models via the neural polynomial method enabled symbolic regressions},
  journal = {Computer Methods in Applied Mechanics and Engineering},
  volume  = {422},
  pages   = {116827},
  year    = {2024},
  doi     = {10.1016/j.cma.2024.116827}
}

@article{batchelor1967introduction,
  author    = {Batchelor, G. K.},
  title     = {An Introduction to Fluid Dynamics},
  journal   = {Cambridge University Press},
  year      = {1967}
}

@book{bingham1922fluidity,
  author    = {Bingham, E. C.},
  title     = {Fluidity and Plasticity},
  publisher = {McGraw-Hill},
  year      = {1922}
}

@article{bomarito2021development,
  author    = {Bomarito, G. F. and Hochhalter, J. D. and Rubin, T. J. and Townsend, K. A.},
  title     = {Development of interpretable, data-driven plasticity models with symbolic regression},
  journal   = {Computers \& Structures},
  volume    = {252},
  pages     = {106557},
  year      = {2021},
  doi       = {10.1016/j.compstruc.2021.106557}
}

@book{boyd2004convex,
  author    = {Boyd, S. and Vandenberghe, L.},
  title     = {Convex Optimization},
  publisher = {Cambridge University Press},
  year      = {2004}
}

@article{byrd1995limited,
  author  = {Byrd, R. H. and Lu, P. and Nocedal, J. and Zhu, C.},
  title   = {A limited memory algorithm for bound constrained optimization},
  journal = {SIAM Journal on Scientific Computing},
  volume  = {16},
  number  = {5},
  pages   = {1190--1208},
  year    = {1995}
}

@article{califano2026enhancing,
  author  = {Califano, F. and Ciambella, J.},
  title   = {Enhancing nonlinear viscoelastic modeling of elastomers through neural networks: A deep rheological element},
  journal = {Mechanics of Materials},
  volume  = {212},
  pages   = {105525},
  year    = {2026}
}

@article{Califano2023,
author = {Califano, F and Ciambella, J},
doi = {10.1098/rspa.2023.0603},
issn = {1364-5021},
journal = {Proc. R. Soc. A Math. Phys. Eng. Sci.},
month = {dec},
number = {2280},
title = {{Viscoplastic simple shear at finite strains}},
volume = {479},
year = {2023}
}

@article{ciambella2009abaqus,
  author    = {Ciambella, J. and Destrade, M. and Ogden, R. W.},
  title     = {On the {ABAQUS} {FEA} model of finite viscoelasticity},
  journal   = {Rubber Chemistry and Technology},
  volume    = {82},
  number    = {2},
  pages     = {184--193},
  year      = {2009},
  doi       = {10.5254/1.3548243}
}

@article{ciambella2021anisotropic,
  author    = {Ciambella, J. and Nardinocchi, P.},
  title     = {A structurally frame-indifferent model for anisotropic visco-hyperelastic materials},
  journal   = {Journal of the Mechanics and Physics of Solids},
  volume    = {147},
  pages     = {104247},
  year      = {2021},
  doi       = {10.1016/j.jmps.2020.104247}
}

@book{colton2013inverse,
  author    = {Colton, D. and Kress, R.},
  title     = {Inverse Acoustic and Electromagnetic Scattering Theory},
  publisher = {Springer},
  address   = {New York},
  edition   = {3rd},
  year      = {2013}
}

@inproceedings{dugas2000softplus,
  author    = {Dugas, C. and Bengio, Y. and B{\'e}lisle, F. and Nadeau, C. and Garcia, R.},
  title     = {Incorporating second-order functional knowledge for better option pricing},
  booktitle = {Proceedings of the 13th International Conference on Neural Information Processing Systems (NIPS'00)},
  pages     = {451--457},
  publisher = {MIT Press},
  year      = {2000}
}

@article{eyring1936viscosity,
  author  = {Eyring, H.},
  title   = {Viscosity, plasticity, and diffusion as examples of absolute reaction rates},
  journal = {The Journal of Chemical Physics},
  volume  = {4},
  number  = {4},
  pages   = {283--291},
  year    = {1936}
}

@techreport{FDA2021,
  author      = {{U.S. Food and Drug Administration}},
  title       = {Artificial Intelligence/Machine Learning ({AI/ML})-Based Software as a Medical Device ({SaMD}) Action Plan},
  institution = {FDA},
  year        = {2021}
}

@article{fuhg2024review,
  author    = {Fuhg, J. N. and Padmanabha, G. A. and Bouklas, N. and Bahmani, B. and Sun, W. and Vlassis, N. N. and Flaschel, M. and De Lorenzis, L.},
  title     = {A review on data-driven constitutive laws for solids},
  journal   = {Archives of Computational Methods in Engineering},
  year      = {2024},
  doi       = {10.1007/s11831-024-10196-2}
}

@article{garbrecht2023pGPSR,
  author    = {Garbrecht, K. and Birky, B. and Lester, B. T. and Emery, J. M. and Hochhalter, J. D.},
  title     = {Complementing a continuum thermodynamic approach to constitutive modeling with symbolic regression},
  journal   = {Journal of the Mechanics and Physics of Solids},
  volume    = {173},
  pages     = {105198},
  year      = {2023},
  doi       = {10.1016/j.jmps.2023.105198}
}

@article{germain1983continuum,
  author  = {Germain, P. and Nguyen, Q. S. and Suquet, P.},
  title   = {Continuum Thermodynamics},
  journal = {Journal of Applied Mechanics},
  volume  = {50},
  number  = {4b},
  pages   = {1010--1020},
  year    = {1983},
  doi     = {10.1115/1.3167184}
}

@article{halphen1975materiaux,
  author  = {Halphen, B. and Nguyen, Q. S.},
  title   = {Sur les mat{\'e}riaux standards g{\'e}n{\'e}ralis{\'e}s},
  journal = {Journal de M{\'e}canique},
  volume  = {14},
  number  = {1},
  pages   = {39--63},
  year    = {1975}
}

@article{hou2024automated,
  author  = {Hou, J. and Chen, X. and Wu, T. and Kuhl, E. and Wang, X.},
  title   = {Automated data-driven discovery of material models based on symbolic regression: A case study on the human brain cortex},
  journal = {Acta Biomaterialia},
  volume  = {188},
  pages   = {276--296},
  year    = {2024},
  doi     = {10.1016/j.actbio.2024.09.005}
}

@article{kabliman2021application,
  author  = {Kabliman, E. and Kolody, A. H. and Kronsteiner, J. and Kommenda, M. and Kronberger, G.},
  title   = {Application of symbolic regression for constitutive modeling of plastic deformation},
  journal = {Applications in Engineering Science},
  volume  = {6},
  pages   = {100052},
  year    = {2021},
  doi     = {10.1016/j.apples.2021.100052}
}

@article{kissas2024language,
  author  = {Kissas, G. and Mishra, S. and Chatzi, E. and De~Lorenzis, L.},
  title   = {The language of hyperelastic materials},
  journal = {Computer Methods in Applied Mechanics and Engineering},
  volume  = {428},
  pages   = {117053},
  year    = {2024},
  doi     = {10.1016/j.cma.2024.117053}
}

@book{koza1992genetic,
  author    = {Koza, J. R.},
  title     = {Genetic Programming: On the Programming of Computers by Means of Natural Selection},
  publisher = {MIT Press},
  year      = {1992}
}

@article{lacava2021contemporary,
  author  = {La~Cava, W. and Burlacu, B. and Virgolin, M. and others},
  title   = {Contemporary symbolic regression methods and their relative performance},
  journal = {Advances in Neural Information Processing Systems},
  year    = {2021}
}

@book{lemaitre2000mechanics,
  author    = {Lemaitre, J. and Chaboche, J.-L.},
  title     = {Mechanics of Solid Materials},
  publisher = {Cambridge University Press},
  year      = {2000}
}

@article{liu2024kan,
  author  = {Liu, Z. and Wang, Y. and Vaidya, S. and others},
  title   = {{KAN}: {Kolmogorov--Arnold Networks}},
  journal = {arXiv preprint arXiv:2404.19756},
  year    = {2024}
}

@article{makke2024review,
  author    = {Makke, N. and Chawla, S.},
  title     = {Interpretable scientific discovery with symbolic regression: a review},
  journal   = {Artificial Intelligence Review},
  volume    = {57},
  number    = {1},
  pages     = {2},
  year      = {2024},
  doi       = {10.1007/s10462-023-10622-0}
}

@article{ostwald1925ueber,
  author  = {Ostwald, W.},
  title   = {{\"U}ber die {G}eschwindigkeitsfunktion der {V}iskosit{\"a}t disperser {S}ysteme. {I}},
  journal = {Kolloid-Zeitschrift},
  volume  = {36},
  number  = {2},
  pages   = {99--117},
  year    = {1925}
}

@article{payne1962dynamic,
  author  = {Payne, A. R.},
  title   = {The dynamic properties of carbon black-loaded natural rubber vulcanizates. {Part I}},
  journal = {Journal of Applied Polymer Science},
  volume  = {6},
  number  = {19},
  pages   = {57--63},
  year    = {1962}
}

@article{payne1963dynamic,
  author  = {Payne, A. R.},
  title   = {The dynamic properties of carbon black-loaded natural rubber vulcanizates. {Part II}},
  journal = {Journal of Applied Polymer Science},
  volume  = {7},
  number  = {3},
  pages   = {873--885},
  year    = {1963}
}

@book{rockafellar1970convex,
  title     = {Convex Analysis},
  author    = {Rockafellar, R. Tyrrell},
  year      = {1970},
  publisher = {Princeton University Press},
  address   = {Princeton, NJ}
}

@article{schmidt2009distilling,
  author  = {Schmidt, M. and Lipson, H.},
  title   = {Distilling free-form natural laws from experimental data},
  journal = {Science},
  volume  = {324},
  number  = {5923},
  pages   = {81--85},
  year    = {2009}
}

@article{thakolkaran2025kan,
  author  = {Thakolkaran, P. and Guo, Y. and Saini, S. and Peirlinck, M. and Alheit, B. and Kumar, S.},
  title   = {Can {KAN} {CAN}s? {Input-convex Kolmogorov--Arnold Networks} ({KAN}s) as hyperelastic constitutive artificial neural networks ({CAN}s)},
  journal = {Computer Methods in Applied Mechanics and Engineering},
  year    = {2025}
}

@article{uhlenbeck1930theory,
  author  = {Uhlenbeck, G. E. and Ornstein, L. S.},
  title   = {On the theory of the {B}rownian motion},
  journal = {Physical Review},
  volume  = {36},
  number  = {5},
  pages   = {823--841},
  year    = {1930}
}

@misc{wirgin2004inverse,
  author = {Wirgin, A.},
  title  = {The inverse crime},
  note   = {arXiv preprint math-ph/0401050},
  year   = {2004}
}

@article{zhu1997algorithm,
  author  = {Zhu, C. and Byrd, R. H. and Lu, P. and Nocedal, J.},
  title   = {Algorithm 778: {L-BFGS-B}: {Fortran} subroutines for large-scale bound-constrained optimization},
  journal = {ACM Transactions on Mathematical Software},
  volume  = {23},
  number  = {4},
  pages   = {550--560},
  year    = {1997}
}

@article{zhu2025physics,
  author  = {Zhu, Y. and Su, C. and Li, X.},
  title   = {A physics-guided symbolic regression framework for efficient and interpretable sand constitutive modeling},
  journal = {Canadian Geotechnical Journal},
  volume  = {62},
  pages   = {1--14},
  year    = {2025}
}

@article{perzyna1966fundamental,
  author  = {Perzyna, Piotr},
  title   = {Fundamental problems in viscoplasticity},
  journal = {Advances in Applied Mechanics},
  volume  = {9},
  pages   = {243--377},
  year    = {1966}
}

@article{rosenkranz2024viscoelasticity,
  author  = {Rosenkranz, Max and Kalina, Karl A. and Brummund, J{\"o}rg and Sun, WaiChing and K{\"a}stner, Markus},
  title   = {Viscoelasticty with physics-augmented neural networks: model formulation and training methods without prescribed internal variables},
  journal = {Computational Mechanics},
  volume  = {74},
  number  = {6},
  pages   = {1279--1301},
  year    = {2024},
  doi     = {10.1007/s00466-024-02477-1}
}

@inproceedings{asad2023mechanics,
  author    = {As'ad, Faisal and Farhat, Charbel},
  title     = {A Mechanics-Informed Neural Network Framework for Data-Driven Nonlinear Viscoelasticity},
  booktitle = {AIAA SCITECH 2023 Forum},
  pages     = {0949},
  year      = {2023},
  doi       = {10.2514/6.2023-0949}
}

@article{holthusen2026inelastic,
  author  = {Holthusen, Hagen and Lamm, Lukas and Brepols, Tim and Reese, Stefanie and Kuhl, Ellen},
  title   = {Theory and implementation of inelastic Constitutive Artificial Neural Networks},
  journal = {Computer Methods in Applied Mechanics and Engineering},
  volume  = {428},
  pages   = {117063},
  year    = {2024},
  doi     = {10.1016/j.cma.2024.117063}
}

@article{tac2023benchmarking,
  author  = {Ta{\c{c}}, Vahidullah and Rausch, Manuel K. and Sahli Costabal, Francisco and Tepole, Adrian Buganza},
  title   = {Data-driven anisotropic finite viscoelasticity using neural ordinary differential equations},
  journal = {Computer Methods in Applied Mechanics and Engineering},
  volume  = {411},
  pages   = {116046},
  year    = {2023},
  doi     = {10.1016/j.cma.2023.116046}
}

@article{flaschel2025convex,
  author  = {Flaschel, Moritz and Steinmann, Paul and De Lorenzis, Laura and Kuhl, Ellen},
  title   = {Convex neural networks learn generalized standard material models},
  journal = {Journal of the Mechanics and Physics of Solids},
  volume  = {200},
  pages   = {106103},
  year    = {2025},
  doi     = {10.1016/j.jmps.2025.106103}
}

@article{boes2026accounting,
  author  = {Boes, B. and Simon, J.-W. and Holthusen, H.},
  title   = {Accounting for plasticity: An extension of inelastic constitutive artificial neural networks},
  journal = {European Journal of Mechanics - A/Solids},
  volume  = {117},
  pages   = {105998},
  year    = {2026},
  doi     = {10.1016/j.euromechsol.2025.105998}
}

@article{boes2026plasticity,
  author  = {van der Velden, T. and Boes, B. and Brepols, T. and Kuhl, E. and Holthusen, H.},
  title   = {A note on constitutive artificial neural networks for finite strain plasticity},
  journal = {Mechanics Research Communications},
  volume  = {155},
  pages   = {104707},
  year    = {2026},
  doi     = {10.1016/j.mechrescom.2026.104707}
}

@article{jadoon2025plasticity,
  author  = {Jadoon, Asghar A. and Meyer, Knut Andreas and Fuhg, Jan N.},
  title   = {Automated model discovery of finite strain elastoplasticity from uniaxial experiments},
  journal = {Computer Methods in Applied Mechanics and Engineering},
  volume  = {435},
  pages   = {117653},
  year    = {2025},
  doi     = {10.1016/j.cma.2024.117653}
}

@article{abdolazizi2024vcanns,
  author  = {Abdolazizi, Kian P. and Linka, Kevin and Cyron, Christian J.},
  title   = {Viscoelastic constitutive artificial neural networks (vCANNs) -- a framework for data-driven anisotropic nonlinear finite viscoelasticity},
  journal = {Journal of Computational Physics},
  volume  = {499},
  pages   = {112704},
  year    = {2024}
}

@article{ji2026ickan,
  author  = {Ji, Chenyi and Abdolazizi, Kian P. and Holthusen, Hagen and Cyron, Christian J. and Linka, Kevin},
  title   = {Inelastic Constitutive Kolmogorov--Arnold Networks: a generalized framework for automated discovery of interpretable inelastic material models},
  journal = {arXiv preprint arXiv:2602.17750},
  year    = {2026}
}

@incollection{holzapfel2001biomechanics,
  author    = {Holzapfel, Gerhard A.},
  title     = {Biomechanics of Soft Tissue},
  booktitle = {Handbook of Materials Behavior Models},
  publisher = {Academic Press},
  pages     = {1057--1071},
  year      = {2001},
  doi       = {10.1016/B978-012443341-0/50107-1}
}

@article{batra1990effect,
  author  = {Batra, R. C. and Kim, C. H.},
  title   = {Effect of viscoplastic flow rules on the initiation and growth of shear bands at high strain rates},
  journal = {Journal of the Mechanics and Physics of Solids},
  volume  = {38},
  number  = {6},
  pages   = {859--874},
  year    = {1990}
}

@article{kabliman2026identification,
  author  = {Kabliman, Evgeniya and Kronberger, Gabriel},
  title   = {Identification of empirical constitutive models for age-hardenable aluminium alloy and high-chromium martensitic steel using symbolic regression},
  journal = {Philosophical Transactions of the Royal Society A: Mathematical, Physical and Engineering Sciences},
  volume  = {384},
  number  = {2317},
  year    = {2026}
}

@article{cranmer2023pysr,
  title   = {Interpretable machine learning for science with {PySR} and {SymbolicRegression.jl}},
  author  = {Cranmer, Miles},
  journal = {arXiv preprint arXiv:2305.01582},
  year    = {2023}
}

@article{nelder1965simplex,
  author  = {Nelder, John A. and Mead, Roger},
  title   = {A simplex method for function minimization},
  journal = {The Computer Journal},
  volume  = {7},
  number  = {4},
  pages   = {308--313},
  year    = {1965}
}

@article{gao2012implementing,
  author  = {Gao, Fuchang and Han, Lixing},
  title   = {Implementing the Nelder--Mead simplex algorithm with adaptive parameters},
  journal = {Computational Optimization and Applications},
  volume  = {51},
  number  = {1},
  pages   = {259--277},
  year    = {2012}
}

@inproceedings{califano2026enforcing,
  author    = {Califano, F. and Ciambella, J.},
  title     = {Enforcing Physics in Hyperelasticity Modeling Using {Kolmogorov--Arnold} Networks},
  booktitle = {Proceedings of the XXVI {AIMETA} Conference 2024},
  pages     = {920--928},
  publisher = {Springer, Cham},
  year      = {2026}
}

@article{kalina2026physics,
  author  = {Kalina, Karl A. and Brummund, J{\"o}rg and K{\"a}stner, Markus},
  title   = {A physics-augmented neural network framework for finite strain incompressible viscoelasticity},
  journal = {Computer Methods in Applied Mechanics and Engineering},
  volume  = {455},
  pages   = {118892},
  year    = {2026},
  doi     = {10.1016/j.cma.2026.118892}
}

@article{kommenda2020parameter,
  author  = {Kommenda, M. and Burlacu, B. and Kronberger, G. and Affenzeller, M.},
  title   = {Parameter identification for symbolic regression using nonlinear least squares},
  journal = {Genetic Programming and Evolvable Machines},
  volume  = {21},
  pages   = {471--501},
  year    = {2020},
  doi     = {10.1007/s10710-019-09371-3}
}

@book{engl1996regularization,
  author    = {Engl, Heinz Werner and Hanke, Martin and Neubauer, Andreas},
  title     = {Regularization of Inverse Problems},
  volume    = {375},
  publisher = {Springer Science \& Business Media},
  year      = {1996}
}

@article{metropolis1949monte,
  author  = {Metropolis, N. and Ulam, S.},
  title   = {The {M}onte {C}arlo method},
  journal = {Journal of the American Statistical Association},
  volume  = {44},
  number  = {247},
  pages   = {335--341},
  year    = {1949}
}

\end{document}